\documentclass[aoas]{imsart}

\RequirePackage{amsthm,amsmath,amsfonts,amssymb,centernot,float,import,makeidx,subfiles}
\RequirePackage{natbib}
\RequirePackage{graphicx}
\usepackage{xr-hyper}
\usepackage{lscape}
\usepackage{hyperref}
\usepackage{longtable}
\usepackage{threeparttable}
\usepackage{multirow}
\usepackage{breqn}
\usepackage{array}
\usepackage{booktabs}
\usepackage{caption}
\usepackage{dsfont}
\usepackage[nokeyprefix]{refstyle}
\usepackage{varioref}
\usepackage[english]{babel}
\usepackage{tikz}
\usepackage{float}
\usetikzlibrary{positioning, calc, shapes.geometric, shapes, shapes.multipart, arrows.meta, arrows, decorations.markings, external, trees}
\newcommand{\mycaption}[2]{\caption[#1]{\textbar\, \textbf{#1.} #2}}

\newlength\myindent
\startlocaldefs

\tikzstyle{Arrow} = [
	thick, 
	decoration={
		markings,
		mark=at position 1 with {
			\arrow[thick]{latex}
			}
		}, 
	shorten >= 3pt, preaction = {decorate}
	]

\theoremstyle{plain}

\newtheorem{theorem}{Theorem}

\newtheorem{proposition}{Proposition}

\newtheorem{algorithm}{Algorithm}
\newtheorem{definition}{Definition}
\usepackage{enumitem}
 
\theoremstyle{remark}
\newtheorem{remark}{Remark}
\newtheorem{corollary}{Corollary}
\newtheorem{assumption}{Assumption}

\endlocaldefs

\makeatletter
\newcommand*{\addFileDependency}[1]{
  \typeout{(#1)}
  \@addtofilelist{#1}
  \IfFileExists{#1}{}{\typeout{No file #1.}}
}
\makeatother

\begin{document}

\begin{frontmatter}
\title{Heterogeneous interventional indirect effects with multiple mediators: non-parametric and semi-parametric approaches}
\runtitle{Heterogeneous mediated effects}

\begin{aug}
\author[A]{\fnms{Max} \snm{Rubinstein}\ead[label = e1]{mrubinstein@rand.org}},
\author[B]{\fnms{Zach} \snm{Branson}\ead[label = e2,mark]{zach@stat.cmu.edu}}, \and
\author[B]{\fnms{Edward H.} \snm{Kennedy}\ead[label = e3,mark]{edward@stat.cmu.edu}}
\address[A]{RAND Corporation 
\printead{e1}}
\address[B]{Department of Statistics \& Data Science, Carnegie Mellon University,
\printead{e2,e3}}
\end{aug}

\begin{flushleft}
We propose semi- and non-parametric methods to estimate conditional interventional indirect effects in the setting of two discrete mediators whose causal ordering is unknown. Average interventional indirect effects have been shown to decompose an average treatment effect into a direct effect and interventional indirect effects that quantify effects of hypothetical interventions on mediator distributions. Yet these effects may be heterogeneous across the covariate distribution. We therefore consider the problem of estimating these effects at particular points. We first propose an influence-function based estimator of the \textit{projection} of the conditional effects onto a working model, and show that under some conditions we can achieve root-n consistent and asymptotically normal estimates of this parameter. Second, we propose a fully non-parametric approach to estimation and show the conditions where this approach can achieve oracle rates of convergence. Finally, we propose a sensitivity analysis for the conditional effects in the presence of mediator-outcome confounding given a bounded outcome. We propose estimating bounds on the conditional effects using these same methods, and show that these results easily extend to allow for influence-function based estimates of the bounds on the average effects. We conclude by demonstrating our methods to examine heterogeneous mediated effects with respect to the effect of COVID-19 vaccinations on depression via social isolation and worries about health during February 2021.
\end{flushleft}

\end{frontmatter}


\section{Introduction}

A goal of causal mediation analysis is to understand the mechanisms through which interventions work. ``Natural effects'' most directly pertain to the idea of mechanism \citep{miles2022causal} and decompose the individual-level treatment effect into pathways that work directly or via changes in mediator values. However, the identifying assumptions required to estimate these effects are unenforceable even in randomized experiments. These effects are also not generally identified in common applied settings that involve multiple mediators unless the mediators are considered jointly. ``Interventional effects'' were proposed as alternative causal estimands that are identifiable under weaker assumptions and in settings with multiple mediators (see, e.g., \cite{didelez2006proceedings}, \cite{vansteelandt2017interventional}). These effects conceptualize hypothetical interventions on the mediator distributions defined at specific covariate values. Unlike natural effects, these effects are identifiable in a sequentially randomized experiment. While the relationship between interventional effects and mechanisms acting at an individual-level is unclear \citep{miles2022causal}, interventional effects have gained popularity in applied research over the past several years. 

This popularity is also in part because the same statistical functionals yield alternative causal interpretations that are often of substantive interest. Specifically, under weaker assumptions these same methods can quantify disparity reductions achieved via interventions on some possibly mediating factor(s). For example, Vansteelandt and Daniel (2017) analyze disparities in breast-cancer survival among high and low socioeconomic status (SES) women. They consider a model where SES causes breast-cancer survival via a direct pathway and via cancer screening and treatment choices, so that SES takes the role of an exposure. They estimate that if low SES women had the same (conditional) distribution of cancer screening and treatment choices as high SES women, the observed disparity in breast-cancer survival between high and low SES women would be reduced by half. This effect requires conceiving of an unspecified hypothetical intervention that could shift the covariate-stratrum specific distributions of cancer screening and treatment choices among low SES women to match that of high SES women. However, this intervention does not require conceiving of potential outcomes with respect to SES, or more generally of the types of controversial counterfactual quantities required to conceive of natural effects \citep{vansteelandt2017interventional}.

To date the literature on interventional effects has primarily focused on estimating \textit{average} effects. We instead consider estimating \textit{conditional} effects across covariates that are possibly continuous. For example, consider the application from \cite{vansteelandt2017interventional}. One natural follow-up question might be how these disparity reductions change as a function of a woman's age. Even if the total disparity in cancer survival were constant across age, it remains possible that age moderates the interventional effects, and therefore also the proportion of the disparity that would be eliminated via such an intervention. These questions pertain to \textit{conditional} interventional indirect effects (CIIE). To our knowledge, proposed strategies to estimate the CIIE have been limited to parametric methods (see, e.g., \cite{vansteelandt2017interventional}, \cite{loh2020heterogeneous}), and the validity of the inferences are tied to assumptions that the models are correctly specified.

Our first contribution is to propose methods that allow for flexible non-parametric and machine learning methods for estimation. Specifically, we consider the setting of a binary intervention and two discrete-valued mediators whose causal ordering is unknown. We propose two estimation procedures: first, a semi-parametric projection-based approach; second, a fully non-parametric approach. Both procedures are conceptually simple, and involve a regression of an estimate of the uncentered influence function for the average effect onto the covariates. However, the semi-parametric approach targets a \textit{projection} of the CIIE onto a parametric model rather than the CIIE itself. Projection-based estimators have frequently been proposed in the context of different causal estimands (see, e.g., \cite{kennedy2019robust}, \cite{cuellar2020non}, \cite{kennedy2021semiparametric}). Our proposal extends this idea to this setting, and we show that under some conditions, root-n consistent and asymptotically normal estimates of the projection parameter are possible. The second proposal considers a fully non-parametric estimation procedure (a ``DR-Learner'') that targets the CIIE directly, extending results from \cite{kennedy2022}. While directly targeting the CIIE instead of its projection may seem preferable, we cannot in general obtain obtain root-n consistent estimates. Even so, we show that we can obtain oracle rates of convergence in some settings. Both the projection estimator and the DR-Learner substantially weaken the modeling assumptions used to date in the literature on estimating the CIIE. Moreover, these methods allow use of flexible non-parametric and machine-learning methods for estimation while still obtaining relatively fast rates of convergence.

Our methods, like most, require several identifying assumptions, including that the mediator-outcome relationship is unconfounded. A natural question is whether our estimators are sensitive to violations of this assumption. We therefore consider bounds on the CIIE and show that we can use both the projection-based approach and the DR-Learner to estimate these quantities. These methods naturally extend to allow for estimating bounds on the average effects, and we show that root-n consistent and asymptotically normal estimates of these bounds are possible under some conditions. Existing approaches frequently focus on the natural rather than interventional effects, and are often tied to strong parametric modeling assumptions (see, e.g., \cite{park2020estimation}, \cite{park2021sensitivity}, \cite{imai2010general}). Moreover, sensitivity analyses for conditional estimands is less seldom discussed (though see \cite{lindmark2018sensitivity}). 

Finally, we demonstrate these methods using an application previously considered in \cite{rubinstein2023}. This study sought to quantify the extent to which COVID-19 vaccines reduced self-reported depression via changes in social isolation versus worries about health among the COVID-19 Trends and Impact Survey (CTIS) respondents in February 2021. The authors only examined effect heterogeneity across discrete subgroups; moreover, they did not conduct a sensitivity analysis with respect to the interventional effect estimates. We revisit this analysis and model how the vote share for Joe Biden in the 2020 US presidential election in each respondent's county of residence moderated the interventional effects. We then demonstrate our sensitivity analysis for the average and the conditional interventional effects.

This paper proceeds as follows. In Section~\ref{sec:1} we review interventional effects, the required identifying assumptions, and efficient estimation. In Section~\ref{sec:2} we introduce the CIIE and propose the projection-based estimator and the DR-Learner. We establish conditions required for asymptotic normality and root-n consistency of the projection estimator, and for obtaining oracle rates of convergence for the DR-Learner. Section~\ref{sec:3} contains a simulation study demonstrating that these theoretic properties hold in practice. Section~\ref{sec:4} proposes our sensitivity analysis, Section~\ref{sec:5} contains our application, and Section~\ref{sec:6} contains a discussion of these results.

\section{Review}\label{sec:1}

We define the average interventional effects, the causal assumptions required to tie the causal targets to observed data, and efficient estimation of the observed data functionals. We largely summarize material covered in \cite{vansteelandt2017interventional} and \cite{benkeser2021nonparametric}, and refer to those papers for more details. We begin by outlining the setup and notation that we will use throughout.

\subsection{Setup and notation}

Assume that we observe $n$ i.i.d. samples of observations $Z_i = (V_i, W_i, M_{1i}, M_{2i}, A_i, Y_i)$, where $Y$ represents some outcome of interest (either continuous or discrete), $A$ represents a binary intervention, and $M_1$ and $M_2$ represent discrete-valued mediators. Finally, we let $[V, W]$ represent a matrix of either discrete or continuous covariates, where $W$ may be empty, and which we jointly denote as $X$.\footnote{We distinguish between $V$ and $W$ because we will estimate effects conditional on $V = v$, which may or may not include all elements of $X$.} The figure below illustrates the assumed relationships between the variables, with an arrow indicating a causal pathway. Importantly, we do not assume that we know the causal relationship between $M_1$ and $M_2$.

\begin{figure}[H]
\caption{Assumed data generating process}
\begin{center}
\begin{tikzpicture}
[
array/.style={rectangle split, 
	rectangle split parts = 3, 
	rectangle split horizontal, 
    minimum height = 2em
    }
]
 \node (3) {X = [V, W]};
 \node [right =of 3] (4) {A};
 \node [right =of 4] (5) {$M_1$};
 \node [below =of 5] (6) {$M_2$};
 \node [right =of 5] (7) {Y};
 \draw[Arrow] (3.east) -- (4.west);
 \draw[Arrow] (4.east) -- (5.west);
 \draw[Arrow] (6.north) -- (7.south);
 \draw[Arrow] (3) to [out = 25, in = 160] (7);
 \draw[Arrow] (3) to [out = 25, in = 160] (5);
 \draw[Arrow] (5) to (6);
 \draw[Arrow] (6) to (5);
 \draw[Arrow] (5) to (7);
 \draw[Arrow] (3) to (6);
 \draw[Arrow]  (4) to [out = 25, in = 160] (7);
 
 \draw[Arrow] (4) to (6);
\end{tikzpicture}
\end{center}
\end{figure}
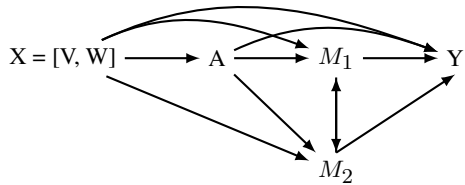

While this figure helps to motivate the problem, we will primarily rely on potential outcomes notation to define our assumptions. Specifically, we assume the potential outcomes $Y^{am_1m_2}$ under $A = a$, $M_1 = m_1$, and $M_2 = m_2$, and that $M_j^a$ represent the counterfactual outcome for the $j$-th mediator under $A = a$. For any discrete variable $C$ we use $p(c)$ as a short-hand for $P(C = c)$ throughout. 

We also define the following functions of the data. Let $\pi_a(X) = p(A = a \mid X)$ and $\mu_a(M_1, M_2, X) = \mathbb{E}[Y \mid A = a, X, M_1, M_2]$. We denote the joint mediator probabilities as $p(M_1, M_2 \mid X, A)$, and the marginal probabilities as $p(M_1 \mid X, A)$ and $p(M_2 \mid X, A)$, respectively. Following the notation in \cite{benkeser2021nonparametric}, we define the marginalized outcome models:

\begin{align}
    \nonumber&\mu_{a, M_1}(M_2, X) = \sum_{m_1} \mu_a(m_1, M_2, X)p(m_1 \mid a, X) \\
    \nonumber&\mu_{a, M_1\times M_2^'}(X) = \sum_{m_1, m_2} \mu_a(m_1, m_2, X)p(m_1 \mid a, X)p(m_2 \mid a', X) 
\end{align}
These quantities are defined with respect to marginalizing the outcome regression over $p(M_1 \mid a, X)$. We use this same notation to define quantities marginalized with respect to $p(M_1 \mid a', X)$ (e.g. $\mu_{a, M_1'\times M_2^'}(M_1, X)$).

We also denote sample averages using $\mathbb{P}_n(Z)$ and regression estimates as $\hat{\mathbb{E}}_n(Y \mid X)$. For possibly random functions $f$, we denote $\|f\|^2$ as the squared $L_2(\mathbb{P})$ norm, $\int f(z)^2 dP(z)$. We use the notation $a \lesssim b$ to indicate that $a \le C b$ for some universal constant $C$. Finally, for any random function $\hat{f}$ learned on an independent sample of size $n$, $D_1^n$, we let $\mathbb{P}\hat{f}(Z) = \int \hat{f}(z) dP(z \mid D_1^n)$. That is, $\mathbb{P}\hat{f}(Z)$ refers to the expected value of this estimated function conditional on the training sample, noting that for a fixed function $f$ this is equivalent to $\mathbb{E}(f)$.

\subsection{Interventional effects}

We can decompose the average effect $\psi = \mathbb{E}[Y^a - Y^{a'}]$ into the sum of the following quantities \citep{vansteelandt2017interventional}:

\begin{align} 
\psi_{IDE} &= \mathbb{E}\left(\sum_{m_1, m_2}\mathbb{E}[Y^{am_1m_2} - Y^{a'm_1m_2} \mid X]p(M_1^{a'} = m_1, M_2^{a'} = m_2 \mid X)\right) \\
\psi_{M_1} &= \mathbb{E}\left(\sum_{m_1, m_2}\mathbb{E}[Y^{am_1m_2} \mid X]\{p(M_1^a = m_1 \mid X) - p(M_1^{a'} = m_1 \mid X)\}p(M_2^{a'} = m_2 \mid X)\right) \\
\psi_{M_2} &= \mathbb{E}\left(\sum_{m_1, m_2}\mathbb{E}[Y^{am_1m_2} \mid X]\{p(M_2^a = m_2 \mid X) - p(M_2^{a'} = m_2 \mid X)\}p(M_1^a = m_1 \mid X)\right) \\
\psi_{Cov} &= \mathbb{E}\left(\sum_{m_1, m_2}\mathbb{E}[Y^{am_1m_2} \mid X]\{p(M_1^a = m_1, M_2^a = m_2 \mid X) - p(M_1^{a} = m_1 \mid X)p(M_2^a = m_2 \mid X) \right.\\ 
\nonumber&\left.- [p(M_1^{a'} = m_1, M_2^{a'} = m_2 \mid X) - p(M_1^{a'} = m_1 \mid X)p(M_2^{a'} = m_2 \mid X)]\}\right)
\end{align}

The interventional direct effect $(\psi_{IDE})$ represents the contrast in mean potential outcomes when we set $A = a$ for everyone in the population versus $A = a'$, while drawing the mediators randomly for each individual from their counterfactual joint distribution under $A = a'$ given subject-specific covariates $X$. By contrast the interventional indirect effect via $M_1$ $(\psi_{M_1})$ holds $A$ fixed at $a$ for all individuals, and considers the average contrast between giving everyone subject-specific values of $M_1$ drawn randomly from the counterfactual distribution under $A = a$ versus the distribution under $A = a'$ given covariates $X$, and simultaneously drawing $M_2$ from the counterfactual distribution under $A = a'$ given covariates $X$. The interventional indirect effect through $M_2$ $\left(\psi_{M_2}\right)$ is defined analogously. The covariant effect $(\psi_{Cov})$ is the difference between the total effect and all three of these effects and captures the effect of the dependence of the mediators on each other. This decomposition also holds switching the role of $a$ and $a'$: \cite{vansteelandt2017interventional} indexes these estimands by $a$ to make this distinction, while for conceptual simplicity we do not. Additionally, for the purposes of this paper we focus on effects via $M_1$; however, our proposed methods can be used for any of these estimands.

\subsection{Identification}\label{sec:identification}

The estimands defined above require knowledge about the potential outcomes under each treatment and mediator value for each subject. However, for any individual we do not observe all of these quantities. We therefore make the following identifying assumptions to connect these causal quantities to the observed data distribution. First, we assume consistency, where for $a^\star \in \{a, a'\}$ and any $(m_1, m_2)$:

\begin{assumption}\label{assumption:consistency}[Consistency]
\begin{align*}
    & A = a^\star \implies (M_1, M_2) = (M_1^{a^\star}, M_2^{a^\star}) \\
    & A = a^\star, M_1 = m_1, M_2 = m_2 \implies Y = Y^{a^\star m_1m_2} 
\end{align*}
\end{assumption}

Consistency precludes the potential outcomes for any individual from depending on another individual's treatment or mediator assignment. We next assume sequential ignorability:

\begin{assumption}\label{assumption:seqig}[Sequential ignorability]
\begin{align}
    &\label{eqn:unconf1}Y^{a^\star m_1m_2} \perp A \mid X \\
    &\label{eqn:unconf2}(M_1^{a^\star}, M_2^{a^\star}) \perp A \mid X \\
    &\label{eqn:unconf3}Y^{a^\star m_1m_2} \perp (M_1, M_2) \mid (A = a^\star, X) 
\end{align}
\end{assumption}

Sequential ignorability consists of three assumptions: equations (\ref{eqn:unconf1})-(\ref{eqn:unconf2}), or Y-A and M-A ignorability, state that $A$ is independent of $Y^{am_1m_2}$ and $(M_1^a, M_2^a)$ given $X$. Equation (\ref{eqn:unconf3}), or Y-M ignorability, states that $(M_1, M_2)$ is independent of $Y^{am_1m_2}$ given $X$ and $A$. We next assume positivity:

\begin{assumption}\label{assumption:positivity}[Positivity]
\begin{align}
    &\label{eqn:apositivity}P(\min_{a^\star} \pi_{a^\star}(X) > \epsilon) = 1, \qquad P(\min_{m_1, m_2, a^\star}p(m_1, m_2 \mid a^\star, X) > \epsilon) = 1, \qquad \epsilon > 0 
\end{align}
\end{assumption}
Equation (\ref{eqn:apositivity}) implies that the propensity scores are bounded away from zero and one and the joint mediator probabilities are bounded away from zero with probability one.\footnote{Equation \ref{eqn:apositivity} is technically stronger than necessary. For example, for any $x$, we only need that $p(m_1, m_2 \mid a, x) > 0$ whenever $p(m_1, m_2 \mid a', x) > 0$.} Under these assumptions we can write $\psi_{M_1}$ in terms of the observed data distribution.

\begin{align}
    \psi_{M_1} &= \mathbb{E}\left(\sum_{m_1, m_2}\mu_a(X, m_1, m_2)\{p(m_1 \mid X, a) - p(m_1 \mid X, a')\}p(m_2 \mid X, a')\right) = \mathbb{E}\left(\psi_{M_1}(X)\right)\label{eqn:parameter1}
\end{align}

The functional in \eqref{eqn:parameter1} reflects other interesting causal parameters under weaker assumptions. For example, consider the case where \eqref{eqn:unconf1} holds but \eqref{eqn:unconf2} does not. This situation is frequently relevant in cases where we are using interventional indirect effects to understand disparities and $A$ is an indicator of some subgroup of interest (for example, Black versus White individuals). In that case \eqref{eqn:parameter1} still targets the causal contrast:

\begin{align}\label{eqn:disparityparam}
    \mathbb{E}\left(\sum_{m_1, m_2}\mathbb{E}[Y^{m_1m_2} \mid a, X][p(m_1 \mid a, X) - p(m_1 \mid a', X)]p(m_2 \mid a', X)\right)
\end{align}

This estimand tells us about how much an intervention on the distribution of $M_1$ could reduce an observed disparity in some outcome of interest \citep{vansteelandt2017interventional}. However, in practice we must carefully consider the relevant conditioning sets when defining these quantities \citep{jackson2020meaningful}.

\subsection{Estimation}\label{ssec:review-estimation}

Regardless of the targeted causal quantity, \eqref{eqn:parameter1} reflects a statistical parameter that is a fixed function of the observed data and we require methods to estimate this quantity. One natural idea would be to estimate each function in \eqref{eqn:parameter1} separately, plug them into that same expression, and take the empirical average. If we were to use correctly specified parametric models to estimate the nuisance functions the resulting estimate would be consistent for $\psi_{M_1}$ and converge at $n^{-1/2}$ rates. Unfortunately, this is unlikely to occur in practice. We could instead estimate these functions flexibly using non-parametric models; however, the subsequent estimator will in general inherit the non-parametric rate of convergence of the slowest estimated nuisance function, a rate generally slower than $n^{-1/2}$ \citep{kennedy2022eifs}. A different estimation strategy instead utilizes the so-called ``influence curve'' of \eqref{eqn:parameter1}. Influence curves are important quantities related to statistical functionals that naturally suggest efficient estimators without parametric modeling assumptions \citep{kennedy2022eifs}.

For example, \cite{benkeser2021nonparametric} previously showed that $\psi_{M_1}$ has the efficient influence curve $\varphi(Z; \eta) - \psi_{M_1}$, where the uncentered influence curve $\varphi(Z; \eta)$ takes the form

\begin{align}\label{eqn:eifm1}
    \varphi(Z; \eta) &= \frac{\mathds{1}(A = a)}{\pi_a(X)}\frac{\{p(M_1 \mid a, X) - p(M_1 \mid a', X)\}p(M_2 \mid a', X)}{p(M_1, M_2, \mid a, X)}(Y - \mu_a(M_1, M_2, X)) \\
    \nonumber &+ \frac{\mathds{1}(A = a)}{\pi_a(X)}\{\mu_{a, M_2^'}(M_1, X) - \mu_{a, M_1\times M_2^'}(X)\} \\ 
    \nonumber &- \frac{\mathds{1}(A = a')}{\pi_{a'}(X)}\{\mu_{a, M_2^'}(M_1, X) - \mu_{a, M_1^'\times M_2^'}(X)\} \\
    \nonumber &+ \frac{\mathds{1}(A = a')}{\pi_{a'}(X)}\left(\mu_{a, M_1}(M_2, X) - \mu_{a, M_1\times M_2'}(X) - (\mu_{a, M_1^'}(M_2, X) - \mu_{a, M_1^'\times M_2^'}(X))\right) \\
    \nonumber &+ \mu_{a, M_1\times M_2^'}(X) - \mu_{a, M_1^'\times M_2^'}(X)
\end{align}
and where 

\begin{align}
\label{eqn:eta}\eta = [p(M_1, M_2 \mid a, X), p(M_1, M_2 \mid a', X), \mu_a(M_1,M_2,X), \pi_a(X)]    
\end{align}
Using $\varphi(Z)$, we can construct the so-called ``one-step'' estimator of $\psi_{M_1}$,

\begin{align}\label{eqn:ose}
\hat{\psi}^{os}_{M_1} = \mathbb{P}_n[\varphi(Z; \hat{\eta})]
\end{align}
This estimator involves estimating $\eta$, plugging these estimates into \eqref{eqn:eifm1} yielding $\varphi(Z; \hat{\eta})$, and taking the empirical mean. As shown by \cite{benkeser2021nonparametric}, the following conditions are sufficient for this estimator to yield root-n consistent and asymptotically normal estimates. 

\begin{enumerate}
    \item The estimates $\hat{\eta}$ are obtained via sample-splitting
    \item $\|\varphi(Z; \hat{\eta}) - \varphi(Z; \eta)\|^2 = o_p(1)$
    \item $\|\hat{\eta} - \eta\| = o_p(n^{-1/4})$
\end{enumerate}

Condition (1) can be enforced in the estimation procedure. Condition (2) requires that the mean-squared error of the estimated influence-function converges in probability to zero at any rate. This would require, for example, that the propensity-scores and their estimates be bounded away from zero and one, the joint mediator probability $p(m_1, m_2 \mid a, x)$ and its estimates are bounded away from zero, and that the nuisance estimates $\hat{\eta}$ are consistent at any rate for $\eta$. Finally, condition (3) holds for a variety of non-parametric estimators of $\eta$ under structural assumptions on the underlying nuisance functions: for example, on their smoothness or sparsity. 

Influence-function based estimators can therefore attain $n^{-1/2}$ convergence rates without parametric modeling assumptions and allowing for non-parametric estimates of the nuisance parameters. Intuitively, the reason is that we can essentially ignore the nuisance estimation error -- assuming that we estimate the nuisance functions well enough (condition (3)). The asymptotics then follow as if we had an oracle that gave us the fixed (but unknown) function $\varphi(Z; \eta)$ and we took the average \citep{kennedy2022eifs}. This remarkable fact occurs because the error of influence-function based estimators is a function of the \textit{product of errors} in the nuisance estimation. By contrast, standard methods are in general \textit{linear} in the nuisance estimation. As a result, influence-function based estimators can tolerate relatively slow rates of convergence for the nuisance estimates (e.g. rates achieved by non-parametric methods) while still obtaining faster rates of convergence for the estimand itself. Analogous to results shown by \cite{kennedy2022} for estimating the conditional average treatment effect (CATE), we will next show that we can also leverage $\varphi(Z; \eta)$ to achieve relatively fast rates of convergence when estimating the CIIE.

\section{Conditional effects}\label{sec:2}

In contrast to the average effects $\psi_{M_1}$, we consider estimating the effect at a given point $V = v$, recalling that $X = [W, V]$, so that $V$ is a subset of the covariates $X$. In a slight abuse of notation,\footnote{The conditioning in the inner expectation is more precisely written as $[W, V = v]$ rather than $X$.} we define this estimand as:

\begin{align}\label{eqn:condfx1}
    \psi_{M_1}(v) &= \mathbb{E}\left[\sum_{m_1, m_2}\mathbb{E}[Y^{am_1m_2} \mid X]\{p(M_1^a = m_1 \mid X) - p(M_1^{a'} = m_1 \mid X)\}p(M_2^{a'} = m_2 \mid X) \mid V = v\right] 
\end{align}
Under Assumptions~(\ref{assumption:consistency})-(\ref{assumption:positivity}), these parameters are identified in the observed data as

\begin{align}
\psi_{M_1}(v) &= \mathbb{E}\left(\sum_{m_1, m_2}\mu_a(X, m_1, m_2)\{p(m_1 \mid X, a) - p(m_1 \mid X, a')\}p(m_2 \mid X, a') \mid V = v\right) 
\end{align}
A natural question is how well we can estimate these effects. Noting that: 

\begin{align*}
\mathbb{E}[\varphi(Z; \eta) \mid V = v] = \psi_{M_1}(v)
\end{align*}
we can think of an ``oracle'' influence-function based estimator as providing a benchmark for comparison for any other CIIE estimate:

\begin{align}\label{eqn:oracle}
\hat{\mathbb{E}}_n[\varphi(Z; \eta) \mid V = v]
\end{align}
Just as the ``oracle'' estimate of $\psi_{M_1}$ would be an empirical average of the true influence function, we can think of \eqref{eqn:oracle} as a \textit{local} average of the true influence function around the point $V = v$. As long as Assumption~\ref{assumption:positivity} holds, we expect that the rate of convergence of \eqref{eqn:oracle} provides a valid, though possibly unachievable, lower bound for the rate of convergence for any estimator of the CIIE. This follows from noting that the convergence rate of this estimator is equivalent, up to constants, of replacing $\mu_a(X, m_1, m_2)$ with the potential outcomes $Y^{am_1m_2}$ in the expression for $\psi_{M_1}(X)$, and regressing this quantity onto $V$. While the remaining discussion compares our estimators against the oracle estimate, we expect that the oracle rates provide, in some settings, the best possible (minimax optimal) rates of convergence.\footnote{In fact this is true for the proposed projection estimator. We show later that this quantity is estimable at root-n rates in some settings, as with the average effect. However, this statement is solely conjecture for the proposed DR-Learner.}

At a high-level, both ideas we propose -- the projection-estimator and the DR-Learner -- substitute the estimated influence function $\varphi(Z; \hat{\eta})$ for $\varphi(Z; \eta)$ into  a  regression model $\hat{\mathbb{E}}_n$. A key difference between these approaches is that the projection-estimator uses a parametric model for the regression and targets a \textit{projection} of the CIIE, while the DR-Learner instead uses a non-parametric model and targets the CIIE itself. For either approach, we show the conditions where the corresponding oracle rates are attainable, analogous to results derived in \cite{kennedy2022}.  Importantly, the theoretic results assume that the nuisance estimates $\hat{\eta}$ are estimated on an independent sample from the regression estimate $\hat{\mathbb{E}}_n$. We therefore refer to this as a ``second-stage'' regression, reflecting the ordering of these two procedures. We discuss both procedures more below. 

\subsection{Projection estimator}

We first consider the case where the second-stage regression estimate $\hat{\mathbb{E}}_n(\cdot \mid V = v)$ is given by the \textit{parametric} model $g(v; \beta)$, where $\beta$ is a finite-dimensional parameter that minimizes some loss function. Specifically,

\begin{align}\label{eqn:projection}
    \beta = \arg\min_{\tilde{\beta}}\mathbb{E}[w(X)\ell\{\psi_{M_1}(X) - g(V; \tilde{\beta})\}]
\end{align}
Importantly, we need not assume $\psi_{M_1}(X) = g(V; \beta)$ for $g(v; \beta)$ to represent a meaningful target of inference. Under no assumptions about $\psi_{M_1}(X)$, $\beta$ represents a population parameter that characterizes a \textit{projection} of $\psi_{M_1}(X)$ onto $g(V; \beta)$. For simplicity, we focus on the case where $\beta$ minimizes the squared-error loss ($\ell(z) = z^2$). The weights $w(X)$ can vary to prioritize different parts of the covariate space when defining the projection, though in the simplest case we can take them to be uniform. This projection, while not necessarily representing the true parameter $\psi_{M_1}(v)$, can nevertheless represent a useful summary of this parameter. The interpretation of parametric models as defining projections is well-known, though perhaps underappreciated, in applied statistics, and is frequently used in more general settings when attempting to summarize characteristics of unknown data-generating processes (see, e.g., \cite{angrist2009mostly}, \cite{buja2019models}). 

To estimate this projection, we assume standard regularity conditions and differentiate \eqref{eqn:projection} with respect to $\beta$ to obtain the following moment condition:

\begin{align}
    \mathbb{E}\left[\frac{\partial g(V; \beta)}{\partial \beta} w(X)\left\{\psi_{M_1}(X) - g(V; \beta)\right\}\right]  := \Psi(\beta; \mathbb{P}) = 0
\end{align}

As with the average effects, our estimation approach is again based on the influence curve of $\Psi(\beta; \mathbb{P})$. This yields an estimating-equation stated formally in Proposition~\ref{proposition2}.

\begin{proposition}\label{proposition2}
Under a non-parametric model, the uncentered efficient influence curve for the moment condition $\Psi(\beta^\star)$ at any fixed $\beta^\star$ is given by

\begin{align}
\phi(Z; \beta^\star, \eta) &= \frac{\partial g(V; \beta^\star)}{\partial \beta}w(X)\left(\varphi(Z; \eta)  - g(V; \beta^\star)\right)
\end{align}
\end{proposition}

This then suggests the estimator $\hat{\beta}$ that satisfies:

\begin{align}
\mathbb{P}_n\left[\frac{\partial g(V; \hat{\beta})}{\partial \beta} w(X)\left(\varphi(Z; \hat{\eta})  - g(V; \hat{\beta})\right)\right] = 0
\end{align}

Theorem~\ref{theorem3} shows the conditions required to obtain root-n consistent and asymptotically normal parameter estimates. 

\begin{theorem}\label{theorem3}
Consider the moment condition $\mathbb{E}[\phi(Z; \beta_0, \eta_0)] = 0$ evaluated at the true parameters $(\beta_0, \eta_0)$. Now consider the estimator $\hat{\beta}$ that satisfies $\mathbb{P}_n[\phi(Z; \hat{\beta}, \hat{\eta})] = 0$, where $\hat{\eta}$ is estimated on an independent sample. Assume that:
\begin{itemize}
    \item The function class $\{\phi(Z; \beta, \eta): \beta \in \mathbb{R}^p\}$ is Donsker in $\beta$ for any fixed $\eta$
    \item $\|\phi(Z; \hat{\beta}, \hat{\eta}) - \phi(Z; \beta_0, \eta_0)\| = o_p(1)$
    \item The map $\beta \to \mathbb{P}[\phi(Z; \beta, \eta)]$ is differentiable at $\beta_0$ uniformly in the true $\eta$, with non-singular derivative matrix $\frac{\partial}{\partial \beta}\mathbb{P}\{\phi(Z; \beta, \eta)\}\mid_{\beta = \beta_0} = M(\beta_0, \eta)$, where $M(\beta_0, \hat{\eta}) \to^p M(\beta_0, \eta_0)$
\end{itemize}

Then

\begin{align*}
\hat{\beta} - \beta &= -M^{-1}[\mathbb{P}_n - \mathbb{P}]\phi(Z; \beta_0, \eta_0) + \mathcal{O}_p(T_{1n} + T_{2n} + T_{3n} + T_{4n}) \\
\end{align*}
where

\begin{align*}
&T_{1n} = \|\hat{\mu}_a(M_1, M_2, X) - \mu_a(M_1, M_2, X)\|\|\pi_{a}(X) - \hat{\pi}_{a}(X)\| \\
&T_{2n} = \|\hat{\mu}_a(M_1, M_2, X) - \mu_a(M_1, M_2, X)\|[\|p(M_1, M_2 \mid a, X) - \hat{p}(M_1, M_2 \mid a, X)\| \\
&+ \|p(M_1 \mid a, X) - \hat{p}(M_1 \mid a, X)\| + \|p(M_2 \mid a', X) - \hat{p}(M_2 \mid a', X)\| \\
&+ \|p(M_1 \mid a', X) - \hat{p}(M_1 \mid a', X)\|] \\
&T_{3n} = \|\pi_{a}(X) - \hat{\pi}_{a}(X)\|[\|p(M_1 \mid a, X) - \hat{p}(M_1 \mid a, x)\| \\
&+\|p(M_1 \mid a', X) - \hat{p}(M_1 \mid a', x)\| + \|p(M_2 \mid a', X) - \hat{p}(M_2 \mid a', X)\|] \\
&T_{4n} = \|p(M_2 \mid a', X) - \hat{p}(M_2 \mid a', X)\|[\|p(M_1 \mid a, X) - \hat{p}(M_1 \mid a, X)\| \\
&+ \|p(M_1 \mid a', X) - \hat{p}(M_1 \mid a', X)\|]
\end{align*}
Suppose further that $T_{1n} + T_{2n} + T_{3n} + T_{4n} = o_p(n^{-1/2})$. Then the proposed estimator attains the non-parametric efficiency bound and is asymptotically normal with

\begin{align}
    \sqrt{n}(\hat{\beta} - \beta) \to^d \mathcal{N}(0, M^{-1}\mathbb{E}[\phi\phi^\top]M^{-\top}) 
\end{align}
This also implies that for any fixed value of $V = v$:

\begin{align}\label{eqn:avarproj}
    \sqrt{n}(g(v; \hat{\beta}) - g(v; \beta)) \to^d \mathcal{N}\left(0, \left(\frac{\partial g(v; \beta)}{\partial \beta}\right)^\top M^{-1}\mathbb{E}[\phi\phi^\top]M^{-\top}\frac{\partial g(v; \beta)}{\partial \beta}\right)
\end{align}
\end{theorem}

\begin{remark}\label{remark:proj-order}
    The expression $T_{1n} + T_{2n} + T_{3n} + T_{4n}$ is expressed in terms of separate estimates for the marginal probability $p(M_1 \mid a, X)$ and the joint probability $p(M_1, M_2 \mid a, X)$. However, in practice the estimate of $p(M_1 \mid a, X)$ may come from a marginalized estimate of $p(M_1, M_2 \mid a, X)$; similarly, the estimates of $p(M_1 \mid a', X)$ and $p(M_2 \mid a', X)$ may come from marginalized estimates of $p(M_1, M_2 \mid a', X)$. In this case these estimates of the marginal probabilities will inherit the same rate of convergence as the estimate of these joint probabilities evaluated at the worst-case value of $m_2$ for $p(M_1 \mid a, X)$ (and similarly for $p(M_1 \mid a', X)$ and $p(M_2 \mid a', X)$), allowing us to simplify the second-order expressions above. Such an estimation approach is reasonable if we believe that these functions have the same underlying complexity. On the other hand, if we believe that the marginal probabilities are less complex than the joint probabilities, we may instead wish to estimate each marginal probability separately. 
\end{remark}

Theorem~\ref{theorem3} illustrates that if we can estimate the components of $\eta$ (defined in \eqref{eqn:eta}) quickly enough $\hat{\beta}$ will be root-n consistent for $\beta$ and asymptotically normal.  It also implies that we can obtain point-wise confidence intervals for $g(v; \hat{\beta})$ by using any consistent estimate of the variance of \eqref{eqn:avarproj} and standard normal quantiles.  One sufficient condition  for this to hold would  be that we are able to estimate all elements of $\eta$ at rates of at least $o_p(n^{-1/4})$. This is attainable by many non-parametric regression models under some conditions \citep{tsybakov2004introduction}. Alternatively, we could allow for slower rates for some nuisance components while requiring faster rates for others. For example, we could allow that $\|\mu_a - \hat{\mu}_a\| = o_p(n^{-1/6})$ and all other components to be estimated at $o_p(n^{-1/3})$. Regardless, as we saw when reviewing estimating the average effect $\psi_{M_1}$ in Section~\ref{ssec:review-estimation}, these conditions would imply that the regression of $\varphi(Z; \hat{\eta})$ onto $g(V; \beta)$ is asymptotically equivalent to the oracle regression of $\varphi(Z; \eta)$ onto $g(V; \beta)$. Intuitively, this is because the error induced by the nuisance estimation is decreasing at rates faster than $n^{-1/2}$. As with estimating $\psi_{M_1}$, this in turn follows because the bias of $\hat{\beta}$ is a function of the product of errors in the nuisance estimation, as shown in the statement of Theorem~\ref{theorem3}. 

\subsection{DR-Learner}

In some applications we may not be satisfied with a projection, and may instead wish to directly estimate $\psi_{M_1}(v)$. We propose estimating this quantity using a non-parametric second-stage regression model, which we call a DR-Learner following \cite{kennedy2022}. Algorithm 1 provides specific details. We then analyze the DR-Learner and derive results analogous to those found in \cite{kennedy2022}, giving model-free error bounds for arbitrary first-stage estimators which reveal that under some conditions, the DR-Learner is as efficient as an oracle estimator that regresses $\varphi(Z; \eta)$ onto $V$ directly. Our results are similar to \cite{kennedy2022}: the primary difference is that we must consider the sums of products of errors between several more sets of nuisance functions.

\begin{algorithm}\label{algorithm:mrlearner}
Let $(D_1^n, D_2^n)$ denote two independent samples of $n$ observations of $Z_i$. \newline

\begin{itemize}
\item Step 1: Nuisance training. Construct estimates of $\eta$ using $D_1^n$

\item Step 2: Pseudo-outcome regression. Construct the pseudo-outcomes $\varphi(Z; \hat{\eta})$ and regress it onto covariates $V$ in the test sample $D_2^n$, giving

\begin{align}
    \hat{\psi}_{M_1}^{dr}(v) = \hat{\mathbb{E}}_n\{\varphi(Z; \hat{\eta}) \mid V = v\}
\end{align}

\item Step 3: Cross-fitting (optional). Repeat Steps 1-2, swapping the roles of $D_1^n$ and $D_2^n$. Use the average of the resulting two estimates as a final estimate of $\psi_{M_1}(v)$.
\end{itemize}

\end{algorithm}

\begin{remark}
    In practice, when implementing either the DR-Learner or the projection estimator, one may wish to use sample-split estimates of $\eta$ and regress them onto $V$ using the entire sample, rather than averaging two separate estimates as suggested in Step 3 above. However, our results do not provide theoretic guarantees for this approach.
\end{remark}

Proposition 1 from \cite{kennedy2022} establishes general conditions where the error of a pseudo-outcome regression of $\hat{f}$ onto $V$ and an oracle regression of $f$ onto $V$ are asymptotically equivalent. This result relies on an assumption on the ``stability'' of the second-stage estimator with respect to a distance measure $d$ and the convergence in probability of $\hat{f}$ to $f$ with respect to $d$. Intuitively, estimator stability requires that the second-stage error between the pseudo-outcome regression and the oracle estimator converges in probability to the conditional bias of the pseudo-outcome estimates at a rate determined by root-mean squared error of the oracle estimator.\footnote{We provide the formal definition of stability in Appendix~\ref{appendix:definitions}.} Theorem 2 of \cite{kennedy2022} shows the conditions for the oracle efficiency of a DR-Learner for the CATE under a direct application of Proposition 1 from \cite{kennedy2022}, where $f$ is the uncentered influence function for the average treatment effect (ATE). We use the same approach for the CIIE to obtain analogous results, formalized in Corollary~\ref{corollary-drlearner}.

\begin{corollary}\label{corollary-drlearner}
Define $\hat{b}^\star(x) = \mathbb{E}\{\hat{\varphi}(Z) - \varphi(Z) \mid D_1^n, X = x\}$; in other words, $\hat{b}^\star(x)$ is the conditional bias of the estimated influence function at $X = x$ conditional on the training data. Assume

\begin{enumerate}
    \item $\hat{\mathbb{E}}_n$ is stable with respect to distance metric $d$
    \item $d(\hat{\varphi}, \varphi) \to^p 0$
\end{enumerate}

Let $\tilde{\psi}_{M_1}(v)$ denote an oracle estimator from a regression of the true efficient influence function onto $V$ and $K_n^{\star}(v)$ denote the oracle root mean square error, $\mathbb{E}[\{\tilde{\psi}_{M_1}(v) - \psi_{M_1}(v)\}^2]^{1/2}$. Then

\begin{align}\label{eqn:drlerrorbound2}
    \hat{\psi}^{dr}_{M_1}(v) - \tilde{\psi}_{M_1}(v) &= \hat{\mathbb{E}}_n\{\hat{b}^\star(X) \mid V = v\} + o_p\left(K_n^\star(v)\right) 
\end{align}

We provide an expression for $\hat{b}^\star(x)$ in Appendix~\ref{app:proofs}. Moreover, $\hat{b}^\star(x) \lesssim \hat{b}(x)$, where:

\begin{align}\label{eqn:bhatx}
    \hat{b}(x) &= T_{1n}(x) + T_{2n}(x) + T_{3n}(x) + T_{4n}(x) 
\end{align}
and

\begin{align*}
    T_{1n}(x) &= (\hat{\pi}_a(x) - \pi_a(x))\sum_{m_1, m_2}(\mu_a(m_1, m_2, x) - \hat{\mu}_a(m_1, m_2, x)) \\
    T_{2n}(x) &= \sum_{m_1, m_2}(\mu_a(m_1, m_2, x) - \hat{\mu}_a(m_1, m_2, x))(p(m_1, m_2 \mid a, x) - \hat{p}(m_1, m_2 \mid a, x) \\
    &+ (p(m_1 \mid a, x) - \hat{p}(m_1 \mid a, x)) + (p(m_1 \mid a', x) - \hat{p}(m_1 \mid a', x)) + (p(m_2 \mid a', x) - \hat{p}(m_2 \mid a', x))) \\
    T_{3n}(x) &= (\pi_a(x) - \hat{\pi}_a(x))(p(m_1 \mid a, x) - \hat{p}(m_1 \mid a, x)) + (p(m_1 \mid a', x) - \hat{p}(m_1 \mid a', x)) \\
    &+ (p(m_2 \mid a', x) - \hat{p}(m_2 \mid a', x))) \\
    T_{4n}(x) &= \sum_{m_1, m_2}(p(m_2 \mid a', x) - \hat{p}(m_2 \mid a', x))((p(m_1 \mid a, x) - \hat{p}(m_1 \mid a, x)) + p(m_1 \mid a', x) - \hat{p}(m_1 \mid a', x))
\end{align*}

Therefore, the DR-Learner is oracle efficient if $\hat{\mathbb{E}}_n\{\hat{b}(X) \mid V = v\} = o_p\left(K_n^\star(v)\right)$.
\end{corollary}

\begin{remark}
    As with the second-order expression in Theorem~\ref{theorem3}, the expression for $\hat{b}(x)$ depends on products of errors with the marginal mediator probabilities; however, these estimates, as described in Algorithm 1, come from estimates of the joint mediator probabilities. The error in the marginal probabilities is therefore of the same order as the sums of the errors in the joint probabilities across the relevant mediator values. As noted in Remark~\ref{remark:proj-order}, we may wish to estimate the marginal mediator probabilities separately if we believe that the underlying complexity of the marginal probabilities is simpler than the joint mediator probabilities, though such a situation may seem unlikely to occur in practice. 
\end{remark}

One interesting implication of Corollary~\ref{corollary-drlearner} is that  the rate of convergence is a function of the cardinality of the joint mediator values $k$. We  may eliminate this dependence  by invoking the following assumption:

\begin{assumption}\label{assumption:simplify} The smoothed product of errors between mediator probabilities and/or the outcome model are of the same order for any values of ($m_1, m_2$). For example, for any $(m_1, m_1')$ and $(m_2, m_2')$:
    
    \begin{align*}
     \hat{\mathbb{E}}_n\{\left(p(m_1 \mid a, X) - \hat{p}(m_1 \mid a, X)\right)\left(p(m_2 \mid a', X) - \hat{p}(m_2 \mid a', X)\right) \mid V = v\} = o_p(a_n)   
    \end{align*}
    
    and 
    
    \begin{align*}
    \hat{\mathbb{E}}_n\{\left(p(m_1' \mid a, X) - \hat{p}(m_1' \mid a, X)\right)\left(p(m_2' \mid a', X) - \hat{p}(m_2' \mid a', X)\right)\mid V = v\} = o_p(a_n)    
    \end{align*}
\end{assumption}

Assumption~\ref{assumption:simplify} would be reasonable if we do not believe the form of functional form of the mediator probabilities or outcome models varies in underlying complexity across different values of the mediators. 

We next consider the form of the second-stage regression $\hat{\mathbb{E}}_n$. Proposition 2 from \cite{kennedy2022} implies that when $\hat{\mathbb{E}}_n$ is a linear smoother of the form $\sum_i w_i(v; V^n)f(Z_i)$ and $\sum_i \lvert w_i(v; V^n) \rvert = \mathcal{O}_{p}(a_n)$, and $\hat{b}(x)$ can be expressed in the form of $\hat{b}_1(x)\hat{b}_2(x)$, then 

\begin{align*}
\hat{\mathbb{E}}_n\{\hat{b}(X) \mid V = v\} = \mathcal{O}_{p}(a_n\|\hat{b}_1\|_{w,2}\|\hat{b}_2\|_{w,2})   
\end{align*}
where 

\begin{align*}
    \|f\|_{w, 2} = \left[\sum_i\left\{\frac{\lvert w_i(v; V^n) \rvert}{\sum_j \lvert w_j(v; V^n) \rvert} \right\}\lvert f(Z_i)\rvert^2\right]^{1/2}
\end{align*}

Corollary~\ref{corollary-extension} applies this result to Corollary~\ref{corollary-drlearner}.

\begin{corollary}\label{corollary-extension}
Assume the conditions of Corollary~\ref{corollary-drlearner} and that $\hat{\mathbb{E}}_n$ is a minimax optimal linear smoother with $\sum_i \lvert w_i(v; V^n) \rvert = \mathcal{O}_p(1)$. Notice that 

\begin{align*}
    \hat{b}(x) = \sum_j \hat{b}_{j1}(x)\hat{b}_{j2}(x)
\end{align*}
where $j$ indexes all of the error products in \eqref{eqn:bhatx}. Therefore:

\begin{align*}
    \hat{\mathbb{E}}_n\{\hat{b}^\star(X) \mid V = v\} &\lesssim \hat{\mathbb{E}}_n\{\sum_j\hat{b}_{j1}(X)\hat{b}_{j2}(X) \mid V = v\}
\end{align*}
By Proposition 2 of \cite{kennedy2022} we then obtain that

\begin{align*}
\hat{\mathbb{E}}_n\{\hat{b}(X) \mid V = v\} &= \sum_j\mathcal{O}_p(\|\hat{b}_{j1}\|_{w,2}\|\hat{b}_{j2}\|_{w,2}) \asymp \max_j \mathcal{O}_p(\|\hat{b}_{j1}\|_{w,2}\|\hat{b}_{j2}\|_{w,2})
\end{align*}

\end{corollary}

To make Corollary~\ref{corollary-extension} concrete consider the case where $K_n^\star(v) = n^{-\theta}$.  This result  implies that when the second-stage regression estimator is a linear smoother, the DR-learner will achieve the corresponding oracle rate when, for example, all of nuisance errors are at least $o_p(n^{-\theta/2})$ in the $\|\cdot\|_{w,2}$ norm. More generally, the DR-Learner is oracle efficient as long as the highest order error product in \eqref{eqn:bhatx} is $o_p(n^{-\theta})$ in this same norm. Whether these rates are actually attainable in a given application depends on the underlying complexity of the nuisance functions. 

\section{Simulations}\label{sec:3}

We verify that the expected performance of these estimation strategies corresponds with the theory outlined above using a simulation study. First, we outline the data generating process; second, we demonstrate the performance of our proposed approaches on samples of size $n = 1000$ when estimating the nuisance functions using SuperLearner \citep{van2007super}; finally, we compare the convergence rates of the DR-Learner versus a plugin approach while controlling the convergence rates of the nuisance estimation.

\subsection{Setup}

To illustrate our proposed approach, we conduct a simulation study with a one-dimensional covariate $X$ (so that $V = X$). Figure~\ref{fig:sim-nuis-funs} illustrates the simulated nuisance functions as a function of $X$. One aspect of this setup is that the outcome models and the propensity score models have complexity unlikely to be fully captured by simple generalized linear models, motivating our use of non-parametric methods. A second aspect is that the implied function $\psi_{M_1}(X)$, illustrated in Figure~\ref{fig:sim-estimands}, is less complex than these functions. We provide all details about the data generating processes for our simulations, including the functions illustrated in Figure~\ref{fig:sim-nuis-funs}, in Appendix~\ref{app:simulation}.

\begin{figure}[H]
\begin{center}
    \mycaption{Simulation: selected nuisance functions}{Nuisance function specifications for simulation study}\label{fig:sim-nuis-funs}
    \includegraphics[scale=0.4]{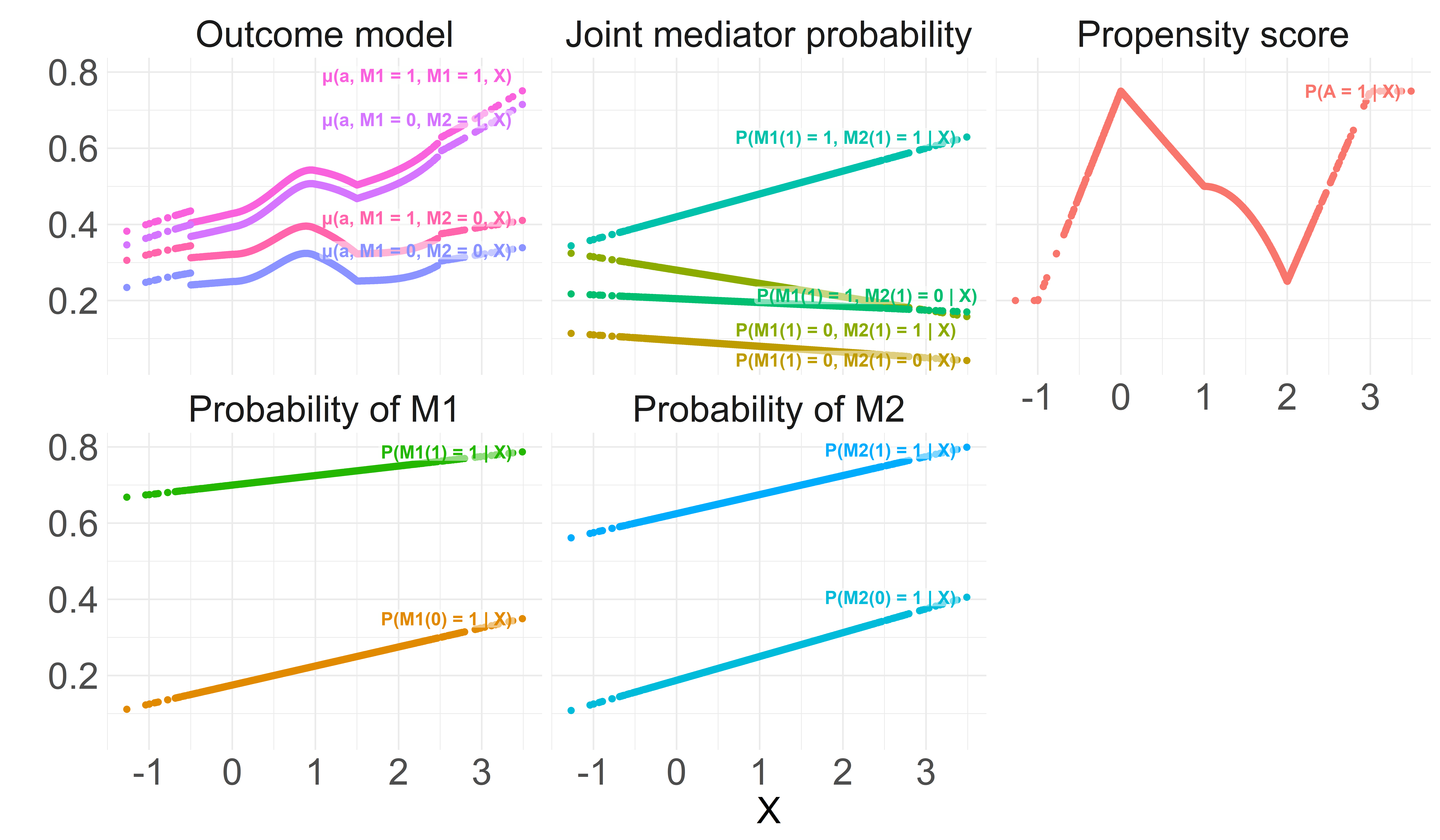}
\end{center}
\end{figure}

Figure~\ref{fig:sim-estimands} also illustrates the implied curves of $\psi_{M_2}(X)$ and $\psi(X)$, as well as the curves for the proportion mediated via each mediator (e.g. $\psi_{M_1}(X) / \psi(X)$). The effects are entirely mediated via $M_1$ and $M_2$,\footnote{Additionally, the covariant effects due to the dependence of the mediators on each other is close to zero throughout.} and the proportion mediated via $M_1$ increases with $X$.

\begin{figure}[H]
\begin{center}
    \mycaption{Estimands}{Total effects, indirect effects, and proportion mediated as a function of $X$}\label{fig:sim-estimands}
    \includegraphics[scale=0.4]{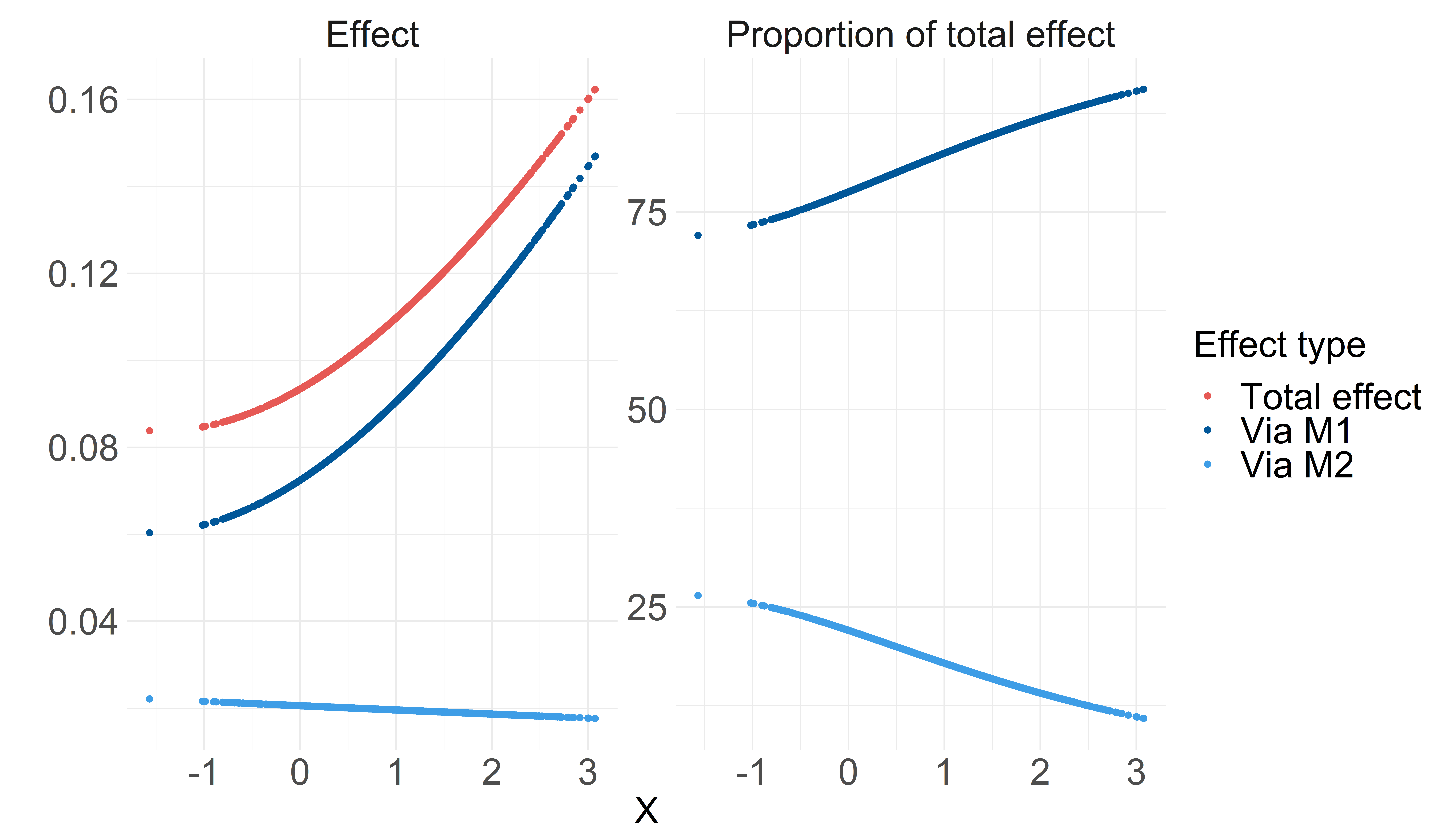}
\end{center}
\end{figure}

\subsection{Estimation: SuperLearner}

We evaluate the performance of each estimator across 1000 simulations on test samples of size 1000. To estimate the nuisance parameters we use SuperLearner, using both the ``SL.glm'' and ``SL.ranger'' libraries.\footnote{In practice it is desirable to use as many libraries as possible when running SuperLearner; however, for the sake of computation time we only use these two libraries for our simulation study.} These model our nuisance functions as a weighted combination of predictions from logistic regression and random forests models. After estimating the nuisance parameters using samples of size 1000, we use a separate test sample to estimate the DR-Learner and projection estimators. We then predict the points at $X = 0$ and $X = 2$. 

While we expect both estimates to be consistent, the mean square error at each point should differ by constants: this is due to differing inverse weights in the expression for $\varphi(Z)$. Figure 8 in Appendix 12 illustrates the maximum possible inverse weight as a function of X: at $X = 0$, the maximum weight is lowest while at $X = 2$ the maximum weight is highest. These two points arguably reflect the easiest and hardest parts of the covariate space to estimate, with the point where $X = 2$ reflecting a ``worst-case scenario'' in our simulation. As long as our assumptions hold, these weights do not affect the asymptotic results. However, they can affect their performance in finite samples, with higher variance estimates where the inverse weights are large. Additionally, confidence intervals may have under-coverage, since their validity is also based on asymptotic approximations. Consequently, we expect the simulations to show better results when estimating the CIIE at $X = 0$ compared to $X = 2$, also with possibly better coverage in these regions. More generally, this comparison highlights a key limitation of our proposed approach: in an actual sample it may be challenging to estimate conditional effects where the inverse weights are extreme.

\begin{table}[ht]
\centering
\caption{Projection estimators: simulation performance, $n = 1000$ \\ (averaged over 1000 simulations)} 
\label{tab:sim1}
\begin{tabular}{rllrrrrr}
  \hline
Point & Strategy & Projection & Truth & Bias & Std & RMSE & Coverage \\ 
  \hline
  0 & Plugin & Linear & 0.07 & 0.00 & 0.03 & 0.03 & 3.10 \\ 
    2 & Plugin & Linear & 0.11 & -0.04 & 0.02 & 0.04 & 1.00 \\ 
    0 & Plugin & Quadratic & 0.07 & -0.00 & 0.03 & 0.03 & 3.10 \\ 
    2 & Plugin & Quadratic & 0.11 & -0.04 & 0.02 & 0.05 & 0.80 \\ 
    0 & Efficient & Linear & 0.07 & 0.00 & 0.06 & 0.06 & 95.50 \\ 
    2 & Efficient & Linear & 0.11 & -0.00 & 0.08 & 0.08 & 93.80 \\ 
    0 & Efficient & Quadratic & 0.07 & 0.00 & 0.06 & 0.06 & 95.10 \\ 
    2 & Efficient & Quadratic & 0.11 & 0.00 & 0.10 & 0.10 & 93.50 \\ 
   \hline
\end{tabular}
\end{table}

Table \ref{tab:sim1} considers the projection estimates, and displays the bias, RMSE, and confidence interval coverage associated with our proposed approach (``Efficient'') and with a plugin approach that regresses plugin estimates of $\psi_{M_1}(x)$ on the same model. Specifically, the plugin estimates involve estimating each component in the expression for $\psi_{M_1}(x)$ and regressing these estimates, rather than estimates of $\varphi(Z; \eta)$, onto $g(X; \beta)$. We predict the projection at the points $X \in \{0, 2\}$ using either a linear or quadratic projection. While the plugin estimator has lower RMSE, the confidence interval coverage is close to zero. This reflects that the bias associated with the nuisance estimation does not converge quickly enough to zero, and we therefore cannot ignore the estimation error in the second-stage regression when conducting inference. By contrast, we obtain close to nominal coverage rates for the efficient estimator. Finally, as expected, the point $X = 0$ is easier to estimate than $X = 2$, reflected by the fact that the RMSE is higher for estimates at $X = 2$ than $X = 0$. 

Table \ref{tab:sim2} displays analogous results when using the DR-Learner to target the true CIIE rather than its projection.\footnote{Because the implied curves are relatively easy to approximate using a linear or quadratic model, we see that the ``Truth'' column in Table~\ref{tab:sim2} is identical to the ``Truth'' column in Table~\ref{tab:sim1}. In fact these are only identical rounded to the second decimal place, but this illustrates that the projections provide good estimates of the true quantities in our simulation.} We use smoothing splines with the default tuning parameters for the second-stage regression,\footnote{The tuning parameters are chosen by default using generalized cross-validation. We technically should account for post-selection inference; however, this does not seem to make a difference in practice in our simulations.} and use the variance estimates from the smoothing matrix and assume that the distribution of the estimates is asymptotically normal to generate confidence intervals. Table \ref{tab:sim2} shows that this procedure yields approximately nominal coverage rates. 

\begin{table}[ht]
\centering
\caption{Non-parametric estimators: simulation performance, $n = 1000$ \\ (averaged over 1000 simulations)} 
\label{tab:sim2}
\begin{tabular}{rlrrrrl}
  \hline
Point & Strategy & Truth & Bias & Std & RMSE & Coverage \\ 
  \hline
  0 & DR-Learner & 0.07 & 0.00 & 0.08 & 0.08 & 94.3 \\ 
    0 & Plugin & 0.07 & -0.00 & 0.03 & 0.03 & - \\ 
    2 & DR-Learner & 0.11 & -0.00 & 0.11 & 0.11 & 92.6 \\ 
    2 & Plugin & 0.11 & -0.04 & 0.03 & 0.05 & - \\ 
   \hline
\end{tabular}
\end{table}

As with the projection approach, the corresponding plugin approach has lower RMSE than the DR-Learner. This is likely a function of the inverse-probability weights associated with the DR-Learner, which could cause this result for a fixed sample size. 

\subsection{Estimation: convergence rates}

We conclude by examining the convergence rates of the DR-Learner versus a plugin estimator by specifying the convergence rates of the nuisance estimation. Roughly, we add $\mathcal{N}(C/n^{\alpha}, C/n^{2\alpha})$ to the true values of the nuisance parameters on the logistic scale to simulate estimates. Using these results we construct ``estimated'' influence functions and regress them onto $X$ using smoothing splines. In contrast to the simulation study above, we use these simulations to estimate the integrated mean square error across the entire domain of $X$. Figure~\ref{fig:sim-imse} displays the results. 

\begin{figure}[H]
\begin{center}
    \mycaption{Convergence of DR-Learner versus Plugin and Oracle estimators}{Scaled RMSE of each estimation strategy as a function of sample size. ``Slow'' nuisance function is estimated $\mathcal{O}_p(n^{-1/10})$ rates, remaining at $\mathcal{O}_p(n^{-1/2})$ rates.}\label{fig:sim-imse}
    \includegraphics[scale=0.45]{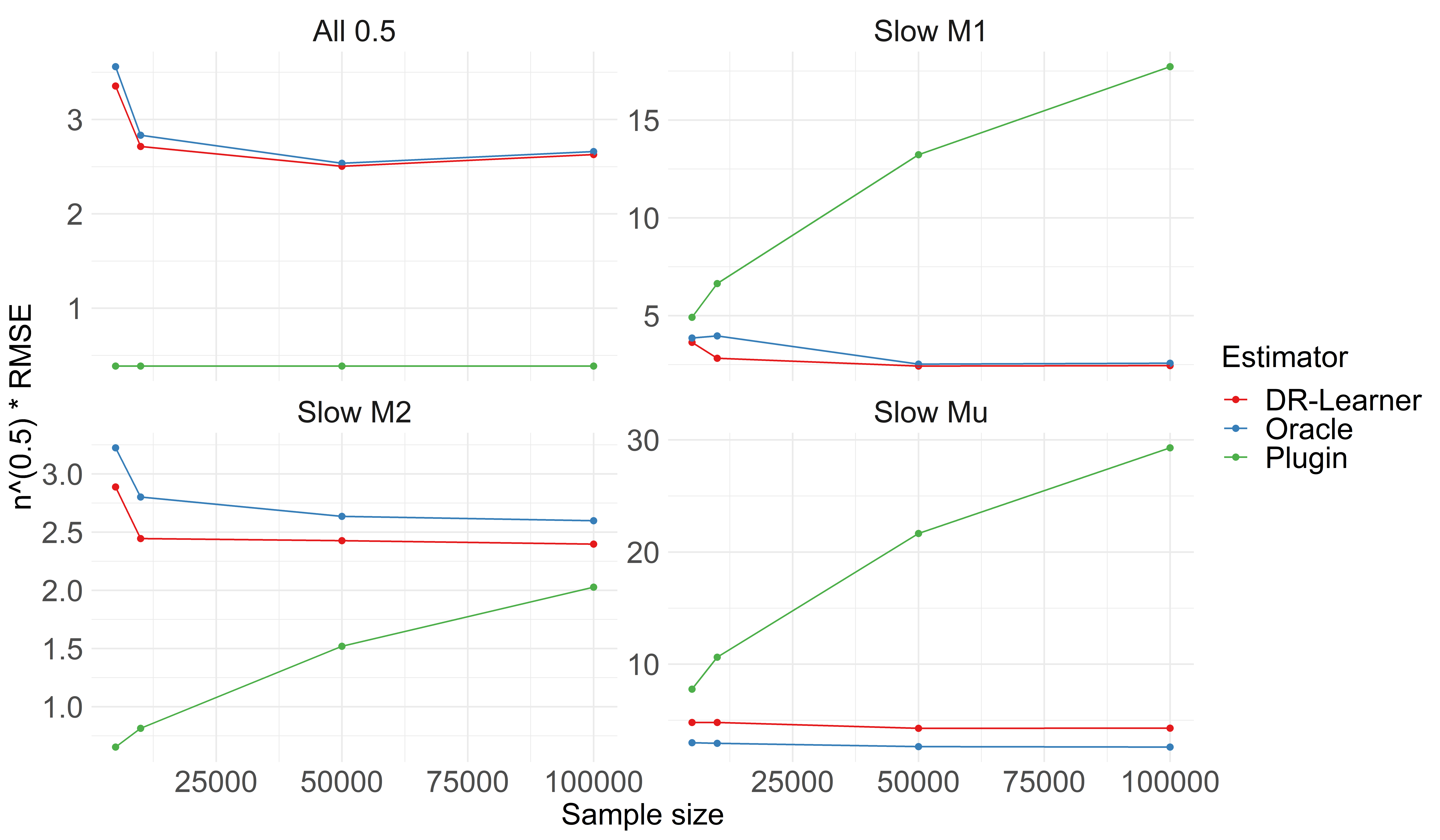}
\end{center}
\end{figure}

The y-axis displays the RMSE of each estimator scaled by $\sqrt{n}$, while the x-axis displays the sample sizes. The top left panel considers the case where we set $\alpha = 0.5$ for all nuisance parameters. The other panels instead set $\alpha = 0.1$ for the function indicated in the panel title. As expected, when all nuisance functions are estimated at the same rate the plugin estimator converges at the same rate as the DR-Learner. However, once one of the nuisance functions  is  estimated slowly, the plugin estimator converges more slowly (illustrated by the diverging green lines), while the DR-Learner appears to attain the oracle rates of convergence in all settings. We again observe that despite the slower convergence rates, the plugin estimator has lower RMSE than either the DR-Learner or oracle estimators in some settings. 

In Appendix~\ref{sec:propmed} we present additional results where we estimate the CIIE as a proportion of the corresponding CATE. We outline two general approaches to this problem: first, where we estimate the CIIE and the CATE and take the ratio of these estimates; second, where we derive the influence function for the mean of the ratio and regress this onto $V$ (this is similar to \cite{cuellar2020non}, who consider estimating a conditional risk-ratio). Our simulations show that this quantity is quite difficult to estimate well due to the high-variance of the estimators, although we are able to construct confidence intervals with approximately nominal coverage rates using either approach.

\section{Sensitivity analysis}\label{sec:4}

We consider estimating $\psi_{M_1}(v)$ when the assumption that the potential outcomes $Y^{am_1m_2}$ are independent of the mediators given the covariates and that $A = a$ does not hold. This might occur, for example, if there were a post-treatment confounder $L$ that occurs prior to $M$ but after $A$; or, a pre-treatment confounder $U$ that affects the Y-M relationship but not the A-M or A-Y relationships. We first outline a general sensitivity framework to generate bounds on the conditional or average effects while specifying the degree of these types of violations, where our approach builds from work in \cite{luedtke2015statistics}. We extend the projection estimator and DR-Learner to estimate bounds on the \textit{conditional} effects, though our proposed method naturally also suggests influence function based estimators of the bounds on the \textit{average} effects. These analyses can help an analyst assess how much inferences may change given a specified degree of confounding. We first briefly introduce additional assumptions and notation to ease exposition. 

\subsection{Setup and notation}

First, we assume for simplicity that $V = X$; that is, that the conditioning set is identical to the observed confounders, so that our target estimand is $\psi_{M_1}(x)$. Second, to construct a bound for $\psi_{M_1}(x)$, it will be helpful to write $\psi_{M_1}(x) = \psi_{M_1, a}(x) - \psi_{M_1, a'}(x)$, where:

\begin{align*}
    \psi_{M_1, a}(x)  &= \sum_{m_1,m_2}\mathbb{E}[Y^{m_1m_2} \mid a, x][p(m_1 \mid a, x)]p(m_2 \mid a', x) \\
    \psi_{M_1, a'}(x) &= \sum_{m_1,m_2}\mathbb{E}[Y^{m_1m_2} \mid a, x][p(m_1 \mid a', x)]p(m_2 \mid a', x)
\end{align*}
 Third, for any $(m_1', m_1, m_2', m_2)$ we let 

\begin{align*}
\mathbb{E}[Y^{m_1m_2} \mid m_1', m_2', a, x] = \mu_{am_1m_2}^\star(m_1', m_2', x)
\end{align*}
Assuming that Y-M ignorability holds when it does not, we define the biased target of our estimator of $\psi_{M_1}(x)$:

\begin{align*}
    \bar{\psi}_{M_1}(x) &= \bar{\psi}_{M_1, a}(x) - \bar{\psi}_{M_1, a'}(v) \\
    &= \sum_{m_1, m_2}\mu_a(m_1, m_2, x)[p(m_1 \mid a, x) - p(m_1 \mid a', x)]p(m_2 \mid a', x) 
\end{align*}
Finally, to ease notation, we let $(\cdot)$ indicate the arguments $(m_1, m_2, x)$. 

While we focus the remaining discussion on the case where $V = X$, all of these results also hold at a point $V = v$ by averaging the relevant quantities over the distribution $p(W \mid V = v)$, recalling that $X = [V, W]$.

\subsection{Sensitivity framework}

The bias $\bar{\psi}_{M_1}(x) - \psi_{M_1}(x)$ occurs because for any $(m_1, m_2, x)$, in general $\mathbb{E}[Y^{m_1m_2} \mid a, x] \ne \mu_a(\cdot)$. We first consider this bias of the outcome regression. Proposition 2 shows that we can bound this bias as a function of $\mu_a(\cdot)$ and a sensitivity parameter $\tau$. We consider a general framework where the meaning of $\tau$ changes based on the chosen sensitivity analysis; however, we describe the interpretation of this parameter under each assumption below.

\begin{proposition}\label{prop:bound-mu}
Assume that we know some functions $b_l(\cdot; \tau)$ and $b_u(\cdot; \tau)$ parameterized by $\tau$ such that for every $(m_1, m_2, x)$:

\begin{align*}
b_u(\cdot; \mu_a, \tau) \ge \mu_{am_1m_2}^\star(M_1 \ne m_1, M_2 \ne m_2, x) - \mu_a(\cdot) \ge b_l(\cdot; \mu_a, \tau) 
\end{align*}
This implies the following bounds:

\begin{align*}
    b_u[\cdot; \mu_a, \tau][1 - p(m_1, m_2 \mid a, x)] \ge \mathbb{E}[Y^{m_1m_2} \mid a, x] - \mu_a(\cdot) \ge  b_l[\cdot; \mu_a, \tau][1 - p(m_1, m_2 \mid a, x)]
\end{align*}
\end{proposition}

Different assumptions on the selection process can motivate different functions $b_l$ and $b_u$. For example, let $\tau(m_1, m_2, x) \in [0, 1]$. Consider the following three assumptions for any $(m_1' \ne m_1)$ and $(m_2' \ne m_2)$:

\begin{align}
    \label{eqn:a1}& \tau(\cdot) \ge \rvert \mu_{am_1m_2}^\star(m_1',m_2',x) - \mu_a(\cdot) \lvert \\
    \label{eqn:a2}&\tau(\cdot) + \mu_a(\cdot)[1 - \tau(\cdot)] \ge \mu_{am_1m_2}^\star(m_1',m_2',x) \ge (1 - \tau(\cdot))\mu_a(\cdot) \\
    \label{eqn:a3}&1 / (1-\tau(\cdot)) \ge \mu_{am_1m_2}^\star(m_1',m_2',x) / \mu_a(\cdot) \ge (1 - \tau(\cdot)) 
\end{align}

Under assumption~\ref{eqn:a1}, $\tau$ bounds the absolute value of the difference between the regression functions $\mu_{am_1m_2}^\star(m_1',m_2',x)$ and $\mu_a(\cdot)$ for each level of $(m_1, m_2, x)$. Under assumption~\ref{eqn:a3}, $\tau$ parameterizes deviations of these same regression functions on the risk-ratio scale: below by $1 - \tau$, and above by $\frac{1}{1-\tau}$.\footnote{Often a sensitivity parameter $\Gamma$, which equals $\frac{1}{1-\tau}$, is used instead in this framework. We choose $\tau$ here to maintain consistency with the other two possible approaches.} Finally, the meaning of $\tau$ changes for the upper and lower bound under assumption~\ref{eqn:a2}. First, the lower bound is equivalent to the lower bound in assumption~\ref{eqn:a3}, and $\tau$ retains an equivalent meaning in this case. However, the upper bound instead specifies that $(1 - \mu_{am_1m_2}^\star(m_1',m_2',x)) / (1 - \mu_a(\cdot)) \ge (1 - \tau(\cdot))$. Under this assumption $\tau$ parameterizes the risk-ratio of the regression function when the event $Y$ did not occur.\footnote{This bound can be used when $Y$ is binary, but more generally when $Y$ is bounded and rescaled to fall within zero and one.} We discuss the trade-offs between these assumptions in greater detail below.

While these assumptions provide bounds on the true outcome model, they imply, but are not equivalent to, bounds on $\psi_{M_1}(x)$. Proposition~\ref{prop:bounds-psi-m1} provides a generic form of these bounds.

\begin{proposition}\label{prop:bounds-psi-m1}

Consider assumptions (\ref{eqn:a1})-(\ref{eqn:a3}) and a sensitivity parameter $\tau(x)$ that is valid for any value of $(m_1, m_2)$ at the point $X = x$. All $[b_l, b_u]$ implied by these assumptions can be expressed as: 

\begin{align*}
    [b_l, b_u] = [(c_l\mu_a(\cdot) + t_l)f_l(\tau(x)), (c_u\mu_a(\cdot) + t_u)f_u(\tau(x))]
\end{align*}
for constants $(c_l, c_u, t_l, t_u) \in \{0, 1\}^4$ and functions $f_l$ and $f_u$. At a point $X = x$, this yields the following bounds on $\psi_{M_1}(x)$:

\begin{align}
    \label{eqn:upperboundx}&\psi_{M_1, ub}(x; \tau) = \\
    \nonumber&[\bar{\psi}_{M_1}(x) + \bar{\psi}_{M_1, a}(x)f_u(\tau)c_u - \bar{\psi}_{M_1, a'}(x)f_l(\tau)c_l + t_uf_u(\tau) - t_lf_l(\tau) \\
    \nonumber&- f_u(\tau)\sum_{m_1,m_2}[c_u\mu_a(m_1,m_2,x) + t_u]p(m_1, m_2 \mid a, x)p(m_1 \mid a, x)p(m_2 \mid a', x) \\
    \nonumber&+ f_l(\tau)\sum_{m_1,m_2}[c_l\mu_a(m_1,m_2,x) + t_l]p(m_1, m_2 \mid a, x)p(m_1 \mid a', x)p(m_2 \mid a', x)] \\
    \label{eqn:lowerboundx}&\psi_{M_1, lb}(x; \tau) = \\
    \nonumber&[\bar{\psi}_{M_1}(x) + \bar{\psi}_{M_1, a}(x)f_l(\tau)c_l - \bar{\psi}_{M_1, a'}(x)f_u(\tau)c_u + t_lf_l(\tau) - t_uf_u(\tau) \\
    \nonumber&- f_l(\tau)\sum_{m_1,m_2}[c_l\mu_a(m_1,m_2,x) + t_l]p(m_1, m_2 \mid a, x)p(m_1 \mid a, x)p(m_2 \mid a', x) \\
    \nonumber&+ f_u(\tau)\sum_{m_1,m_2}[c_u\mu_a(m_1,m_2,x) + t_u]p(m_1, m_2 \mid a, x)p(m_1 \mid a', x)p(m_2 \mid a', x)]
\end{align}
\end{proposition}

\begin{remark}
If we desired bounds at the point $V = v$, we could choose a $\tau$ valid for any realization of $W$ at the point $V = v$ and average these expressions over the conditional distribution of $W$ given $V = v$. If we desired bounds on the average effect, we could choose a $\tau$ valid for all $(x, m_1, m_2)$ and marginalize the above expressions over the entire covariate distribution. 
\end{remark}

The assumptions outlined in equations (\ref{eqn:a1})-(\ref{eqn:a3}) yield different bounds. While \eqref{eqn:a1} is perhaps most intuitive, for a given $\tau$ the width of the implied bounds on $\psi_{M_1}$ in equations (\ref{eqn:upperboundx}) and (\ref{eqn:lowerboundx}) are twice as large as those from \eqref{eqn:a2}, and are thus perhaps less useful in practice than the others. Comparing the assumptions in equations (\ref{eqn:a2}) and (\ref{eqn:a3}) is difficult: for a fixed $\tau$ the scale of the confounding for the upper bounds of $\mu_{am_1m_2}^\star$ in these equations is simply different. One benefit of \eqref{eqn:a3} is that it only requires reasoning about the risk-ratio $\mu_{am_1m_2}^\star(m_1',m_2',x) / \mu_a(\cdot)$. By contrast \eqref{eqn:a2} requires additional reasoning about the risk-ratio bias in the estimates that the event $Y$ did not occur, demanding more thought from the user. On the other hand, when using a binary outcome, equation (\ref{eqn:a3}) may result in an upper bound on $\mu_{am_1m_2}^\star$ greater than one, while equation (\ref{eqn:a2}), and the resulting bound on $\psi_{M_1}(x)$, will always respect the parameter space. Finally, for a fixed value of $\tau$ it is unclear whether the bounds yielded by equations (\ref{eqn:a2}) or (\ref{eqn:a3}) will be wider. However, for fixed $(m_1, m_2, x)$, the upper bound on $\mu_{am_1m_2}^\star$ in \eqref{eqn:a3} will always be larger than the bounds in \eqref{eqn:a2} when $\mu_a(\cdot) \ge 0.5$. Heuristically, we may therefore expect the bounds yielded by \eqref{eqn:a3} to be narrower than those from \eqref{eqn:a2} when $\mu_a$ tends to be small across values of $(m_1,m_2,x)$. As a final point, the bound given by \eqref{eqn:a2} is only useful for binary outcomes, or bounded outcomes rescaled to fall within 0 and 1, so that the sensitivity analysis given by \eqref{eqn:a3} is more general.

Finally, we can extend this general approach in several ways. For example, Assumptions (\ref{eqn:a1})-(\ref{eqn:a3}) yield similar bounds for $\psi_{M_2}$, $\psi_{Cov}$, $\psi_{IDE}$ by averaging over the relevant distributions. We discuss these extensions in Appendix~\ref{app:bounds}. We could also specify $\tau$ as a function that varies across $(m_1, m_2, x)$ to arrive at a slightly different expression for the bounds. However, specifying this function would be challenging in practice.

\subsection{Illustration}

Figure~\ref{fig:sim-bounds} uses simulated data to illustrate the estimand, the biased target, and the bounds as a function of $x$ choosing $\tau = 0.1$ under \eqref{eqn:a2} and $\tau = 0.15$ under \eqref{eqn:a3}. These parameters reflect the true maximal values of $\tau$ guaranteed to hold under these assumptions in our simulation. We obtain biased estimates for our outcome model based on \eqref{eqn:a2} and a parameter $\tau$ that varies between $0.1 * \{1/3, 2/3, 1\}$ depending on the value of $x$. We describe the selection process in greater detail in Appendix~\ref{sec:selection}. The upper and lower bounds are depicted in purple and red, while the orange and green lines reflects $\bar{\psi}_{M_1}(x)$ and $\psi_{M_1}(x)$, respectively. As $x$ increases, the bounds given by \eqref{eqn:a3} are at first narrower and eventually wider than the bounds given by \eqref{eqn:a2}. This is generally expected as the values of $\mu_a(m_1,m_2,x)$ tend to increase with $x$ (see Figure \ref{fig:sim-selectionbias2} in Appendix~\ref{app:simulation}). These bounds are also quite conservative: assuming \eqref{eqn:a2}, the bounds are only guaranteed to hold for all $x$ for $\tau = 0.1$. However, even $\tau = 0.02$ provides valid upper and lower bounds across the entire domain in our simulation. Of course, the gap between the value of $\tau$ guaranteed to hold and the minimum $\tau$ that actually does may be smaller for other data distributions; however, it does suggest that these bounds may be conservative in practice.

\begin{figure}[H]
\begin{center}
    \mycaption{Bounds on CIIE}{Target estimand in green, biased target of inference in orange. Purple and red lines reflect upper and lower bounds that differ in terms of $\tau$ specification.}\label{fig:sim-bounds}
    \includegraphics[scale=0.35]{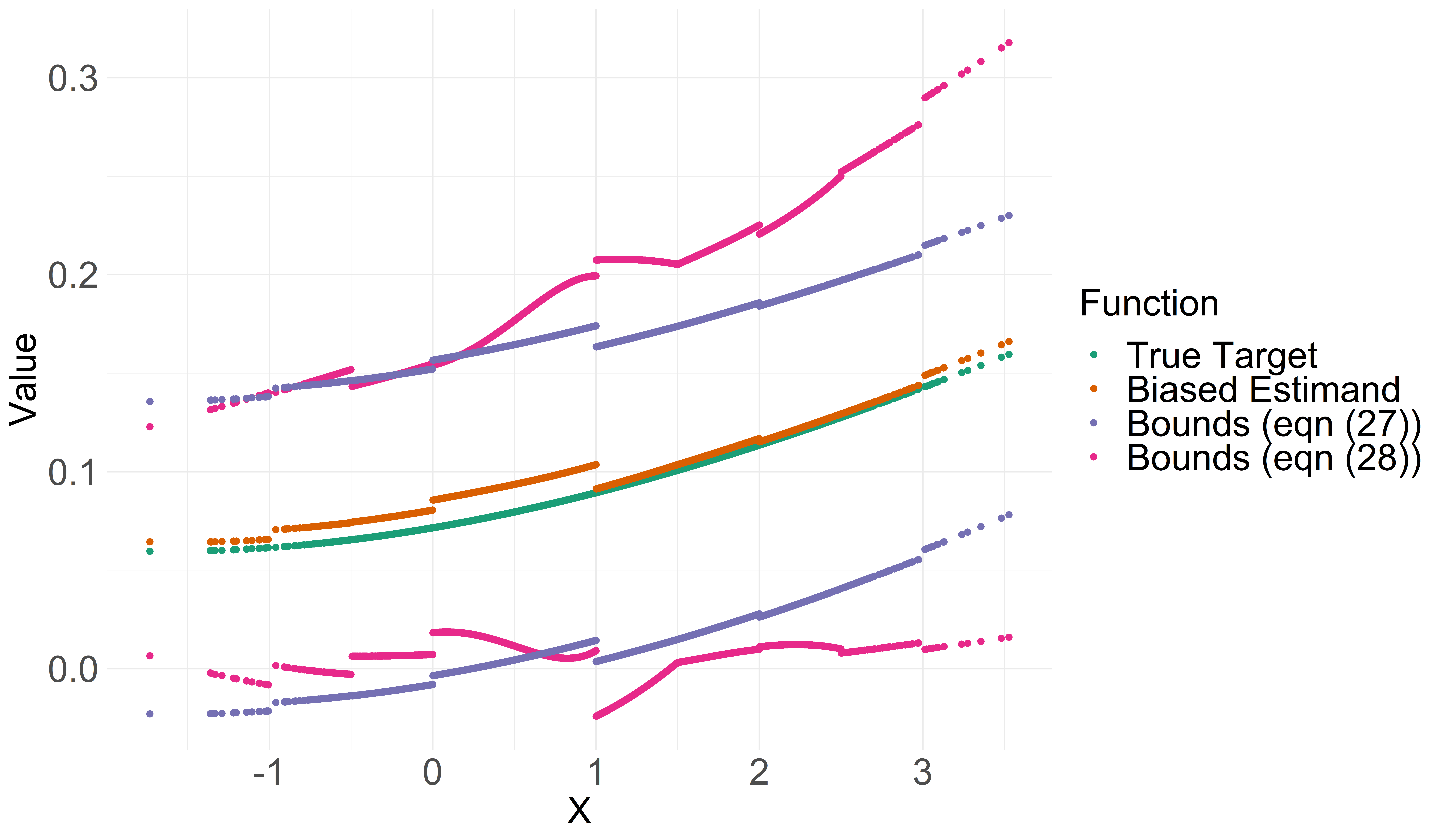}
\end{center}
\end{figure}

\subsection{Alternative approach}

\cite{tchetgen2012semiparametric} and \cite{vanderweele2014sensitivity} considered similar approaches for bounds on natural effects under the assumption that M-A and Y-A ignorability holds but that Y-M ignorability does not. Both proposals assume a known selection function that holds with equality rather than inequality. We could modify our proposed approach in a similar spirit. For example, we could assume that for all $(m_1' \ne m_1)$, $(m_2' \ne m_2)$:

\begin{align*}
\mu_{am_1m_2}^\star(m_1', m_2', x) - \mu_a(m_1, m_2, x) = f(\tau)[c\mu_a(m_1, m_2, x) + t]
\end{align*}
We could then recover $\psi_{M_1}(x)$ (and any averages of it) exactly as:

\begin{align*}
    \psi_{M_1}(x) &= \bar{\psi}_{M_1}(x)[1 + cf(\tau)] \\
    &-c\sum_{m_1,m_2}\mu_a(m_1,m_2,x) p(m_1, m_2 \mid a, x)[p(m_1 \mid a, x) - p(m_1 \mid a', x)]p(m_2 \mid a', x) \\
    &- tf(\tau)\sum_{m_1,m_2}p(m_1, m_2 \mid a, x)[p(m_1 \mid a, x) - p(m_1 \mid a', x)]p(m_2 \mid a', x)
\end{align*}
While such an assumption would allow us to point identify $\psi_{M_1}(x)$, the concept requires knowledge about a selection function that we are unlikely to have. Despite the conservative inferences, we therefore prefer our proposed approach.

\subsection{Estimation}

We extend all of the above methods to estimate the bounds on $\psi_{M_1}(v)$. Theorem~\ref{theorem4} in Appendix~\ref{sec:othertheorems} provides the expressions for the influence functions of the bounds on the average effect $\psi_{M_1}$. Equipped with this expression, we can again use a projection estimator or the DR-Learner to estimate the bounds.\footnote{Moreover, the choice of $\tau$ need only be valid across the domain of $x$ where the weights $w(v; V^n)$ in the second-stage regression are non-zero.} Intuitively, these approaches share the property that the upper bound on the convergence rates in the estimation is governed by the products of errors in the nuisance estimation. Corollaries~\ref{corollary4} and~\ref{corollary5} in Appendix~\ref{sec:othertheorems} give formal statements of these results using the lower bound as an example, although we can derive an upper bound analogously. We also provide expressions for the influence function for the bounds of $\psi_{M_2}$, $\psi_{Cov}$, and $\psi_{IDE}$ in Appendix~\ref{app:bounds} and conjecture that results analogous to Corollaries~\ref{corollary4} and~\ref{corollary5} can be derived for these estimands. Finally, for a fixed $\tau$, Theorem~\ref{theorem5} in Appendix~\ref{sec:othertheorems} establishes the conditions where the one-step estimator for the bounds on the \textit{average} effects is root-n consistent and asymptotically normal. We illustrate this estimation procedure in the application in Section~\ref{sec:5}.

\section{Application}\label{sec:5}

To demonstrate these methods we revisit the data and application considered in \cite{rubinstein2023}, who examined the effect of COVID-19 vaccinations on depression, social isolation, and worries about health during February 2021 using the COVID-19 Vaccine Trends and Impact Survey (CTIS). The Delphi group at Carnegie Mellon University (CMU) conducted the CTIS from April 2020 through June 2022 in collaboration with the Facebook Data for Good group \citep{salomon2021us}. Using this data, \cite{rubinstein2023} posit a model that COVID-19 vaccinations affect depression via a direct path, social isolation, and worries about health. Using the decomposition from \cite{vansteelandt2017interventional}, they found that pathways via social isolation were more important than pathways via worries about health in explaining the effect of COVID-19 vaccinations on depression. We refer to that paper for details on the data and the limitations of this analysis. While this study examined effect heterogeneity, the authors only examined heterogeneity within discrete subgroups and primarily focused on the outcomes analysis. Moreover, the authors did not find substantial effect heterogeneity with respect to the mediation analysis among the specified subgroups.

We examine the decomposition of the total effect within the following subset of CTIS respondents: employed, non-Hispanic, White respondents aged 25-54 with at least a college degree, no chronic health conditions, who work outside the home, and who had previously received an influenza vaccination. This included a total of 13,764 individuals. Table~\ref{tab:table1} displays the average effect estimates using influence-function based estimators, where the nuisance parameters were estimated using twenty stacked XGBoost models with different hyperparameter settings on the full dataset. While restricted to a much smaller subgroup, these results are qualitatively comparable to the average estimates in \cite{rubinstein2023}.

\begin{table}[ht]\caption{Average effect estimates on CTIS subset in February 2021, $N = 13,764$}\label{tab:table1}
\centering
\begin{tabular}{lrrr}  
\hline
Estimand & Estimate & Lower 95\% CI & Upper 95\% CI \\ 
  \hline
Total effect & -4.61 & -6.10 & -3.12 \\ 
  IIE - M1 & -1.86 & -2.45 & -1.28 \\ 
  IIE - M2 & -0.46 & -0.71 & -0.20 \\ 
  IIE - Cov & 0.08 & -0.15 & 0.30 \\ 
  IDE & -2.37 & -3.81 & -0.92 \\ 
   \hline
\end{tabular}
\end{table}

We next compare whether the interventional effects differed among those who live in counties where Trump led Biden by 50 percentage points in the 2020 election (``Trump counties''), and those where Biden led Trump by 50 percentage points (``Biden counties''). By limiting our sample to the subgroup defined above, effect heterogeneity across the Biden vote share may proxy for how social factors may moderate the mediated effects.\footnote{Since we are unable to fully control for socio-demographic variables, this variable may also pick up on these moderating influence of these omitted factors that vary with the Biden vote share.} Specifically, we hypothesize that these relatively educated, vaccine-accepting, and health conscious respondents who live in Trump-voting counties may have lower total effects than those who live in Biden areas due to the added stress of living in areas that generally took relatively fewer COVID precautions. We similarly hypothesize that the effects via worries about health might be lower in Trump-voting counties than Biden-voting counties for this same reason. Figure~\ref{fig:application} displays the results using both the DR-Learner and projection estimators at these two points, where we use a simple linear model for the projection.\footnote{While we use sample-splitting to estimate $\eta$ and construct influence function value estimates, we run the second-stage regression on the entire sample instead of averaging two separate estimates.} Figure~\ref{fig:application1} in Appendix~\ref{app:application} display the entire estimated curves.

\begin{figure}[H]
\begin{center}
    \mycaption{Application results}{Conditional total, indirect, and direct effects in Biden (Pct Dem Lead = 50) versus Trump (Pct Dem Lead = -50) counties}\label{fig:application}
    \includegraphics[scale=0.35]{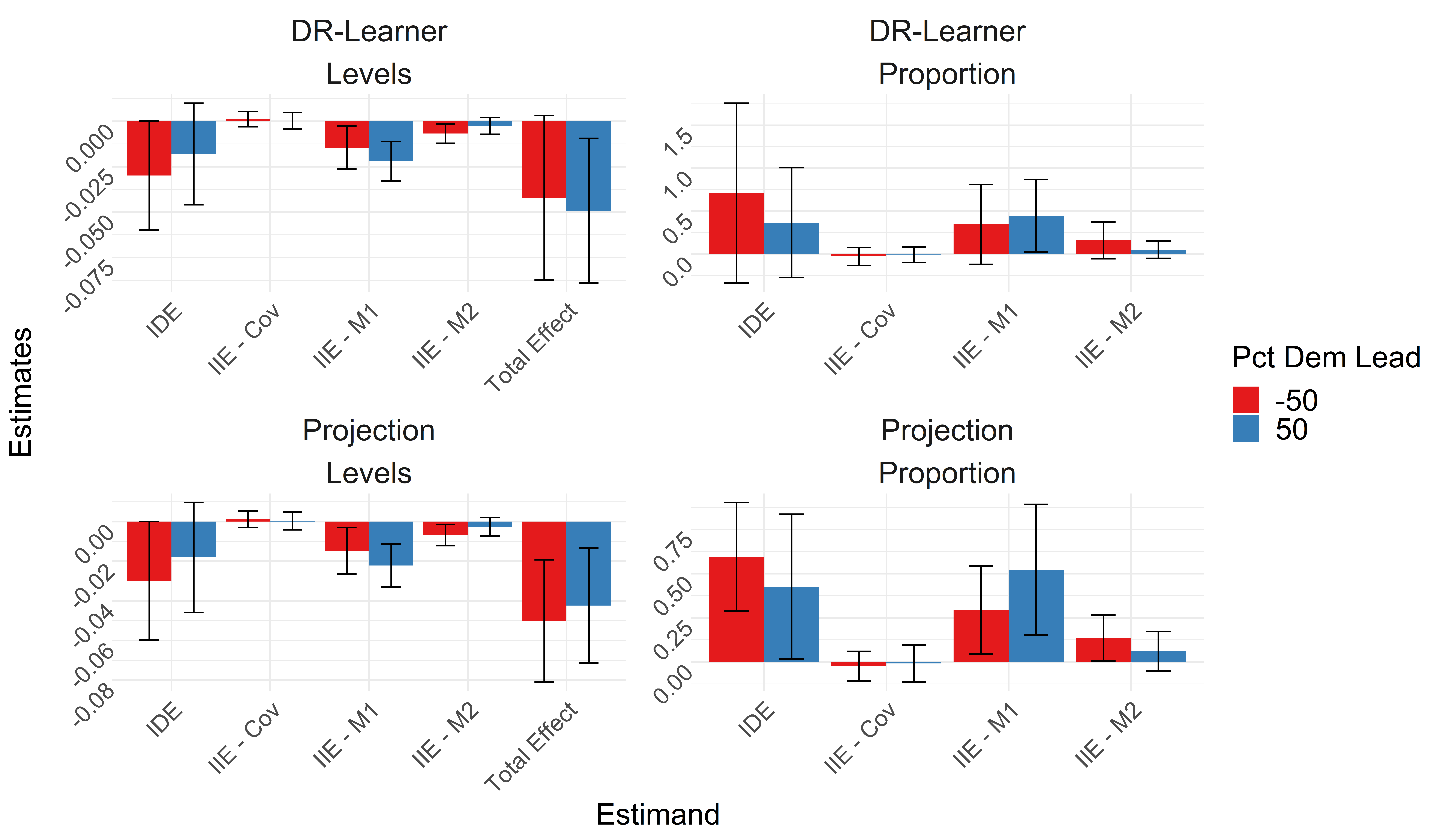}
\end{center}
\end{figure}

The total effect estimates are comparable in the Trump counties relative to Biden counties; however, the projection estimates are slightly lower in Trump relative to Biden counties while the DR-Learner suggests these effects may be slightly higher.\footnote{The uncertainty estimates for the proportion mediated are obtained via the delta method, and the DR-Learner uncertainty estimates are only valid assuming positive dependence between the errors in the models. The uncertainty estimates for the DR-Learner also do not account for post-selection inference and are therefore likely to be anti-conservative.} On the other hand, the effects via worries about health and social isolation are nearly identical for both the the projection-estimator and DR-Learner. As seen in Figure~\ref{fig:application1} in Appendix~\ref{app:application}, the chosen smoothing parameter ends up essentially fitting a linear model for all functions other than the total effect. Regardless, the point estimates are consistent with our expectations, where effects via worries about health are lower in Trump counties relative to Biden counties. Meanwhile, effects via isolation appear slightly larger in Trump counties relative to Biden counties. However, all observed differences in these effects are small relative to the uncertainty estimates and we are unable to draw statistically significant conclusions. 

\subsection{Sensitivity analysis}

We conduct a sensitivity analysis for $\psi_{M_1}$ both on average and as a function of Biden's vote share. Figure~\ref{fig:bound-application} displays the results assuming \eqref{eqn:a3} and where the conditional bounds are estimated using the DR-Learner. 

We find that that our average effect estimates are robust to a $\tau$ as large as 0.05, where $\tau$ parameterizes the deviations of the unobserved counterfactual regression function to the observed regression function on the risk-ratio scale (see equation~\ref{eqn:a3}). In other words, if this ratio were less than $0.95$ ($1 - 0.05$), or greater than $1.05$ ($\frac{1}{1-0.05}$) for any value of $(x, m_1, m_2)$, our bounds would include a null effect. Our conditional effect estimates are less robust, in part due to the greater uncertainty estimates. For example, our estimates for Trump counties is robust only up to $\tau$ of 0.01 and for Biden counties is robust to $\tau$ of 0.03. 

As a point of comparison, if we assumed that no unmeasured confounding held conditional on $X$, but we failed to control for \textit{any} covariates, across all values of $(m_1, m_2)$ we would calculate a maximal $\tau = 0.95$. To be precise, we estimate that:

\begin{align*}
\frac{1}{1-0.95} &\ge \frac{\mathbb{E}[Y^{m_1m_2} \mid A = a, M_1 \ne m_1, M_2 \ne m_2]}{\mathbb{E}[Y^{m_1m_2} \mid A = a, M_1 = m_1, M_2 = M_2]} \\
&= \frac{\mathbb{E}[\mathbb{E}[Y \mid A = a, M_1 = m_1, M_2 = m_2, X]]}{\mathbb{E}[Y \mid A = a, M_1 = m_1, M_2 = M_2]}
\ge (1 - 0.95)
\end{align*}
where the equality holds via assuming no unmeasured confounding conditional on $X$ and consistency. Therefore, a set of unmeasured confounders with comparable association with the potential outcome regression would easily explain away our estimated effects, as we find that we would be unable to rule out a null effect at $\tau = 0.05$. In other words, our significant effect would disappear if there were some unmeasured confounder $U$ that were at least approximately 5\% (0.05 / 0.95) as associative with the outcome as our entire observed covariate set. However, this comparison might be best thought of as a ``worst-case scenario,'' as we estimate $\tau$ using all measured confounders and our covariate set is quite rich. Interesting future work would be to estimate different values of $\tau$ under different covariate subsets to obtain possibly less conservative ranges of $\tau$. Sensitivity results for the remaining parameters are available in Appendix~\ref{app:application}.

\begin{figure}[H]
\begin{center}
    \mycaption{Bounds for application}{Bounds for average and conditional interventional indirect effects via social isolation as a function of $\tau$}\label{fig:bound-application}
    \includegraphics[scale=0.35]{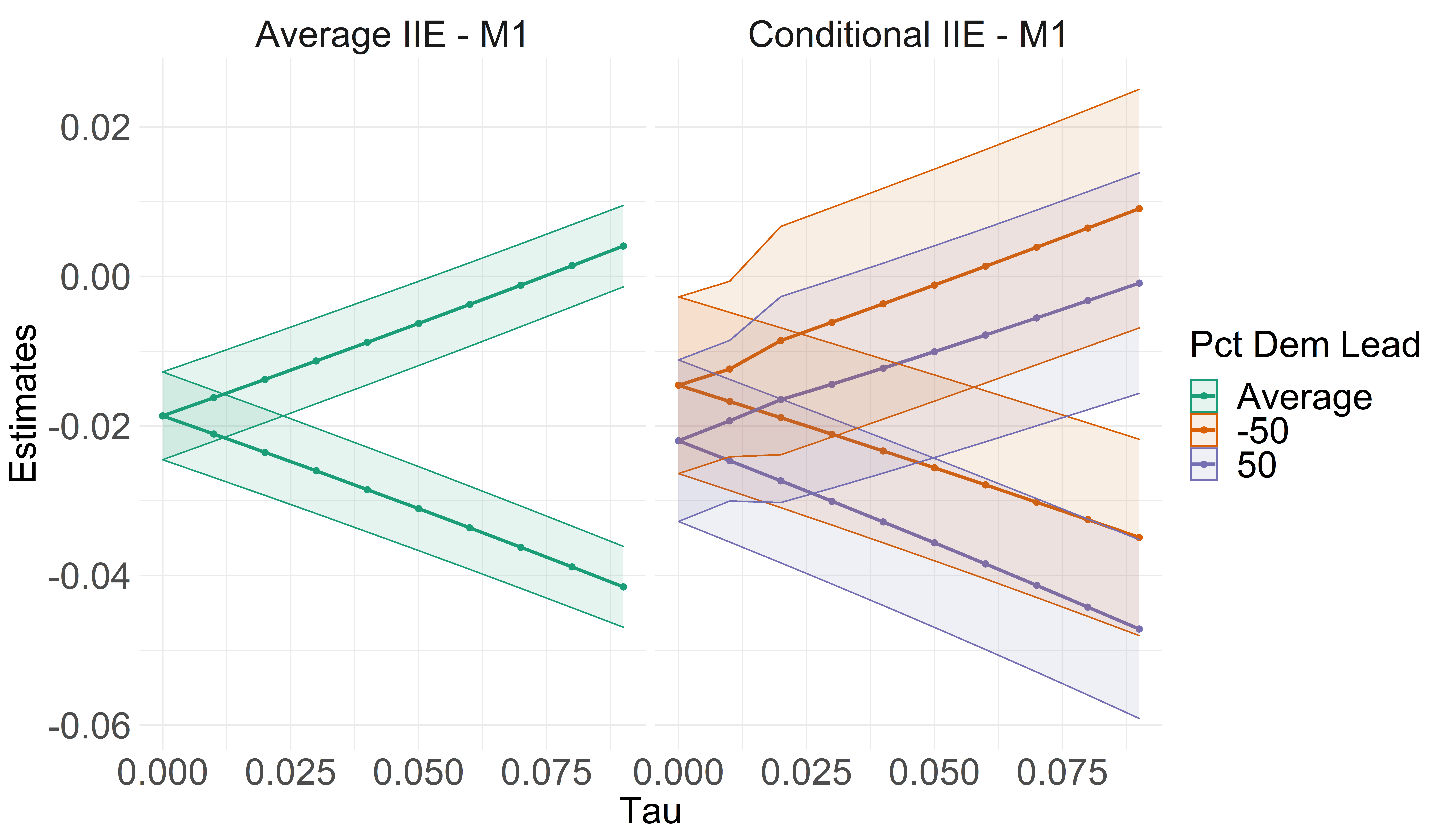}
\end{center}
\end{figure}

\section{Discussion}\label{sec:6}

We propose two methods for estimating conditional average interventional indirect effects: a semi-parametric projection-based approach and a fully non-parametric approach. These procedures are conceptually simple: regress an estimate of the uncentered influence function for the average parameter onto the desired covariates. The projection-based estimator uses a parametric regression model, and therefore targets a projection of the CIIE, while the DR-Learner uses a fully non-parametric for this regression, and therefore targets the CIIE itself. Our primary contribution is to establish the conditions where the convergence rates of these estimators is equivalent to that of an oracle regression of the true influence function onto these same models. As with estimating the CATE, the error of these estimators is a function of the product of errors in the nuisance estimation. However, unlike the CATE, we must consider the sums of several products of nuisance functions, which is in general a function of the cardinality of the joint mediators. While our discussion focused primarily on estimating the effect via $M_1$, this approach can be extended to estimate other interventional effects, mediated effects, and likely a broad class of causal estimands.

As a second contribution we propose a sensitivity analysis for the conditional effects that allows for mediator-outcome confounding. While the resulting bounds may be quite wide in practice, they make only weak assumptions on the underlying confounding mechanisms. Moreover, if one is willing to make stronger assumptions on the selection mechanism, tighter bounds can be obtained using a slight variant of our approach. We propose a general approach to estimating these bounds using the projection estimators or DR-Learner, where our results are again not tied to any particular estimation method. Our methods also easily extend to estimating bounds on the \textit{average} effects, allowing for root-n consistent and asymptotically normal estimates under some standard conditions. 

Our proposed methods have several limitations: first, we only consider two discrete mediators and a binary treatment. However, we could broaden this general approach for more complex settings. For example, we could likely allow for several mediators by regressing the corresponding influence functions derived by Benkeser and Ran (2021) onto $V$. Similarly, we could likely extend our results to allow for continuous mediators. This would require additional assumptions, including, for example, the boundedness of the joint mediator density. A complete treatment of this topic would be an interesting area for future research. On the other hand, allowing for a continuous treatment would be a more challenging problem as the causal estimands themselves would have to be redefined, and an influence-function for the average effect does not exist. A second limitation of our proposed method is that our sensitivity analysis provides bounds that may be conservative. This is in part a function of the fact that the methods we considered are all with respect to worst-case scenarios that may occur infrequently in practice. A third limitation is that we do not study any number of other possible non-parametric estimation methods, such as an extension of the R-Learner proposed by \cite{nie2021quasi}. Finally, we do not explore the minimax optimal rates for CIIE estimation or propose estimators that might achieve these rates. Valuable future work could explore any of these questions. \\

\noindent\textbf{Acknowledgments:} The authors would like to thank Amelia Haviland for helpful discussions as this work developed. The authors would also like to thank the two anonymous reviews and the Associate Editor for helpful comments, questions, and suggestions that improved the quality of this manuscript. \hfill \break

\noindent\textbf{Funding information:} Authors state no funding involved. \hfill \break

\noindent\textbf{Conflict of interest:} Authors state no conflict of interest. \hfill \break

\noindent\textbf{Data availability:} The data used in the application for this study are not publicly available, but are available on request at \href{https://dataforgood.facebook.com/dfg/docs/covid-19-trends-and-impact-survey-request-for-data-access}{this link}. Code is available online at \url{https://github.com/mrubinst757/ciie}.

\bibliographystyle{imsart-nameyear} 
\bibliography{research.bib}       

\newpage

\appendix

\section{Other definitions, theorems, and corollaries}\label{sec:othertheorems}

\subsection{Definitions and theorems from \cite{kennedy2022}}\label{appendix:definitions}

For completeness we reproduce key results from \cite{kennedy2022} that we use in our proofs. In particular, we reproduce the definition of estimator stability, a theorem that linear smoothers are stable (Theorem 1 from \cite{kennedy2022}), and a proposition about the asymptotics of pseudo-outcome regression where the second-stage regression is a linear smoother (Proposition 1 from \cite{kennedy2022}). The proofs can be found in \cite{kennedy2022}.

\begin{definition}[Stability]\label{defn:stability}
Suppose $D^n = (Z_{01}, ..., Z_{0n})$ and $Z^n = (Z_1, ..., Z_n)$ are independent test and training samples, respectively, with covariates $X_i \subset Z_i$ and $V_i \subseteq X_i$. Let 

\begin{enumerate}
    \item $\hat{f}(z) = \hat{f}(z; D^n)$ be an estimate of the function $f$ using training data $D^n$
    \item $\hat{b}(x) = \mathbb{E}\{\hat{f}(Z) - f(Z) \mid X = x, D^n\}$ be the conditional bias of the estimator $\hat{f}$
    \item $\hat{\mathbb{E}}_n(Y \mid V = v)$ denote a generic regression estimator that regresses outcomes $(Y_1, ..., Y_n)$ on the covariates $(V_1, ..., V_n)$ in the test sample $Z^n$. 
\end{enumerate}

Then the regression is defined as ``stable'' with respect to distance metric $d$ if 

\begin{align}
    \frac{\hat{\mathbb{E}}_n\{\hat{f}(Z) \mid V = v\} - \hat{\mathbb{E}}_n\{f(Z) \mid V = v\} - \hat{\mathbb{E}}_n\{\hat{b}(X) \mid V = v\}}{\sqrt{\mathbb{E}\left(\left[\hat{\mathbb{E}}_n\{f(Z) \mid V = v\} - \mathbb{E}\{f(Z) \mid V = v\}\right]^2 \right)}} \to^p 0
\end{align}

whenever $d(\hat{f}, f) \to^p 0$.
\end{definition}

\begin{theorem}[Theorem 1 from \cite{kennedy2022}]
Linear smoothers of the form $\hat{\mathbb{E}}_n(\hat{f}(Z) \mid V = v) = \sum_{i=1}^nw(v; V^n)\hat{f}(Z_i)$ are stable with respect to the distance

\begin{align}
    d(\hat{f}, f) = \|\hat{f} - f\|_{w^2} = \sum_{i=1}^n\frac{w_i(v; V^n)^2}{\sum_{j=1}^n w_j(v; V^n)^2}\int \left\{\hat{f}(z) - f(z)\right\}^2dP(z \mid V_i)
\end{align}

whenever $1/\|\sigma\|_{w^2} = \mathcal{O}_p(1)$ for $\sigma^2(v) = Var(f(Z) \mid V = v)$.
\end{theorem}

\begin{proposition}[Proposition 2 from \cite{kennedy2022}]\label{kennedyprop2}
If $\hat{b}(x) = \hat{b}_1(x)\hat{b}_2(x)$, and $\hat{\mathbb{E}}_n$ is a linear smoother with $\sum_i \lvert w_i(v; V^n) \rvert = \mathcal{O}_p(c_n)$, then

\begin{align}
    \hat{\mathbb{E}}_n\{\hat{b}(X) \mid V = v\} = \mathcal{O}_p(c_n \|\hat{b}_1\|_{w,p}\|\hat{b}_2\|_{w,q}) 
\end{align}

for the norm $\|f\|_{w,p} = \left[\sum_i\left\{\frac{\lvert w_i(v; V^n) \rvert}{\sum_j\lvert w_j(v; V^n)\rvert} \right\}\lvert f(X_i)^p\rvert \right]^{1/p}$ for $1/p + 1/q = 1$ and $p, q > 1$.
\end{proposition}

\newpage

\subsection{Other theorems and corollaries}\label{ssec:othertheorems}

This section provides a statement of Theorem~\ref{theorem4}, which provides an expression for the efficient influence function of the bounds on $\psi_{M_1}$. We then state the implied one-step estimator for the bounds on the average effect, and establish the conditions for the root-n consistency and asymptotic normality of these estimates in Theorem~\ref{theorem5}. We then extend the results of Theorem~\ref{theorem3} applied to the projection estimator of the bounds in Corollary~\ref{corollary4}, and the results of Corollary~\ref{corollary-drlearner} for the DR-Learner in Corollary~\ref{corollary5}. Proofs are included in Section~\ref{app:proofs}.

\begin{theorem}[Efficient influence function of bounds]\label{theorem4}
Define 

\begin{align*}
    \Gamma_{1,a}  &= \mathbb{E}\left\{\sum_{m_1,m_2}\mu_a(m_1,m_2,X)p(m_1,m_2\mid a, X)p(m_1\mid a, X)p(m_2 \mid a', X)\right\} = \mathbb{E}\left\{\zeta_{1,a}(X)]\right\} \\
    \Gamma_{1,a'} &= \mathbb{E}\left\{\sum_{m_1,m_2}\mu_a(m_1,m_2,X)p(m_1,m_2\mid a, X)p(m_1\mid a', X)p(m_2 \mid a', X)\right\} = \mathbb{E}\left\{\zeta_{1,a'}(X)]\right\} \\
    \Gamma_{2,a}  &= \mathbb{E}\left\{\sum_{m_2,m_2}p(m_1,m_2\mid a, X)p(m_1\mid a, X)p(m_2 \mid a',X)\right\} = \mathbb{E}\left\{\zeta_{2,a}(X)]\right\} \\
    \Gamma_{2,a'} &= \mathbb{E}\left\{\sum_{m_2,m_2}p(m_1,m_2\mid a, X)p(m_1\mid a', X)p(m_2 \mid a',X)\right\} = \mathbb{E}\left\{\zeta_{2,a'}(X)]\right\} \\
\end{align*}

Let $\phi_{1,a}(Z)$, $\phi_{1,a'}(Z)$, $\phi_{2,a}(Z)$, and $\phi_{2,a'}(Z)$ be the uncentered efficient influence functions for $\Gamma_{1,a}$, $\Gamma_{1,a'}$, $\Gamma_{2,a}$, and $\Gamma_{2,a'}$, respectively. The efficient influence function for each quantity is given below:

\begin{align*}
    \phi_{1,a}(Z) &= \frac{\mathds{1}(A = a)}{\pi_a(X)}[Yp(M_1 \mid a, X)p(M_2 \mid a', X) - \zeta_{1,a}(X)] \\
    &+\frac{\mathds{1}(A = a)}{\pi_a(X)}\left[\sum_{m_2}\mu_a(M_1, m_2, X)p(M_1, m_2 \mid, a, X)p(m_2 \mid a', X) - \zeta_{1,a}(X)\right]\\
    &+ \frac{\mathds{1}(A = a')}{\pi_{a'}(X)}\left[\sum_{m_1}\mu_a(m_1, M_2, X)p(m_1, M_2 \mid a, X)p(m_1 \mid a, X) - \zeta_{1,a}(X)\right] + \zeta_{1,a}(X) \\
    \phi_{1,a'}(Z) &= \frac{\mathds{1}(A = a)}{\pi_a(X)}[Yp(M_1 \mid a', X)p(M_2 \mid a', X) - \zeta_{1,a'}(X)] + \\
    &+\frac{\mathds{1}(A = a')}{\pi_{a'}(X)}\left[\sum_{m_2}\mu_a(M_1, m_2, X)p(M_1, m_2 \mid a, X)p(m_2 \mid a', X) - \zeta_{1,a'}(X)\right] \\
    &+ \frac{\mathds{1}(A = a')}{\pi_{a'}(X)}\left[\sum_{m_1}\mu_a(m_1, M_2, X)p(m_1, M_2 \mid a, X)p(m_1 \mid a', X) - \zeta_{1,a'}(X)\right] + \zeta_{1,a'}(X) \\
    \phi_{2,a}(Z) &= \frac{\mathds{1}(A = a)}{\pi_a(X)}[p(M_1 \mid a, X)p(M_2 \mid a', X) - \zeta_{2,a}(X)] \\
    &+ \frac{\mathds{1}(A = a)}{\pi_a(X)}\left[\sum_{m_2}p(m_2 \mid a', X)p(M_1, m_2 \mid, a, x) - \zeta_{2,a}(X)\right] \\
    &+ \frac{\mathds{1}(A = a')}{\pi_{a'}(X)}\left[\sum_{m_1}p(m_1 \mid a, X)p(m_1, M_2 \mid a, x) - \zeta_{2,a}(X)\right] + \zeta_{2,a}(X) \\
    \phi_{2,a'}(Z) &= \frac{\mathds{1}(A = a)}{\pi_a(X)}\left[p(M_1 \mid a', X)p(M_2 \mid a', X) - \zeta_{2,a'}(X)\right] \\
    &+ \frac{\mathds{1}(A = a')}{\pi_{a'}(X)}\left[\sum_{m_2}p(m_2 \mid a', X)p(M_1, m_2 \mid, a, x) - \zeta_{2,a'}(X)\right] \\
    &+ \frac{\mathds{1}(A = a')}{\pi_{a'}(X)}\left[\sum_{m_1}p(m_1 \mid a', X)p(m_1, M_2 \mid a, x) - \zeta_{2,a'}(X)\right] + \zeta_{2,a'}(X) 
\end{align*}

Recalling the expressions for the bounds in equations (\ref{eqn:upperboundx})-(\ref{eqn:lowerboundx}), the (uncentered) efficient influence function for the parameters $\psi_{M_1, lb}(\tau)$ and $\psi_{M_1, ub}(\tau)$ are therefore equal to

\begin{align*}
    \Xi_{ub}(Z; \eta, \tau) &= \varphi(Z; \eta) + \varphi_a(Z; \eta)f_u(\tau)c_u - \varphi_{a'}(Z)f_l(\tau)c_l + t_uf_u(\tau) - t_lf_l(\tau) \\
    &- f_u(\tau)[c_u\phi_{1, a} + t_u\phi_{1, a'}] + f_l(\tau)[c_l\phi_{1,a'} + t_l\phi_{2,a'}] \\
    \Xi_{lb}(Z; \eta, \tau) &= \varphi(Z; \eta) + \varphi_a(Z; \eta)f_l(\tau)c_l - \varphi_{a'}(Z)f_u(\tau)c_u + t_lf_l(\tau) - t_uf_u(\tau) \\
    &- f_l(\tau)[c_l\phi_{1, a} + t_l\phi_{1, a'}] + f_u(\tau)[c_u\phi_{1,a'} + \phi_{2,a'}] \\
\end{align*}

\end{theorem}

    
\begin{theorem}[Asymptotic normality of estimates of bounds]\label{theorem5}
Given $\tau$, consider the one-step estimator of the bounds:

\begin{align}\label{eqn:osebounds}
    [\hat{\psi}_{M_1, lb}(\tau), \hat{\psi}_{M_1, ub}(\tau)]^\top &= \left[\mathbb{P}_n[\Xi_{lb}(Z; \hat{\eta}, \tau)], \mathbb{P}_n[\Xi_{ub}(Z; \hat{\eta}, \tau)]\right]^\top ]]
    &= \mathbb{P}_n \Xi(Z; \hat{\eta}, \tau) = \hat{\psi}_{M_1,b}
\end{align}

Let $\hat{\psi}_{M_1, b}$ denote the vector of estimates in \eqref{eqn:osebounds}. Assume that

\begin{enumerate}
    \item Nuisance parameters $\eta$ are estimated using sample-splitting
    \item $\|\hat{\Xi} - \Xi\|^2 = o_p(1)$
    \item $\|\hat{\eta} - \eta\| = o_p(n^{-1/4})$
\end{enumerate}

We then obtain the following limiting distribution for the vector of estimates:

\begin{align}
    \sqrt{n}\mathbb{P}_n[\hat{\psi}_{M_1, b}(\tau) - \psi_{M_1, b}(\tau)] \to^d \mathcal{N}(0, \mathbb{E}[\Xi\Xi^\top])
\end{align}

\end{theorem}



\begin{corollary}[Projection estimator of bounds]\label{corollary4}
    For fixed parameter $\tau$, consider the function $f(V; \alpha)$, the moment condition,
    
    \begin{align*}
        \mathbb{E}\left[w(X)\frac{\partial f(V; \alpha)}{\partial \alpha}[\psi_{M_1, lb}(Z; \tau) - f(V; \alpha)]\right] = 0
    \end{align*}
    the associated influence curve at any $\alpha^\star$,
    
    \begin{align*}
        \phi_{lb}(Z; \alpha^\star, \eta, \tau) &= \frac{\partial f(V; \alpha)}{\partial \alpha} w(X)\{\Xi_{lb}(Z; \eta, \tau) - f(V; \alpha^\star)\}
    \end{align*}
    and the associated influence-function based estimating equation for $\hat{\alpha}$:
    
    \begin{align*}
        \mathbb{P}_n\left[\frac{\partial f(V; \hat{\alpha})}{\partial \alpha} w(X)\{\hat{\Xi}_{lb}(Z; \hat{\eta}, \tau) - f(V; \hat{\alpha})\} \right] = 0
    \end{align*}

    Consider the moment condition $\mathbb{E}[\phi_{lb}(Z; \alpha_0, \eta_0, \tau)] = 0$ evaluated at the true parameters $(\alpha_0, \eta_0)$. Now consider the estimator $\hat{\alpha}$ that satisfies $\mathbb{P}_n[\phi_{lb}(Z; \hat{\alpha}, \hat{\eta}, \tau)] = 0$, where $\hat{\eta}$ is estimated on an independent sample. Assume that:

    \begin{itemize}
        \item The function class $\{\phi_{lb}(Z; \alpha, \eta): \alpha \in \mathbb{R}^p\}$ is Donsker in $\alpha$ for any fixed $\eta$
        \item $\|\phi_{lb}(Z; \hat{\alpha}, \hat{\eta}, \tau) - \phi(Z; \alpha_0, \eta_0, \tau)\| = o_p(1)$
        \item The map $\alpha \to \mathbb{P}[\phi_{lb}(Z; \alpha, \eta, \tau)]$ is differentiable at $\alpha_0$ uniformly in the true $\eta$, with non-singular derivative matrix $\frac{\partial}{\partial \alpha}\mathbb{P}\{\phi_{lb}(Z; \alpha, \eta, \tau\}\mid_{\alpha = \alpha_0} = G(\alpha_0, \eta, \tau)$, where $G(\alpha_0, \hat{\eta},\tau) \to^p M(\alpha_0, \eta_0)$
    \end{itemize}

    Then
    
    \begin{align*}
        \hat{\beta} - \beta &= -G^{-1}[\mathbb{P}_n - P]\phi_{lb}(Z; \alpha_0, \eta_0, \tau) + \mathcal{O}_p(T_{1n} + T_{2n} + T_{3n} + T_{4n} \\
        &+ \underbrace{\|p(M_1, M_2 \mid a, X) - \hat{p}(M_1, M_2 \mid a, X)\|\|p(M_1 \mid a', X) - \hat{p}(M_1 \mid a', X)\|}_{(xi)} \\
        &+ \underbrace{\|p(M_1, M_2 \mid a, X) - \hat{p}(M_1, M_2 \mid a, X)\|\|p(M_2 \mid a', X) - \hat{p}(M_2 \mid a', X)\|}_{(xii)} \\
        &+ \underbrace{\|p(M_1, M_2 \mid a, X) - \hat{p}(M_1, M_2 \mid a, X)\|\|p(M_1 \mid a, X) - \hat{p}(M_1 \mid a, X)\|}_{(xiii)} \\
        &+ \underbrace{\|p(M_1, M_2 \mid a, X) - \hat{p}(M_1, M_2 \mid a, X)\|\|\pi_a(X) - \hat{\pi}_a(X)\|}_{(xiv)}) + o_p(1 / \sqrt{n})
    \end{align*}
    where the terms $T_{1n}$-$T_{4n}$ are defined in Theorem~\ref{theorem3}. Moreover, if $Z_n = T_{1n} + T_{2n} + T_{3n} + T_{4n} + (x)-(iv)$ are $o_p(n^{1/2})$, then
    \begin{align*}
        \sqrt{n}(\hat{\alpha} - \alpha) \to^d \mathcal{N}(0, G^{-1}\mathbb{E}[\phi_{lb}\phi_{lb}^\top]G^{-\top})
    \end{align*}
    where $G =\frac{\partial}{\partial \alpha}\mathbb{P}\{\phi_{lb}(Z; \alpha, \eta, \tau)\}_{\alpha = \alpha_0}$.
\end{corollary}

\begin{remark}\label{remark:boundproj1}
    There are 21 possible combinations of second-order terms among nuisance elements $[\pi_a(X), \mu_a(M_1, M_2, X), p(M_1 \mid a, X), p(M_1 \mid a', X), p(M_2 \mid a', X), p(M_1, M_2 \mid a, X)]$. $\hat{\psi}^{dr}_{M_1}(x)$ requires considering ten of these error products. Estimating $\hat{\psi}^{dr}_{M_1, lb}(x)$ requires considering the four additional terms noted above. 
    
    On the other hand, if the estimate of $p(M_1 \mid a, X)$ comes from an estimate of $p(M_1, M_2 \mid a, X)$ (and similarly $p(M_1 \mid a', x)$ and $p(M_2 \mid a', x)$ come from an estimate of $p(M_1, M_2 \mid a', X)$), then these expressions simplify. In that case there are only 10 total possible second-order terms. $\psi_{M_1}$ requires estimating seven of the ten and the bounds only require one additional error product that comes from (xiii), which is on the same order of the square of the L2-norm of the error of the estimated joint mediator probability $p(M_1, M_2 \mid a, X)$. 
\end{remark}

\begin{remark}\label{remark:boundproj2}
    This result is agnostic as to which of equations (\ref{eqn:a1})-(\ref{eqn:a3}) are assumed for the sensitivity analysis. While not all second-order remainder terms from $\phi_{1,a}$, $\phi_{1,a'}$, $\phi_{2,a}$, $\phi_{2,a'}$ will appear in each expression, as the constants on these functions may be zero depending on the assumption used, these same terms will appear regardless. This can be seen from the Proof of Theorem~\ref{theorem5}, where we derive the second-order terms for 
    
    \begin{align*}
        P\left[[\hat{\phi}_{1,a}, \hat{\phi}_{1,a'}, \hat{\phi}_{2,a}, \hat{\phi}_{2,a'}]^\top \right] - P\left[[\phi_{1,a}, \phi_{1,a'}, \phi_{2,a}, \phi_{2,a'}]^\top\right]
    \end{align*}
    
\end{remark}

\begin{remark}
    It is easy to show an analogous result for a projection estimator of $\psi_{M_1, ub}(v)$.
\end{remark}

\begin{corollary}[DR-Learner for bounds]\label{corollary5}
    Define Algorithm 2 as Algorithm 1, where Step 2 instead constructs the estimated pseudo-outcomes $\Xi_{lb}(Z; \hat{\eta}, \tau)$. Assume that the conditions of Corollary~\ref{corollary-drlearner} hold and that $\tau \ne 0$.
    
    Denote $\tilde{\psi}_{M_1, lb}(v; \tau)$ as an oracle estimator from a regression of $\Xi_{lb}(Z; \eta, \tau)$ onto $V$ and let $Q_n^{2\star}(v)$ denote the oracle risk at the point $V = v$. Then 
    
    \begin{align*}
    [\hat{\psi}_{M_1, lb}(\tau) - \psi_{M_1, lb}(\tau)] \lesssim \hat{\mathbb{E}}_n(\hat{q}_{lb}(X) \mid V = v) + o_p(Q^\star_n(v))       
    \end{align*}
    where
    
    \begin{align*}
        \hat{q}_{lb}(x) &= \hat{b}(x) + (\pi_a(x) - \hat{\pi}_a(x))\sum_{m_1, m_2} (p(m_1, m_2 \mid a, x) - \hat{p}(m_1, m_2 \mid a, x)) \\ 
        &+ \sum_{m_1, m_2}(p(m_1, m_2 \mid a, x) - \hat{p}(m_1, m_2 \mid a, x))(p(m_1 \mid a, x) - \hat{p}(m_1 \mid a, x)) \\
        &+ \sum_{m_1, m_2}(p(m_1, m_2 \mid a, x) - \hat{p}(m_1, m_2 \mid a, x))(p(m_1 \mid a', x) - \hat{p}(m_1 \mid a', x)) \\
        &+\sum_{m_1, m_2}(p(m_1, m_2 \mid a, x) - \hat{p}(m_1, m_2 \mid a, x))(p(m_2 \mid a', x) - \hat{p}(m_2 \mid a', x))
    \end{align*}
    
    Moreover, if $\hat{\mathbb{E}}_n$ is a linear smoother with weights $\sum_i\lvert w_i(v; V^n)\rvert = \mathcal{O}_p(a_n)$, then
    
    \begin{align*}
    \hat{\mathbb{E}}_n(\hat{q}_{lb}(X) \mid V = v) = \max_j \mathcal{O}_p\left(a_n\|\hat{q}_{lb, j1}\|_{w,2}\|\hat{q}_{lb, j2}\|_{w, 2}\right)
    \end{align*}
    where $\hat{q}_{lb, j}$ represents the $j$-th second-order term in the expression of $\hat{q}_{lb}(x)$ above.
\end{corollary}

\begin{remark}
    The comments in Remarks \ref{remark:boundproj1} and \ref{remark:boundproj2} apply here as well. The result for an estimate of $\psi_{M_1, ub}(v)$ is asymptotically equivalent. 
\end{remark}

\newpage

\section{Proofs}\label{app:proofs}

This section contains proofs of all theorems, corollaries, and propositions in the main paper and in Section~\ref{ssec:othertheorems}. We divide the proofs into two sections: the first pertaining to conditional effect estimation; the second pertaining to the bounds on $\psi_{M_1}$, both average and conditional. We defer tedious algebraic derivations to Section~\ref{app:soterms}.

\subsubsection{Conditional effect estimation}

The following proofs pertain to estimating the conditional effects. This section contains four proofs:

\begin{itemize}
    \item Proposition \ref{proposition2}: efficient influence function for the projection estimator
    \item Theorem \ref{theorem3}: root-n consistency and asymptotic normality of projection estimator
    \item Corollaries \ref{corollary-drlearner} and \ref{corollary-extension}: DR-Learner rates of convergence
\end{itemize} 

\begin{proof}[Proof of Proposition~\ref{proposition2}]
We consider the functional $\Psi(\beta; \mathbb{P})$ defined in \eqref{eqn:projestimator} and derive its efficient influence curve. This result follows directly from the result previously derived in \cite{cuellar2020non}, though we reproduce it here for completeness. We treat the data as discrete throughout to simplify the derivation. Let $h(x) = \frac{\partial g(v; \beta)}{\partial \beta}w(x)$. We first define our parameter:

\begin{align*}
\Psi(\beta; P) &= \mathbb{E}[h(X)[\psi_{M_1}(X) - g(V; \beta)]] \\
&= \sum_x h(x)[\psi_{M_1}(x) - g(v; \beta)]p(x)
\end{align*}
We can use the product rule to obtain that:

\begin{align*}
\text{IF}[\Psi(\beta)] &= \sum_x h(x) \text{IF}[\psi_{M_1}(x)]p(x) + h(x)[\psi_{M_1}(x) - g(v; \beta)]\text{IF}[p(x)] \\
&=\sum_x h(x) \text{IF}[\psi_{M_1}(x)]p(x) + h(x)[\psi_{M_1}(x) - g(v; \beta)][\mathds{1}(X = x) - p(x)] \\
&= h(X)[\varphi(Z) - g(V; \beta)] - \Psi(\beta; P) \\
&= \phi(Z) - \Psi(\beta; P)
\end{align*}
where $\varphi(Z)$ takes the form defined in \eqref{eqn:eifm1}. 
\end{proof}

\begin{proof}[Proof of Theorem \ref{theorem3}]
This result follows directly from Lemma 3 from \cite{kennedy2021semiparametric}. We further make the assumptions outlined in Theorem \ref{theorem3}, which correspond to the same assumptions of Lemma 3. The following term drives our ability to obtain root-n consistency and asymptotically normal estimates:

\begin{align*}
    \mathbb{P}[\phi(Z; \beta, \hat{\eta}) - \phi(Z; \beta, \eta)]
\end{align*}
This is simply equal to:

\begin{align*}
    \mathbb{P}[\varphi(Z; \hat{\eta}) - \varphi(Z; \eta)]
\end{align*}
We have shown that this expression is bounded by the product of errors of the nuisance functions (these products are reproduced and used to prove Corollary~\ref{corollary-drlearner}). These terms then characterize the rate $R_n = T_{1n} + T_{2n} + T_{3n} + T_{4n}$ in the $\mathcal{O}_p(R_n)$ term in Lemma 3. An immediate consequence is that if $R_n = o_p(n^{-1/4})$ then this term is of smaller order and thus

\begin{align*}
    \hat{\beta} - \beta = -M^{-1}[\mathbb{P}_n - \mathbb{P}]\phi(Z; \beta, \eta) + o_p(n^{-1/2})
\end{align*}
Because this expression is simply a centered sample mean of a fixed function, the asymptotic normality of the difference scaled by $\sqrt{n}$ follows via the Central Limit Theorem with variance equal to $M^{-1}\mathbb{E}[\phi\phi^\top]M^{-\top}$.
\end{proof}

\begin{proof}[Proof of Corollary~\ref{corollary-drlearner}]
Let $\varphi(Z; \eta)$ be the (uncentered) efficient influence function for $\psi_{M_1}$. Let $\hat{\psi}_{M_1} = n^{-1}\sum_{i=1}^n\varphi(Z; \hat{\eta})$. Recall that $P\{f(Z)\} = \int f(z) dP(z \mid D_1^n)$ where $D_1^n$ is an independent sample of size $n$. To ease notation we omit the conditioning on the training samples for the remainder of this proof; however, the consequence is that we derive the results thinking of the estimated functions as fixed and not dependent on the data. We show in Section~\ref{app:soterms} that $P[\hat{\psi}_{M_1} - \psi_{M_1}]$ decomposes into the following form:\footnote{\cite{benkeser2021nonparametric} show a similar derivation; our decomposition largely agrees though we find a few additional second-order terms. They do not provide a proof of their result in their paper, and we therefore cannot compare our derivations; however, we include ours here for completeness. The differences between our decomposition and theirs do not substantively affect our results.}

\begin{align}
    \nonumber & P[\hat{\psi}_{M_1} - \psi_{M_1}] \\
    \label{eqn:m1decomp.1}&= \mathbb{E}\left[\frac{\pi_a(X)}{\hat{\pi}_a(X)}\sum_{m_1, m_2}\frac{(\hat{p}(m_1 \mid a, X) - \hat{p}(m_1 \mid a', X))\hat{p}(m_2 \mid a', X)}{\hat{p}(m_1, m_2 \mid a', X)}\right.\\
    \nonumber&\left.\times \left(p(m_1, m_2 \mid a, X) - \hat{p}(m_1, m_2 \mid a, X))(\mu_a(m_1, m_2, X) - \hat{\mu}_a(m_1, m_2, X)) \right)\right] \\
    \label{eqn:m1decomp.2}&+ \mathbb{E}\left[\left(\frac{\pi_a(X) - \hat{\pi}_a(X)}{\hat{\pi}_a(X)}\right)\sum_{m_1, m_2}\hat{p}(m_2 \mid a', X)[\hat{p}(m_1 \mid a, X) - \hat{p}(m_1 \mid a', X)] \right.\\
    \nonumber&\left.\times(\mu_a(m_1, m_2, X) - \hat{\mu}_a(m_1, m_2, X))\right] \\
    \label{eqn:m1decomp.3}&+ \mathbb{E}\left[\left(\frac{\pi_a(X) - \hat{\pi}_a(X)}{\hat{\pi}_a(X)}\right)\sum_{m_1, m_2}\hat{\mu}_a(m_1, m_2, X)\hat{p}(m_2 \mid a', X)[(p(m_1 \mid a, X) - \hat{p}(m_1 \mid a, X)) \right.\\
    \nonumber&\left.+ (p(m_1 \mid a', X) - \hat{p}(m_1 \mid a', X))\right] \\
    \label{eqn:m1decomp.4}&+ \mathbb{E}\left[\left(\frac{\pi_{a'}(X) - \pi_{a'}(X)}{\hat{\pi}_{a'}(X)}\right)\sum_{m_1, m_2}\hat{\mu}_a(m_1, m_2, X)\right.\\
    \nonumber&\left.\times[\hat{p}(m_1 \mid a, X) - \hat{p}(m_1 \mid a', X)][p(m_2 \mid a, X) - \hat{p}(m_2 \mid a', X)\right]\\
    \label{eqn:m1decomp.5}&- \mathbb{E}\left[\sum_{m_1,m_2}\hat{p}(m_2 \mid a', X)(\mu_a(m_1, m_2, X) - \hat{\mu}_a(m_1, m_2, X))\right.\\
    \nonumber&\left.\times((p(m_1 \mid a, X) - \hat{p}(m_1 \mid a, X)) - (p(m_1 \mid a', X) - \hat{p}(m_1 \mid a', X))\right] \\
    \label{eqn:m1decomp.6}&- \mathbb{E}\left[\sum_{m_1, m_2}\hat{\mu}_a(m_1, m_2, X)[p(m_2 \mid a', X) - \hat{p}(m_2 \mid a', X)] \right. \\
    \nonumber&\left.\times ((p(m_1 \mid a, X) - \hat{p}(m_1 \mid a, X)) - (p(m_1 \mid a', X) - \hat{p}(m_1 \mid a', X))\right] \\
    \label{eqn:m1decomp.7}&- \mathbb{E}\left[\sum_{m_1, m_2}(p(m_1 \mid a, X) - p(m_1 \mid a', X))\right. \\
    \nonumber&\left.\times(\mu_a(m_1, m_2, X) - \hat{\mu}_a(m_1, m_2, X))(p(m_2 \mid a', X) - \hat{p}(m_2 \mid a', X) \right]
\end{align}

We use this result to bound the error of the DR-Learner at a given point $V = v$. We first define the conditional bias at $X = x$ as $\hat{b}^\star(x) = \mathbb{E}\{\hat{\varphi}(Z) - \varphi(Z) \mid X = x, D_1^n\}$. Following similar logic we obtain the following expression for $\hat{b}^\star(x)$:

\begin{align*}
    &\hat{b}^\star(x) = \frac{\pi_a(x)}{\hat{\pi}_a(x)}\left[\sum_{m_1, m_2}\frac{(\hat{p}(m_1 \mid a, x) - \hat{p}(m_1 \mid a', x))\hat{p}(m_2 \mid a', x)}{\hat{p}(m_1, m_2 \mid a', x)}\right.\\
    \nonumber&\left.\times \left(p(m_1, m_2 \mid a, x) - \hat{p}(m_1, m_2 \mid a, x))(\mu_a(m_1, m_2, x) - \hat{\mu}_a(m_1, m_2, x)) \right)\right] \\
    &+ \left(\frac{\pi_a(x) - \hat{\pi}_a(x)}{\hat{\pi}_a(x)}\right)\left[\sum_{m_1, m_2}\hat{p}(m_2 \mid a', x)[\hat{p}(m_1 \mid a, x) - \hat{p}(m_1 \mid a', x)] \right.\\
    \nonumber&\left.\times(\mu_a(m_1, m_2, x) - \hat{\mu}_a(m_1, m_2, x))\right] \\
    &+ \left[\left(\frac{\pi_a(x) - \hat{\pi}_a(x)}{\hat{\pi}_a(x)}\right)\sum_{m_1, m_2}\hat{\mu}_a(m_1, m_2, x)\hat{p}(m_2 \mid a', x)[p(m_1 \mid a, x) - \hat{p}(m_1 \mid a, x) \right.\\
    &\left.+ (p(m_1 \mid a', x) - \hat{p}(m_1 \mid a', x))\right] \\
    &+ \left(\frac{\pi_{a'}(x) - \hat{\pi}_{a'}(x)}{\hat{\pi}_{a'}(x)}\right)\left[\sum_{m_1, m_2}\hat{\mu}_a(m_1, m_2, x)\right.\\
    &\left.\times[\hat{p}(m_1 \mid a, x) - \hat{p}(m_1 \mid a', x)][p(m_2 \mid a, x) - \hat{p}(m_2 \mid a', x)\right]\\
    &- \left[\sum_{m_1,m_2}\hat{p}(m_2 \mid a', x)(\mu_a(m_1, m_2, x) - \hat{\mu}_a(m_1, m_2, x))\right.\\
    \nonumber&\left.\times(p(m_1 \mid a, x) - \hat{p}(m_1 \mid a, x) - (p(m_1 \mid a', x) - \hat{p}(m_1 \mid a', x))\right] \\
    &- \left[\sum_{m_1, m_2}\hat{\mu}_a(m_1, m_2, x)[p(m_2 \mid a', x) - \hat{p}(m_2 \mid a', x)] \right. \\
    \nonumber&\left.\times (p(m_1 \mid a, x) - \hat{p}(m_1 \mid a, x) - (p(m_1 \mid a', x) - \hat{p}(m_1 \mid a', x))\right] \\
    &- \left[\sum_{m_1, m_2}(p(m_1 \mid a, x) - p(m_1 \mid a', x))\times(\mu_a(m_1, m_2, x) - \hat{\mu}_a(m_1, m_2, x))(p(m_2 \mid a', x) - \hat{p}(m_2 \mid a', x) \right]
\end{align*}

The expression for $\hat{b}(x)$ in Corollary~\ref{corollary-drlearner} follows by the boundedness of the propensity-scores, the joint mediator probabilities, and their estimates. The result then follows from the definition of estimator stability, defined in \eqref{defn:stability}.
\end{proof}

\begin{proof}[Proof of Corollary~\ref{corollary-extension}]
    This result follows directly from Proposition~\ref{kennedyprop2} from \cite{kennedy2022}.
\end{proof}

\subsubsection{Bounds on interventional indirect effect}

This section contains proofs pertaining to estimating the proposed bounds on $\psi_{M_1}$.

\begin{proof}[Proof of Proposition~\ref{prop:bound-mu}]
Using the law of iterated expectations we obtain:

\begin{align}\label{eqn:iie}
    &\mathbb{E}[Y^{m_1m_2} \mid a, x] - \mu_a(m_1, m_2, x) \\
    \nonumber&= \sum_{m_1' \ne m_1, m_2' \ne m_2}\underbrace{[\mu_{am_1m_2}^\star(m_1', m_2', x) - \mu_a(m_1, m_2, x)]}_{(i)}p(m_1', m_2' \mid a, x)
\end{align}

The Proposition follows by replacing the term $(i)$ with $b_l(\cdot)$ or $b_u(\cdot)$.
\end{proof}

\begin{proof}[Proof of Proposition~\ref{prop:bounds-psi-m1}]
Recall that we assume throughout that $\tau(x)$ is valid for all values of $(m_1, m_2)$ at $X = x$. The assumptions in (\ref{eqn:a1})-(\ref{eqn:a3}) imply:

\begin{align*}
    &(\ref{eqn:a1}) \implies b_l(\cdot) = -\tau(x), \qquad \qquad b_u(\cdot) = \tau(x)\\
    &(\ref{eqn:a2}) \implies b_l(\cdot) = -\tau(x)\mu_a(\cdot), \qquad b_u(\cdot) = [1 - \mu_a(\cdot)]\tau(x) \\
    &(\ref{eqn:a3}) \implies b_l(\cdot) = -\tau(x)\mu_a(\cdot), \qquad b_u(\cdot) = \mu_a(\cdot)\tau(x) / (1-\tau(x))
\end{align*}

Notice that all functions $b_l$ and $b_u$ that take the form:

\begin{align*}
    h(\cdot) = [c\mu_a(\cdot) + t]f(\tau(x))
\end{align*}
for constants $(c, t) \in \{0, 1\}^2$. Replacing $b_l$ and $b_u$ with expressions of these form, combining this with the result from Proposition~\ref{prop:bound-mu}, and then averaging this expression over $[p(m_1 \mid a, x) - p(m_1 \mid a', x)]p(m_2 \mid a', x)$ for all values of $(m_1, m_2)$ gives the desired result. 
\end{proof}

\begin{proof}[Proof of Theorem~\ref{theorem4}]
We focus on the terms $\Gamma_{1, a}$ and $\Gamma_{2, a}$ noting that the other derivations follow analogously. We treat the data as discrete throughout to simplify the derivations (see, e.g., \cite{kennedy2022eifs}).

\begin{align*}
    \Gamma_{2,a} &= \sum_{x, m_1, m_2}p(m_1, m_2 \mid a, x)p(m_1 \mid a, x)p(m_2 \mid a', x)p(x) \\
    \Gamma_{1,a} &= \sum_{x, m_1, m_2}\mu_a(x, m_1, m_2) p(m_1, m_2 \mid a, x)p(m_1 \mid a, x)p(m_2 \mid a', x)p(x) 
\end{align*}

We can use the chain rule to obtain the result:

\begin{align*}
    \text{IF}(\Gamma_{2,a}) &= \sum_{x, m_1, m_2}\text{IF}[p(m_1, m_2 \mid a, x)]p(m_1 \mid a, x)p(m_2 \mid a', x)p(x) \\ 
    &+ p(m_1, m_2 \mid a, x)\text{IF}[p(m_1 \mid a, x)]p(m_2 \mid a', x)p(x) \\
    &+ p(m_1, m_2 \mid a, x)p(m_1 \mid a, x)\text{IF}[p(m_2 \mid a', x)]p(x) \\
    &+ p(m_1, m_2 \mid a, x)p(m_1 \mid a, x)p(m_2 \mid a', x)\text{IF}[p(x)]  \\
    &= \sum_{x, m_1, m_2}\frac{\mathds{1}(A = a, X = x)}{p(a \mid x)}[M_1M_2 - p(m_1, m_2 \mid a, x)]p(m_1 \mid a, x)p(m_2 \mid a', x) \\
    &+\frac{\mathds{1}(A = a, X = x)}{p(a \mid x)}[M_1 - p(m_1 \mid a, x)]p(m_1, m_2 \mid a, x)p(m_2 \mid a', x) \\
    &+ \frac{\mathds{1}(A = a', X = x)}{p(a' \mid x)}[M_2 - p(m_2 \mid a', x)]p(m_1 \mid a', x)p(m_1, m_2 \mid a, x) \\
    &+ p(m_1, m_2 \mid a, x)p(m_1 \mid a, x)p(m_2 \mid a', x)[\mathds{1}(X = x) - p(x)] \\
    &= [p(M_1 \mid a, X)p(M_2 \mid a', X) - \zeta_{2,a}(X)] \\ 
    &+ \frac{\mathds{1}(A = a)}{\pi_a(X)}[\sum_{m_2}p(M_1, m_2 \mid a, X)p(m_2 \mid a', x) - \zeta_{2,a}(X)] \\
    &+ \frac{\mathds{1}(A = a')}{\pi_{a'}(X)}\sum_{m_1}[p(m_1, M_2 \mid a, X)p(m_1 \mid a, X) - \zeta_{2,a}(X)] + \zeta_{2,a} - \Gamma_{2,a}
\end{align*} 

We now show the derivation for $\Gamma_{1, a}$ noting that $\Gamma_{1, a} = \sum_{x, m_1, m_2}\mu_a(x, m_1, m_2)\zeta_{2,a}(x)p(x)$. Therefore,

\begin{align*}
    \text{IF}(\Gamma_{1, a}) &= \sum_{x, m_1, m_2} \text{IF}[\mu_a(x, m_1, m_2)]\zeta_{2,a}(x)p(x) \\ 
    &+ \mu_a(x, m_1, m_2)\text{IF}[\zeta_{2,a}(x)]p(x) + \mu_a(x, m_1, m_2)\zeta_{2,a}(x)\text{IF}[p(x)] \\
    &= \sum_{x, m_1, m_2}\frac{\mathds{1}(A = a, M_1 = m_1, M_2 = m_2)}{p(m_1, m_2 \mid a, x)p(a \mid x)}[Y - \mu_a(m_1, m_2, x)]\\
    &\times[p(m_1 \mid a, x)p(m_2 \mid a', x)p(m_1, m_2 \mid a, x) - \zeta_{1, a}] \\
    &+ \mu_a(m_1, m_2, x)[\frac{\mathds{1}(A = a, X = x)}{p(a \mid x)}[M_1M_2 - p(m_1, m_2 \mid a, x)]p(m_1 \mid a, x)p(m_2 \mid a', x) \\
    &+\frac{\mathds{1}(A = a, X = x)}{p(a \mid x)}[M_1 - p(m_1 \mid a, x)]p(m_1, m_2 \mid a, x)p(m_2 \mid a', x) \\
    &+ \frac{\mathds{1}(A = a', X = x)}{p(a' \mid x)}[M_2 - p(m_2 \mid a', x)]p(m_1 \mid a', x)p(m_1, m_2 \mid a, x) \\
    &+ p(m_1, m_2 \mid a, x)p(m_1 \mid a, x)p(m_2 \mid a', x)[\mathds{1}(X = x) - p(x)] \\
    &+ \mu_a(m_1, m_2, x)\zeta_{2,a}(x)[\mathds{1}(X = x) - p(x)] \\
    &= \frac{\mathds{1}(A = a)}{\pi_a(X)}[Yp(M_1 \mid a, X)p(M_2 \mid a', X) - \zeta_{1,a}(X)] \\
    &+\frac{\mathds{1}(A = a)}{\pi_a(X)}\sum_{m_2}[\mu_a(M_1, m_2, X)p(M_1, m_2 \mid, a, X)p(m_2 \mid a', X) - \zeta_{1, a}(X) \\
    &+ \frac{\mathds{1}(A = a')}{\pi_{a'}(X)}\left[\sum_{m_1}\mu_a(m_1, M_2, X)p(m_1, M_2 \mid a, X)p(m_1 \mid a, X) - \zeta_{1,a}(X)\right] \\
    &+ \zeta_{1,a}(X) - \Gamma_{1, a}
\end{align*}

The derivations for $\Gamma_{1, a'}$ and $\Gamma_{2, a'}$ are analogous.
\end{proof}

\begin{proof}[Proof of Theorem~\ref{theorem5}]
Let $\hat{\Xi} = \Xi(Z; \hat{\eta}, \tau)$. We begin by considering the decomposition in \cite{kennedy2022eifs}:

\begin{align*}
    \mathbb{P}_n \hat{\Xi} - P \Xi &= \underbrace{[\mathbb{P}_n - P][\hat{\Xi}(Z) - \Xi(Z)]}_{T_1} + \underbrace{[\mathbb{P}_n - P]\Xi(Z)}_{T_2} + \underbrace{P[\hat{\Xi}(Z) - \Xi(Z)]}_{T_3} 
\end{align*}

$T_1$ is $o_p(1/\sqrt{n})$ as long as $\hat{\Xi}$ is consistent for $\Xi$ in the $L_2(\mathbb{P})$ norm -- i.e. $\|\hat{\Xi} - \Xi \| = o_p(1)$ (see, e.g., \cite{kennedy2022eifs}). $T_2$ is a sample average of a fixed function with zero mean and therefore by the Central Limit Theorem converges in distribution to $\mathcal{N}(0, \mathbb{E}[\phi(Z)\phi(Z)^\top]$.

It remains to analyze $T_3$. We proceed in typical fashion by showing that this term can be expressed as the product of errors in the nuisance estimation. We previously showed in the proof of Corollary~\ref{corollary-drlearner} that $P[\varphi(Z; \hat{\eta}) - \varphi(Z; \eta)]$ is second-order in the nuisance estimation. It suffices to show that 

\begin{align*}
P\{[\hat{\phi}_{1,a}(Z), \hat{\phi}_{1,a'}(Z), \hat{\phi}_{2,a}(Z), \hat{\phi}_{2,a'}(Z)]^\top-[\phi_{1,a}(Z), \phi_{1,a'}(Z), \phi_{2,a}(Z), \phi_{2,a'}(Z)]^\top\}
\end{align*}
is also second-order in the nuisance estimation. We show this for each term separately. We provide the final results below and include the derivations in Section~\ref{app:soterms}. 

\begin{align*}
    &P[\hat{\phi}_{1,a}(Z) - \phi_{1,a}(Z)] \\
    &= \mathbb{E}\left[\frac{\pi_a(X) }{\hat{\pi}_a(X)}\sum_{m_1,m_2}(p(m_1, m_2 \mid a, x) - \hat{p}(m_1, m_2 \mid a, x))(\mu_a(m_1, m_2, X) - \hat{\mu}_a(m_1, m_2, X))\right.\\
    &\left.\times\hat{p}(m_1 \mid a, X)\hat{p}(m_2 \mid a', X)\right] \\
    &+ \left(\frac{\pi_a(X) - \hat{\pi}_a(X)}{\hat{\pi}_a(X)}\right) \sum_{m_1,m_2}(\mu_a(m_1, m_2, X) - \hat{\mu}_a(m_1, m_2, X))\hat{p}(m_1 \mid a, X) \\
    &\times \hat{p}(m_2 \mid a', X)\hat{p}(m_1, m_2 \mid a, X) \\
    &- \frac{\pi_a(X)}{\hat{\pi}_a(X)}\sum_{m_1,m_2}\hat{\mu}_a(m_1, m_2 \mid a, X)[p(m_1, m_2 \mid a, X) - \hat{p}(m_1, m_2 \mid a, X)] \\
    &\times[p(m_2 \mid a', X)(p(m_1 \mid a, X) - \hat{p}(m_1 \mid a, X)) + \hat{p}(m_1 \mid a, X)(p(m_2 \mid a', X) - \hat{p}(m_2 \mid a', X))] \\
    &+ \left(\frac{\pi_a(X) - \hat{\pi}_a(X)}{\hat{\pi}_a(X)}\right)\sum_{m_1, m_2}\hat{\mu}_a(m_1, m_2, X)(p(m_1, m_2 \mid a, X) - \hat{p}(m_1, m_2 \mid a, X))\hat{p}(m_1 \mid a, X)\hat{p}(m_2 \mid a', X) \\
    &+\left(\frac{\pi_a(X) - \hat{\pi}_a(X)}{\hat{\pi}_a(X)}\right)\sum_{m_1,m_2}\hat{\mu}_a(m_1, m_2, X)\hat{p}(m_1, m_2 \mid a, X)[p(m_1 \mid a, X) - \hat{p}(m_1 \mid a, X)]\hat{p}(m_2 \mid a', X) \\
    &+\left(\frac{\pi_{a'}(X) - \hat{\pi}_{a'}(X)}{\hat{\pi}_{a'}(X)}\right)\sum_{m_1,m_2}\hat{\mu}_a(m_1, m_2, X)\hat{p}(m_1, m_2 \mid a, X)(p(m_2 \mid a', X) - \hat{p}(m_2 \mid a', X))\hat{p}(m_1 \mid a, X) \\
    &-\sum_{m_1, m_2}\hat{\mu}_a(m_1, m_2 \mid a, X)(p(m_1, m_2 \mid a, X) - \hat{p}(m_1, m_2 \mid a, X))[(p(m_1 \mid a, x) - \hat{p}(m_1 \mid a, X))]\hat{p}(m_2 \mid a', X) \\
    &- \sum_{m_1,m_2}\hat{\mu}_a(m_1, m_2, X)\hat{p}(m_1, m_2 \mid a, X)[p(m_2 \mid a', X) - \hat{p}(m_2 \mid a', X)] \\
    &\times\{\hat{p}(m_1, m_2 \mid a, X)[p(m_1 \mid a, X) - \hat{p}(m_1 \mid a, X)] + p(m_1 \mid a, X)[p(m_1, m_2 \mid a, X) - \hat{p}(m_1, m_2 \mid a, X)]\} \\
    &- \sum_{m_1, m_2}[\mu_a(m_1,m_2, X) - \hat{\mu}_a(m_1, m_2, X)]\{p(m_1, m_2 \mid a, X)p(m_2 \mid a', X)[p(m_1 \mid a, X) - \hat{p}(m_1 \mid a, X)] \\
    &+ p(m_1,m_2 \mid a, X)\hat{p}(m_1 \mid a, X)[p(m_2 \mid a', X) - \hat{p}(m_2 \mid a', X)] \\
    &+ \hat{p}(m_1 \mid a, X)\hat{p}(m_2 \mid a', X)[p(m_1, m_2 \mid a, X) - \hat{p}(m_1, m_2 \mid a, x)]\}]
\end{align*}

\begin{align*}
    &P[\hat{\phi}_{1,a'}(Z) - \phi_{1,a'}(Z)] \\
    &= \mathbb{E}\left[\frac{\pi_a(X) }{\hat{\pi}_a(X)}\sum_{m_1,m_2}(p(m_1, m_2 \mid a, x) - \hat{p}(m_1, m_2 \mid a, x))(\mu_a(m_1, m_2, X) - \hat{\mu}_a(m_1, m_2, X))\right.\\
    &\left.\times\hat{p}(m_1 \mid a', X)\hat{p}(m_2 \mid a', X)\right] \\
    &+ \left(\frac{\pi_a(X) - \hat{\pi}_a(X)}{\hat{\pi}_a(X)}\right) \sum_{m_1,m_2}(\mu_a(m_1, m_2, X) - \hat{\mu}_a(m_1, m_2, X))[\hat{p}(m_1 \mid a', X)] \\
    &\times \hat{p}(m_2 \mid a', X)\hat{p}(m_1, m_2 \mid a, X) \\
    &- \frac{\pi_a(X)}{\hat{\pi}_a(X)}\sum_{m_1,m_2}\hat{\mu}_a(m_1, m_2 \mid a, X)[p(m_1, m_2 \mid a, X) - \hat{p}(m_1, m_2 \mid a, X)] \\
    &\times[p(m_2 \mid a', X)(p(m_1 \mid a', X) - \hat{p}(m_1 \mid a', X)) + \hat{p}(m_1 \mid a', X)(p(m_2 \mid a', X) - \hat{p}(m_2 \mid a', X))] \\
    &- \left(\frac{\pi_a(X) - \hat{\pi}_a(X)}{\hat{\pi}_a(X)}\right)\sum_{m_1, m_2}\hat{\mu}_a(m_1, m_2, X)(p(m_1, m_2 \mid a, X) - \hat{p}(m_1, m_2 \mid a, X))\hat{p}(m_1 \mid a', X) \hat{p}(m_2 \mid a', X) \\
    &-\left(\frac{\pi_a(X) - \hat{\pi}_a(X)}{\hat{\pi}_a(X)}\right)\sum_{m_1,m_2}\hat{\mu}_a(m_1, m_2, X)\hat{p}(m_1, m_2 \mid a, X)[(p(m_1 \mid a', X) - \hat{p}(m_1 \mid a', X))]\hat{p}(m_2 \mid a', X) \\
    &+\left(\frac{\pi_{a'}(X) - \hat{\pi}_{a'}(X)}{\hat{\pi}_{a'}(X)}\right)\sum_{m_1,m_2}\hat{\mu}_a(m_1, m_2, X)\hat{p}(m_1, m_2 \mid a, X)(p(m_2 \mid a', X) - \hat{p}(m_2 \mid a', X))\hat{p}(m_1 \mid a', X) \\
    &-\sum_{m_1, m_2}\hat{\mu}_a(m_1, m_2 \mid a, X)(p(m_1, m_2 \mid a, X) - \hat{p}(m_1, m_2 \mid a, X))[p(m_1 \mid a', x) - \hat{p}(m_1 \mid a', X)]\hat{p}(m_2 \mid a', X) \\
    &- \sum_{m_1,m_2}\hat{\mu}_a(m_1, m_2, X)\hat{p}(m_1, m_2 \mid a, X)[p(m_2 \mid a', X) - \hat{p}(m_2 \mid a', X)] \\
    &\times \{\hat{p}(m_1, m_2 \mid a, X)[p(m_1 \mid a', X) - \hat{p}(m_1 \mid a', X)] + p(m_1 \mid a', X)[p(m_1, m_2 \mid a, X) - \hat{p}(m_1, m_2 \mid a, X)]\} \\
    &- \sum_{m_1, m_2}[\mu_a(m_1,m_2, X) - \hat{\mu}_a(m_1, m_2, X)]\{p(m_1, m_2 \mid a, X)p(m_2 \mid a', X)[p(m_1 \mid a', X) - \hat{p}(m_1 \mid a', X)] \\
    &+ p(m_1,m_2 \mid a, X)\hat{p}(m_1 \mid a', X)[p(m_2 \mid a', X) - \hat{p}(m_2 \mid a', X)] \\
    &+ \hat{p}(m_1 \mid a', X)\hat{p}(m_2 \mid a', X)[p(m_1, m_2 \mid a, X) - \hat{p}(m_1, m_2 \mid a, x)]\}]
\end{align*}

\begin{align*}
    &P[\hat{\phi}_{2,a}(Z) - \phi_{2,a}(Z)] \\
    &= \mathbb{E}\left[\left(\frac{\pi_a(X) - \hat{\pi}_a(X)}{\hat{\pi}_a(X)}\right)\sum_{m_1, m_2} [p(m_1 \mid a, X) - \hat{p}(m_1 \mid a, X)]\hat{p}(m_2 \mid a', X)\hat{p}(m_1, m_2 \mid a, X)\right.\\
    \nonumber&\left.+ \left(\frac{\pi_{a'}(X) - \hat{\pi}_{a'}(X)}{\hat{\pi}_a(X)}\right)\sum_{m_1, m_2} [p(m_2 \mid a', X) - \hat{p}(m_2 \mid a', X)]\hat{p}(m_1 \mid a, X)\hat{p}(m_1, m_2 \mid a, X) \right.\\ 
    \nonumber&\left.+ \left(\frac{\pi_a(X) - \hat{\pi}_a(X)}{\hat{\pi}_a(X)}\right)\sum_{m_1, m_2} [p(m_1, m_2 \mid a, X) - \hat{p}(m_1, m_2 \mid a, X)]\hat{p}(m_1 \mid a, X)\hat{p}(m_2 \mid a', X)  \right.\\
    \nonumber&\left.-\sum_{m_1, m_2}\hat{p}(m_1, m_2 \mid a, X)[p(m_1 \mid a, X) - \hat{p}(m_1 \mid a, X)][p(m_2 \mid a', X) - \hat{p}(m_2 \mid a', X)] \right.\\
    \nonumber&\left.-\sum_{m_1, m_2}\hat{p}(m_1 \mid a, x)[p(m_1, m_2 \mid a, X) - \hat{p}(m_1, m_2 \mid a, X)][p(m_2 \mid a', X) - \hat{p}(m_2 \mid a', X)] \right.\\
    \nonumber&\left.- \sum_{m_1, m_2}p(m_2 \mid a', X)[p(m_1, m_2 \mid a, X) - \hat{p}(m_1, m_2 \mid a, X)][p(m_1 \mid a, X) - \hat{p}(m_1 \mid a, X)]\right]
\end{align*}

\begin{align*}
    &P[\hat{\phi}_{2,a'}(Z) - \phi_{2,a'}(Z)] \\
    &= \mathbb{E}\left[\left(\frac{\pi_{a'}(X) - \hat{\pi}_{a'}(X)}{\hat{\pi}_{a'}(X)}\right)\sum_{m_1, m_2} [p(m_1 \mid a', X) - \hat{p}(m_1 \mid a', X)]\hat{p}(m_2 \mid a', X)\hat{p}(m_1, m_2 \mid a, X)\right.\\
    \nonumber&\left.+ \left(\frac{\pi_{a'}(X) - \hat{\pi}_{a'}(X)}{\hat{\pi}_{a'}(X)}\right)\sum_{m_1, m_2} [p(m_2 \mid a', X) - \hat{p}(m_2 \mid a', X)]\hat{p}(m_1 \mid a', X)\hat{p}(m_1, m_2 \mid a, X) \right.\\ 
    \nonumber&\left.+ \left(\frac{\pi_a(X) - \hat{\pi}_a(X)}{\hat{\pi}_a(X)}\right)\sum_{m_1, m_2} [p(m_1, m_2 \mid a, X) - \hat{p}(m_1, m_2 \mid a, X)]\hat{p}(m_1 \mid a', x)\hat{p}(m_2 \mid a', X)  \right.\\
    \nonumber&\left.-\sum_{m_1, m_2}\hat{p}(m_1 \mid a', x)[p(m_1, m_2 \mid a, X) - \hat{p}(m_1, m_2 \mid a, X)][p(m_2 \mid a', X) - \hat{p}(m_2 \mid a', X)] \right.\\
    \nonumber&\left.-\sum_{m_1, m_2}\hat{p}(m_1, m_2 \mid a, x)[p(m_1 \mid a', X) - \hat{p}(m_1 \mid a', X)][p(m_2 \mid a', X) - \hat{p}(m_2 \mid a', X)] \right.\\
    \nonumber&\left.- \sum_{m_1, m_2}p(m_2 \mid a', X)[p(m_1, m_2 \mid a, X) - \hat{p}(m_1, m_2 \mid a, X)][p(m_1 \mid a', X) - \hat{p}(m_1 \mid a', X)]\right]
\end{align*}

Each of the terms above can be shown to be bounded by the product of nuisance estimation times a constant, having assumed that the propensity scores, the joint mediator probabilities, and their estimates are bounded away from zero (and one for the propensity scores) by $\epsilon$. For example, consider the first term in the decomposition of $P[\hat{\phi}_{2,a}(Z) - \phi_{2,a}(Z)]$:

\begin{align*}
    &\mathbb{E}\left[\left(\frac{\pi_a(X) - \hat{\pi}_a(X)}{\hat{\pi}_a(X)}\right)\sum_{m_1, m_2} [p(m_1 \mid a, X) - \hat{p}(m_1 \mid a, X)]\hat{p}(m_2 \mid a', X)\hat{p}(m_1, m_2 \mid a, X)\right]
\end{align*}

By the boundedness of the propensity score estimates, we obtain that this term:

\begin{align*}
    &\le \frac{1}{\epsilon} P\left[\frac{\mathds{1}(A = a)[\pi_a(X) - \hat{\pi}_a(X)][p(M_1 \mid a, X) - \hat{p}(M_1 \mid a, X)]\hat{p}(M_2 \mid a', X)\hat{p}(M_1, M_2 \mid a, X)}{p(A \mid X)p(M_1, M_2 \mid A, X)}\right] \\
    &\le \frac{1}{\epsilon^3} \|\pi_a(X) - \hat{\pi}_a(X)\|\|p(M_1 \mid a, X) - \hat{p}(M_1 \mid a, X)\|
\end{align*}
where the final line holds by the boundedness of the propensity scores, mediator probabilities, and their estimates, and by Cauchy-Schwarz. We can make similar derivations for the remaining terms. Therefore, as long as $Z_n = o_p(n^{-1/2})$, where

\begin{align*}
    Z_n &= \|\hat{\mu}_a(M_1, M_2, X) - \mu_a(M_1, M_2, X)\|\left[\|\hat{\pi}(X)_a - \pi_a(X)\| \right.\\
    &\left.+\|\hat{p}(M_1 \mid a, X) - p(M_1 \mid a, X)\| + \|\hat{p}(M_1, M_2 \mid a, X) - p(M_1, M_2 \mid a, X)\| \right. \\
    &\left.+ \|\hat{p}(M_1 \mid a', X) - p(M_1 \mid a', X)\| + \|\hat{p}(M_2 \mid a', X) - p(M_2 \mid a', X)\|\right] \\
    &+\|\hat{\pi}_a(X) - \pi_a(X)\|\left[\|\hat{p}(M_1 \mid a, X) - p(M_1 \mid a, X)\| + \|\hat{p}(M_1 \mid a', X) - p(M_1 \mid a', X)\| \right.\\
    &\left.+ \|\hat{p}(M_2 \mid a', X) - p(M_2 \mid a', X)\| +
    \|\hat{p}(M_1, M_2 \mid a, X) - p(M_1, M_2 \mid a, X)\|\right] \\
    &+ \|\hat{p}(M_1, M_2 \mid a, X) - p(M_1, M_2 \mid a, X)\|\left[\|\hat{p}(M_1 \mid a, X) - p(M_1 \mid a, X)\| \right.\\
    &\left. +\|\hat{p}(M_1 \mid a', X) - p(M_1 \mid a', X)\| + \|\hat{p}(M_2 \mid a', X) - p(M_2 \mid a', X)\|\right] \\
    &+ \|\hat{p}(M_2 \mid a', X) - p(M_2 \mid a', X)\|[\|\hat{p}(M_1 \mid a, X) - p(M_1 \mid a, X)\| \\
    &+ \|\hat{p}(M_1 \mid a', X) - p(M_1 \mid a', X)\|]
\end{align*}
then the result follows. This would be satisfied, if, for example, each element of $\eta$ were estimated at a rate of $o_p(n^{-1/4})$. 
\end{proof}

\begin{remark}
Estimating each element of $\eta$ at rates of $o_p(n^{-1/4})$ suffices to guarantee both asymptotic normality of the point estimate of $\psi_{M_1}$ and the bounds above. 
\end{remark}

\begin{proof}[Proof of Corollary~\ref{corollary4}]
    The proof of this corollary for either the projection estimator of the bounds is identical to the proof of Theorem \ref{theorem3}, where the result follows from the fact that: 
    
    \begin{align*}
        P[\phi(Z; \hat{\eta}, \tau) - \phi(Z; \eta, \tau)] 
    \end{align*}
    is second-order in the nuisance estimation, with the second-order expression provided by the term $Z_n$ in the statement of Corollary~\ref{corollary4}. 
\end{proof}

\begin{proof}[Proof of Corollary~\ref{corollary5}]
    The proof of this corollary relies primarily on the fact that the conditional bias term at $X = x$ is second-order in the nuisance estimation. This follows using the decomposition from the proof of Theorem~\ref{theorem5} and following the steps of Corollaries~\ref{corollary-drlearner} and \ref{corollary-extension}. 
    
    We first derive the following expressions:  
    
    \begin{align*}
        &\hat{k}_{1, a}^\star(x)  = \mathbb{E}[\hat{\phi}_{1, a} - \phi_{1, a} \mid X = x, D_1^n] \\
        &\hat{k}_{1, a'}^\star(x) = \mathbb{E}[\hat{\phi}_{1, a'} - \phi_{1, a'} \mid X = x, D_1^n] \\
        &\hat{k}_{2, a}^\star(x)  = \mathbb{E}[\hat{\phi}_{2, a} - \phi_{2, a} \mid X = x, D_1^n] \\
        &\hat{k}_{2, a'}^\star(x) = \mathbb{E}[\hat{\phi}_{2, a'} - \phi_{2, a'} \mid X = x, D_1^n] \\
    \end{align*}
    noting that $\mathbb{E}[\hat{\Xi}_{M_1, lb} - \Xi_{M_1, lb} \mid X = x, D_1^n]$ can be derived as linear combinations of these these expressions, $\hat{b}(v)$ (derived previously), $\tau$. To ease notation, we let $\pi_a = \pi_a(x)$, $\mu_a = \mu_a(x, m_1, m_2)$, $p_1 = p(m_1 \mid a, x)$, $p_2 = p(m_2 \mid a', x)$, and so forth,
    
    \begin{align*}
        \hat{k}_{1,a}^\star(x) &= \frac{\pi_a}{\hat{\pi}_a}\sum_{m_1,m_2}(\mu_a - \hat{\mu}_a)(p_{12}-\hat{p}_{12})\hat{p}_1\hat{p}_2' + \left(\frac{\pi_a - \hat{\pi}_a}{\hat{\pi}_a}\right)\sum_{m_1,m_2} (\mu_a - \hat{\mu}_a)\hat{p}_{12}\hat{p}_1 \hat{p}_2' \\
        &- \frac{\pi_a}{\hat{\pi}_a}\sum_{m_1,m_2}\hat{\mu}_a(p_{12} - \hat{p}_{12})[(\hat{p}_1(p_2' - \hat{p}_2') + p_2'(p_1 - \hat{p}_1)] \\
        &+ \left(\frac{\pi_a - \hat{\pi}_a}{\hat{\pi}_a}\right)\sum_{m_1,m_2}\hat{\mu}_a(p_{12} - \hat{p}_{12})\hat{p}_2'\hat{p}_1 \\
        &+ \left(\frac{\pi_a - \hat{\pi}_a}{\hat{\pi}_a}\right)\sum_{m_1,m_2}\hat{\mu}_a\hat{p}_{12}\hat{p}_2'[p_1 - \hat{p}_1] + \left(\frac{\pi_{a'} - \hat{\pi}_{a'}}{\hat{\pi}_{a'}}\right)\sum_{m_1,m_2}\hat{\mu}_a\hat{p}_{12}\hat{p}_1[(p_2' - \hat{p}_2')] \\
        &- \sum_{m_1, m_2}\hat{\mu}_a(p_{12} - \hat{p}_{12})[(p_1 - \hat{p}_1)]\hat{p}_2' \\
        &- \sum_{m_1,m_2}\hat{\mu}_a(p_2' - \hat{p}_2')[p_1(p_{12} - \hat{p}_{12}) + \hat{p}_{12}(p_1 - \hat{p}_1)] \\
        &- \sum_{m_1,m_2}(\mu_a - \hat{\mu}_a)[p_2'p_{12}(p_1 - \hat{p}_1) + \hat{p}_1p_{12}(\hat{p}_2' - p_2') + \hat{p}_1\hat{p}_2'(p_{12} - \hat{p}_{12})]    
    \end{align*}

    \begin{align*}
        \hat{k}_{1,a'}^\star(x) &= \frac{\pi_a}{\hat{\pi}_a}\sum_{m_1,m_2}(\mu_a - \hat{\mu}_a)(p_{12}-\hat{p}_{12})\hat{p}_1'\hat{p}_2' + \left(\frac{\pi_a - \hat{\pi}_a}{\hat{\pi}_a}\right)\sum_{m_1,m_2} (\mu_a - \hat{\mu}_a)\hat{p}_{12}\hat{p}_1'\hat{p}_2' \\
        &- \frac{\pi_a}{\hat{\pi}_a}\sum_{m_1,m_2}\hat{\mu}_a(p_{12} - \hat{p}_{12})[(p_1'(p_2' - \hat{p}_2') + \hat{p}_2'(p_1' - p_1')] \\
        &- \left(\frac{\pi_a - \hat{\pi}_a}{\hat{\pi}_a}\right)\sum_{m_1,m_2}\hat{\mu}_a(p_{12} - \hat{p}_{12})\hat{p}_2'\hat{p}_1' \\
        &- \left(\frac{\pi_a - \hat{\pi}_a}{\hat{\pi}_a}\right)\sum_{m_1,m_2}\hat{\mu}_a\hat{p}_{12}\hat{p}_2'[(p_1' - \hat{p}_1')] + \left(\frac{\pi_{a'} - \hat{\pi}_{a'}}{\hat{\pi}_{a'}}\right)\sum_{m_1,m_2}\hat{\mu}_a\hat{p}_{12}\hat{p}_1'[(p_2' - \hat{p}_2')] \\
        &- \sum_{m_1, m_2}\hat{\mu}_a(p_{12} - \hat{p}_{12})[(p_1' - \hat{p}_1')]\hat{p}_2' \\
        &- \sum_{m_1,m_2}\hat{\mu}_a(p_2' - \hat{p}_2')[p_1'(p_{12} - \hat{p}_{12}) + \hat{p}_{12}(p_1' - \hat{p}_1')] \\
        &- \sum_{m_1,m_2}(\mu_a - \hat{\mu}_a)[p_2'p_{12}(p_1' - \hat{p}_1') + \hat{p}_1'p_{12}(p_2' - \hat{p}_2') + \hat{p}_1'\hat{p}_2'(p_{12} - \hat{p}_{12})]    
    \end{align*}
    
    \begin{align*}
        \hat{k}_{2,a}^\star(x) &= \left(\frac{\pi_a - \hat{\pi}_a}{\hat{\pi}_a}\right)\sum_{m_1, m_2} [p_{12} - \hat{p}_{12}]\hat{p}_1\hat{p}_2'
        + \left(\frac{\pi_a - \hat{\pi}_a}{\hat{\pi}_a}\right) \sum_{m_1, m_2}[p_1 - \hat{p}_1]\hat{p}_{12}\hat{p}_2' \\
        &+\left(\frac{\pi_{a'} - \hat{\pi}_{a'}}{\hat{\pi}_{a'}}\right)\sum_{m_1, m_2}[p_2' - \hat{p}_2']\hat{p}_{12}\hat{p}_1 \\
        &-\sum_{m_1,m_2}[p_{12} - \hat{p}_{12}]p_2'[p_1 - \hat{p}_1] - \sum_{m_1,m_2}[p_{12} - \hat{p}_{12}]\hat{p}_1[p_2' - \hat{p}_2'] \\
        &-\sum_{m_1, m_2}\hat{p}_{12}[p_1 - \hat{p}_1][p_2' - \hat{p}_2']
    \end{align*}

    \begin{align*}
        \hat{k}_{2,a'}^\star(x) &= \left(\frac{\pi_a - \hat{\pi}_a}{\hat{\pi}_a}\right)\sum_{m_1, m_2} [p_{12} - \hat{p}_{12}]\hat{p}_1\hat{p}_2'  + \left(\frac{\pi_{a'} - \hat{\pi}_{a'}}{\hat{\pi}_{a'}}\right) \sum_{m_1, m_2}[p_1' - \hat{p}_1']\hat{p}_{12}\hat{p}_2' \\ 
        &+\left(\frac{\pi_{a'} - \hat{\pi}_{a'}}{\hat{\pi}_{a'}}\right)\sum_{m_1, m_2}[p_2' - \hat{p}_2']\hat{p}_{12}\hat{p}_1' \\
        &-\sum_{m_1,m_2}[p_{12} - \hat{p}_{12}]p_2'[p_1' - \hat{p}_1'] - \sum_{m_1,m_2}[p_{12} - \hat{p}_{12}]\hat{p}_1'[p_2' - \hat{p}_2'] \\
        &-\sum_{m_1, m_2}\hat{p}_{12}[p_1' - \hat{p}_1'][p_2' - \hat{p}_2']
    \end{align*}
    
    The expression for $\hat{q}_{lb}(v)$ in Corollary~\ref{corollary5} follows from collecting all possible second-order products and the boundedness of the outcome model, propensity scores, mediator probabilities, and their estimates. The final result follows from the definition of estimator stability and applying Proposition~\ref{kennedyprop2} from \cite{kennedy2022}.
\end{proof}

\newpage

\newpage

\section{Simulation details}\label{app:simulation}

We provide additional details about the data generating process for our simulation and additional simulation results. 

\subsection{Data generation}

We first define the covariates and the propensity-score model:
\begin{align*}
 & X^\star \sim \mathcal{N}(1, 0.5) \\
 & X = -2\mathds{1}(X^\star < -2) + 4 \mathds{1}(X^\star > 4) + Z\mathds{1}(X^\star \ge -2, X^\star \le 4) \\
 & \pi_a(X) = 0.2\mathds{1}(X < -1) + (0.2 + 0.55 \lvert X + 1 \rvert \mathds{1}(X > -1, X < 0) \\
 &+ (0.75 - 0.25 X)\mathds{1}(X > 0, X < 1) + (0.5 - 0.25 (X - 1)^2) \mathds{1}(X > 1, X < 2) \\
 &+ (0.25 + 0.5 (X - 2)) \mathds{1}(X > 2, X < 3) + 0.75 \mathds{1}(X > 3) 
\end{align*}
We next define the counterfactual mediator probabilities:

\begin{align*}
 &U \sim \text{Bern}(0.5) \\
 &p(M_1(0) = 1 \mid U = 0, X) = 0.15 + 0.1(X + 1) \\
 &p(M_1(1) = 1 \mid U = 0, X) = 0.55 + 0.05(X + 1)  \\
 &p(M_1(0) = 1 \mid U = 1, X) = 0.1 \\
 &p(M_1(1) = 1 \mid U = 1, X) = 0.8 \\
 &p(M_2(0) = 1 \mid U = 0, X) = 0.15 + 0.125(X + 1)  \\
 &p(M_2(1) = 1 \mid U = 0, X) = 0.4 + 0.1(X + 0.5)  \\
 &p(M_2(0) = 1 \mid U = 1, X) = 0.1 \\
 &p(M_2(1) = 1 \mid U = 1, X) = 0.8
\end{align*}
We multiply the relevant probabilities together to obtain the joint mediator probabilities given $U = u$. For $j \in \{1, 2\}$ and $a^\star \in \{0, 1\}$, we then marginalize all of these probabilities over $U$ to obtain $p(M_{j}(a^\star) \mid X)$ and $p(M_1^{a^\star}, M_2^{a^\star} \mid X)$. We next define the following intermediate functions for the outcome models:

\begin{align*}
&\zeta(X) = (X - X^2) \mathds{1}(X < - 0.5) + (-2 + X) \mathds{1} (X > -0.5, X < 0)  \\
&+(- 12 + 10\sin(X^2) + 10\cos(X^2)) \mathds{1} (X > 0, X < 1)) \\
&+ (-12 + 10\sin(1) + 10\cos(1) - 5(X - 1) - 5(X - 1)^2) \mathds{1}(X > 1, X < 1.5) \\
&+(- 12 + 10\sin(1) + 10\cos(1) - 3.75 + 0.5 * (X - 1.5) - (X - 1.5)^2 + 3(X - 1.5)^3) \mathds{1} (X > 1.5, X < 2.5)) \\
&+(-4 + 2*X) \mathds{1} (X > 2.5)) \\
&\tilde{\mu}_{a^\star}(M_1 = 1, M_2 = 1, X) = 10 + \zeta(X) + 2X + 0.5X^2 \\ 
&\tilde{\mu}_{a^\star}(M_1 = 0, M_2 = 1, X) = 4 + \zeta(X) \\  
&\tilde{\mu}_{a^\star}(M_1 = 1, M_2 = 0, X) = 8 + \zeta(X) + 2X + 0.5X^2 \\
&\tilde{\mu}_{a^\star}(M_1 = 0, M_2 = 0, X) = \zeta(X) \\
&z_{a^\star,u} = \max_{m_1, m_2, x}\tilde{\mu}_a(M_1 = m_1, M_2 = m_2, X = x) \\
&z_{a^\star,l} = \max_{m_1, m_2, x}\tilde{\mu}_a(M_1 = m_1, M_2 = m_2, X = x)
\end{align*}

Finally, we define the outcome models for $a^\star \in \{0, 1\}$ at any $(m_1, m_2, x)$ as:

\begin{align*}
\mu_{a^\star}(m_1, m_2, x) &= (\tilde{\mu}_{a^\star}(m_1, m_2, x) - z_{a^\star,l} + 10) / (z_{a^\star,u} - z_{a^\star,l} + 20)
\end{align*}
All realizations of the potential mediators and outcomes at the individual-level are drawn Bernoulli with the mean parameter at the corresponding mediator or outcome probability.

\subsection{Inverse weights}

Figure~\ref{fig:sim-inv-weights} displays the maximum possible inverse probability weight associated with the true influence curve for $\psi_{M_1}$ associated at different values of $X = x$. This plot illustrates that $\psi_{M_1}(x)$ is likely easiest to estimate at $X = 0$ and hardest to estimate at $X = 2$.

\begin{figure}[H]
\begin{center}
    \caption{Maximum inverse probability weights for CIIE via $M_1$}\label{fig:sim-inv-weights}
    \includegraphics[scale=0.45]{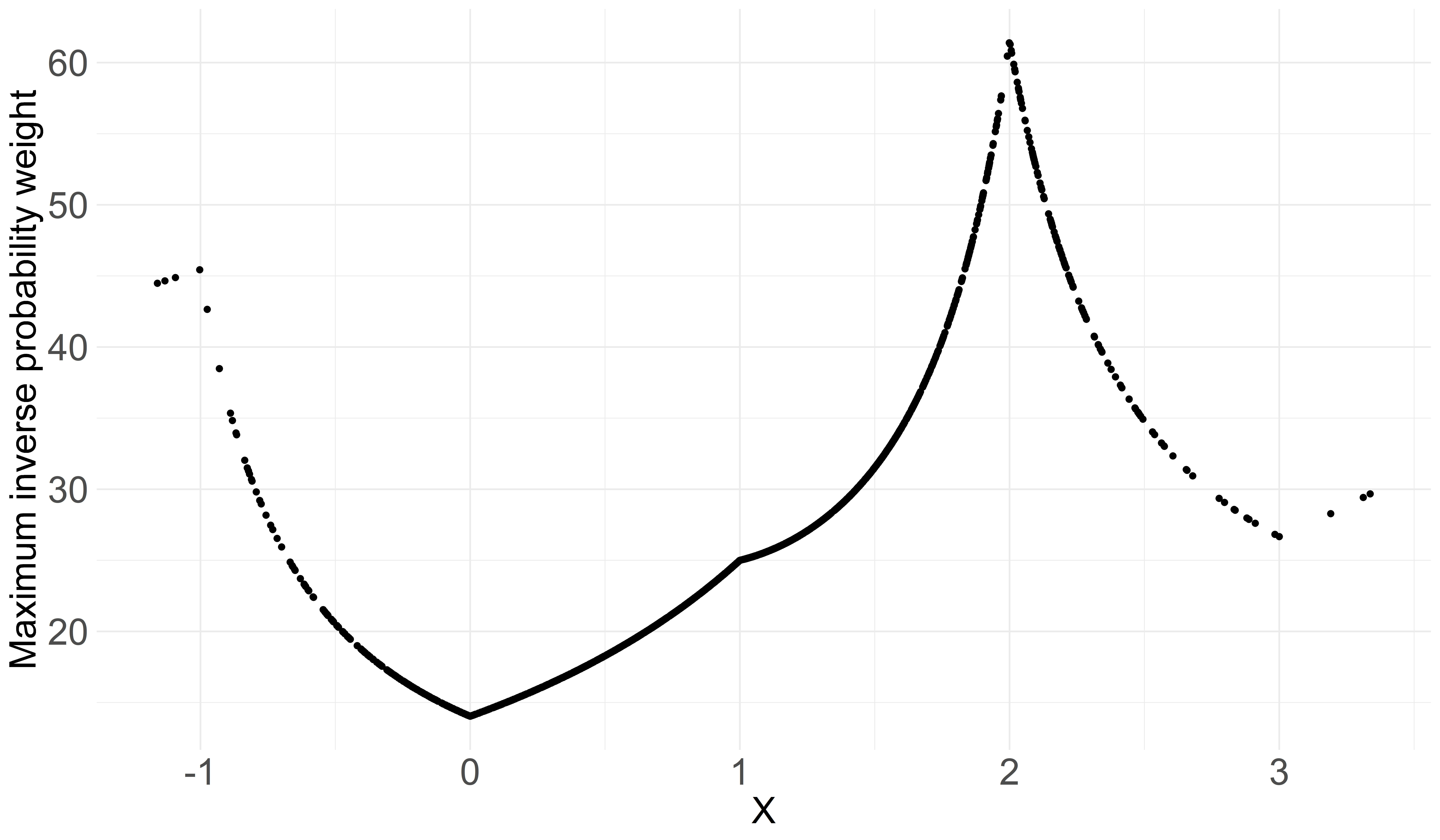}
\end{center}
\end{figure}

\subsection{Proportion mediated}\label{sec:propmed}

The conditional proportion mediated ($\psi_{R}(x)$) is frequently a quantity of interest that relates the CATE ($\psi(x)$) to the CIIE ($\psi_{M_1}(x)$). Specifically, 

\begin{align*}
 \psi_{R}(x)  &= \frac{\psi_{M_1}(x)}{\psi(x)}
\end{align*}

We again consider a DR-learner and projection-based approach to estimate this quantity; however, we consider two approaches. First, where we derive the uncentered influence function for the average quantity $\mathbb{E}[\psi_{R}(x)]$\footnote{This quantity is not the same as the proportion mediated, which is instead the ratio of the expectations.}, which we call $\Lambda(Z; \eta)$. Treating the data as discrete and using standard differentiation rules, it is easy to see that this takes the following form:

\begin{align}
    \nonumber \Lambda(Z; \eta) &= \sum_x \frac{\text{IF}[\psi_{M_1}(x)]\psi(x) - \text{IF}[\psi(x)]\psi_{M_1}(x)]}{\psi(x)^2}p(x) \\
    \nonumber&+ \psi_{R}(x)[\mathds{1}(X = x) - p(x)] \\
    \label{eqn:eifratio}&= \frac{1}{\psi(X)}\left[\frac{\mathds{1}(A = a)}{\pi_a(X)}\frac{\{p(M_1 \mid a, X) - p(M_1 \mid a', X)\}p(M_2 \mid a', X)}{p(M_1, M_2, \mid a, X)}(Y - \mu_a(M_1, M_2, X)) \right.\\
    \nonumber\nonumber&\left.+ \frac{\mathds{1}(A = a)}{\pi_a(X)}\{\mu_{a, M_2^'}(M_1, X) - \mu_{a, M_1\times M_2^'}(X)\} \right.\\ 
    \nonumber\nonumber&\left.- \frac{\mathds{1}(A = a')}{\pi_{a'}(X)}\{\mu_{a, M_2^'}(M_1, X) - \mu_{a, M_1^'\times M_2^'}(X)\} \right. \\
    \nonumber &\left.+ \frac{\mathds{1}(A = a')}{\pi_{a'}(X)}\left(\mu_{a, M_1}(M_2, X) - \mu_{a, M_1\times M_2'}(X) - (\mu_{a, M_1^'}(M_2, X) - \mu_{a, M_1^'\times M_2^'}(X)) \right)\right] \\
    \nonumber&- \frac{\psi_{R}(X)}{\psi(X)}\left[\left(\frac{A}{\pi(X)} - \frac{1 - A}{1 - \pi(X)}\right)(Y - \mu_A(X)) \right] + \psi_{R}(X)
\end{align}

We can then regress an estimate of $\Lambda(Z; \eta)$ onto $X$ using either a non-parametric second-stage model or a projection. Table~\ref{tab:propmedcomp} presents the results of this strategy under the column heading ``EIF Ratio.'' As a second approach, we separately construct estimates of $\psi(x)$ and $\psi_{M_1}(x)$ and take the ratio of these estimates. Table~\ref{tab:propmedcomp} presents the results from this approach in the ``EIF Separate'' column. 

\begin{table}[!h]

\caption{Proportion mediated: comparison of two approaches \label{tab:propmedcomp}}
\centering
\begin{threeparttable}
\begin{tabular}[t]{lrrlrlr}
\toprule
\multicolumn{1}{c}{ } & \multicolumn{1}{c}{ } & \multicolumn{1}{c}{ } & \multicolumn{2}{c}{EIF Ratio} & \multicolumn{2}{c}{EIF Separate} \\
\cmidrule(l{3pt}r{3pt}){4-5} \cmidrule(l{3pt}r{3pt}){6-7}
Estimator & Sample Size & Point & RMSE & Coverage & RMSE & Coverage\\
\midrule
DRLearner & 1000 & 0 & 6e+06 & 0.951 & 475.8 & 0.947\\
DRLearner & 1000 & 2 & 1e+09 & 0.936 & 17.7 & 0.911\\
DRLearner & 2000 & 0 & 6e+07 & 0.958 & 9.3 & 0.958\\
DRLearner & 2000 & 2 & 1e+06 & 0.942 & 2.0 & 0.910\\
Projection-Linear & 1000 & 0 & 2e+07 & 0.965 & 49.7 & 0.956\\
Projection-Linear & 1000 & 2 & 1e+09 & 0.962 & 14.2 & 0.944\\
Projection-Linear & 2000 & 0 & 6e+07 & 0.955 & 8.3 & 0.960\\
Projection-Linear & 2000 & 2 & 9e+05 & 0.964 & 0.6 & 0.944\\
Projection-Quad & 1000 & 0 & 1e+07 & 0.965 & 16.5 & 0.940\\
Projection-Quad & 1000 & 2 & 1e+09 & 0.959 & 8.2 & 0.945\\
Projection-Quad & 2000 & 0 & 6e+07 & 0.947 & 14.0 & 0.939\\
Projection-Quad & 2000 & 2 & 3e+05 & 0.962 & 2.2 & 0.932\\
\bottomrule
\end{tabular}
\begin{tablenotes}
\item EIF Ratio reflects one regression of EIF of the mean ratio onto X; EIF Separate reflects two separate regressions of the EIF of the proportion mediated and the EIF of the total effect onto X
\end{tablenotes}
\end{threeparttable}
\end{table}

Our confidence interval estimates have approximately nominal coverage rates for either strategy; however, the RMSE of the estimator is orders of magnitude higher for the ``ratio'' approach than the ``separate approach.'' This shows that the second approach is more desirable in this setting. 

The form of $\Lambda(Z; \eta)$ expressed in \eqref{eqn:eifratio} reveals why we might expect this estimator to have higher variance more generally. While the influence function for $\psi$ and $\psi_{M_1}$ are functions of inverse weights with respect to the propensity-scores and/or mediator probabilities, the influence function for the average ratio is also a function of $1 / \psi(x)$. Even if this quantity in truth is far away from zero, the estimates may be arbitrarily close to zero. Examining the simulation results confirms this: across 1000 simulations, the median maximum inverse weight for the ratio is approximately 2000 times what it is for the maximum inverse weight for each regression separately, where the difference is driven by estimates of $1 / \psi(x)$. On the other hand, there may be instances where the ratio of the CATE to the CIIE may have less complexity than each function individually; the first-approach can exploit such a case while the second approach does not. Investigating this further would be an interesting avenue for future research.

Variance estimation is also more challenging using the second approach. For the projection estimators we estimate the covariance matrix between the model parameters using the residuals for each model. However, for the non-parametric ``separate'' approach our variance estimates are valid assuming that the errors in each model are positively dependent (which should lead to conservative variance estimates). 

The next table displays the results for the projection estimators of the CIIE, CATE, and conditional proportion mediated using the separate estimation approach.

\newpage

\begin{landscape}
\begin{table}[!h]

\caption{Projection estimators: simulation performance}
\centering
\begin{tabular}[t]{llllllllllllllll}
\toprule
\multicolumn{4}{c}{ } & \multicolumn{4}{c}{CIIE-M1} & \multicolumn{4}{c}{CATE} & \multicolumn{4}{c}{Proportion Mediated} \\
\cmidrule(l{3pt}r{3pt}){5-8} \cmidrule(l{3pt}r{3pt}){9-12} \cmidrule(l{3pt}r{3pt}){13-16}
Point & Sample Size & Strategy & Projection & Truth & Bias & RMSE & Coverage & Truth & Bias & RMSE & Coverage & Truth & Bias & RMSE & Coverage\\
\midrule
0 & 1000 & Plugin & Linear & 0.070 & 5.1e-04 & 0.027 & 3.1 & 0.091 & -0.0094 & 0.053 & 3.9 & 0.77 & 1.042 & 20.89 & 4.5\\
2 & 1000 & Plugin & Linear & 0.112 & -3.7e-02 & 0.045 & 1.0 & 0.129 & 0.0149 & 0.056 & 4.2 & 0.87 & -0.202 & 0.93 & 2.5\\
0 & 1000 & Plugin & Quadratic & 0.072 & -2.5e-03 & 0.027 & 3.1 & 0.092 & -0.0109 & 0.053 & 3.8 & 0.78 & -0.830 & 24.19 & 5.2\\
2 & 1000 & Plugin & Quadratic & 0.113 & -4.0e-02 & 0.047 & 0.8 & 0.131 & 0.0132 & 0.056 & 3.9 & 0.87 & -0.252 & 1.53 & 3.3\\
0 & 1000 & DR & Linear & 0.070 & 1.5e-03 & 0.059 & 95.5 & 0.091 & -0.0025 & 0.056 & 95.0 & 0.77 & -1.585 & 49.67 & 95.6\\
2 & 1000 & DR & Linear & 0.112 & -9.5e-05 & 0.076 & 93.8 & 0.129 & 0.0045 & 0.059 & 93.9 & 0.87 & -0.011 & 14.23 & 94.4\\
0 & 1000 & DR & Quadratic & 0.072 & 3.5e-03 & 0.055 & 95.1 & 0.092 & 0.0019 & 0.065 & 94.0 & 0.78 & -0.623 & 16.51 & 94.0\\
2 & 1000 & DR & Quadratic & 0.113 & 1.6e-03 & 0.096 & 93.5 & 0.131 & 0.0108 & 0.068 & 93.2 & 0.87 & 0.046 & 8.20 & 94.5\\
0 & 2000 & Plugin & Linear & 0.070 & 4.3e-04 & 0.018 & 3.0 & 0.091 & -0.0103 & 0.038 & 3.5 & 0.77 & -0.131 & 11.85 & 3.5\\
2 & 2000 & Plugin & Linear & 0.112 & -3.7e-02 & 0.041 & 0.4 & 0.129 & 0.0128 & 0.039 & 2.7 & 0.87 & -0.293 & 0.39 & 0.9\\
0 & 2000 & Plugin & Quadratic & 0.072 & -2.3e-03 & 0.019 & 3.0 & 0.092 & -0.0119 & 0.038 & 3.9 & 0.78 & -3.001 & 91.57 & 3.7\\
2 & 2000 & Plugin & Quadratic & 0.113 & -4.0e-02 & 0.043 & 0.2 & 0.131 & 0.0112 & 0.039 & 2.2 & 0.87 & -0.302 & 0.40 & 1.3\\
0 & 2000 & DR & Linear & 0.070 & 6.6e-04 & 0.040 & 95.6 & 0.091 & -0.0033 & 0.042 & 94.3 & 0.77 & 0.131 & 8.29 & 96.0\\
2 & 2000 & DR & Linear & 0.112 & 8.5e-06 & 0.047 & 94.7 & 0.129 & 0.0079 & 0.040 & 94.7 & 0.87 & 0.014 & 0.65 & 94.4\\
0 & 2000 & DR & Quadratic & 0.072 & 2.8e-03 & 0.037 & 95.2 & 0.092 & 0.0027 & 0.046 & 94.1 & 0.78 & -0.198 & 13.96 & 93.9\\
2 & 2000 & DR & Quadratic & 0.113 & 2.0e-03 & 0.058 & 94.6 & 0.131 & 0.0144 & 0.046 & 93.9 & 0.87 & -0.028 & 2.24 & 93.2\\
\bottomrule
\end{tabular}
\end{table}

\end{landscape}

\subsection{Selection mechanism}\label{sec:selection}

We detail the selection mechanism that we use in Figure~\ref{fig:sim-bounds}. Specifically, consider the function:

\begin{align*}
    \tau^\star(X; \sigma) &= \sigma \left[\mathds{1}(X < -1)*0.01 + \mathds{1}(-1 < X < 0) * 0.02 \right. \\
    &\left.+ \mathds{1}(0 < X < 1) * 0.03 + \mathds{1}(1 < X < 2) * 0.02 \right.\\
    &\left.+ \mathds{1}(2 < X < 3) * 0.01 + \mathds{1}(X > 3) * 0.03\right]
\end{align*}
Given the function $\mathbb{E}[Y^{am_1m_2} \mid a, X]$, we generate the observed functions:

\begin{align}\label{eqn:observedfun}
    \mathbb{E}[Y \mid a, m_1, m_2, X] & = \mu_a(m_1, m_2, X) \\
    \nonumber&= \frac{\mathbb{E}[Y^{am_1m_2} \mid a, X]}{1 - \tau^\star(X)(1 - p(m_1, m_2 \mid a, X))}
    \mu_l(m_1, m_2, X)\mathds{1}(X \ge 1) \\
    \nonumber&+ \frac{\mathbb{E}[Y^{am_1m_2} \mid a, X] - \tau^\star(X)(1 - p(m_1, m_2 \mid a, X))}{1 - \tau^\star(X)(1 - p(m_1, m_2 \mid a, X))}\mathds{1}(X < 1)
\end{align}
for each value of $(m_1, m_2)$. The form of this function follows from assuming Y-A and M-A ignorability and a version of the model in \eqref{eqn:a2} where we replace the inequalities with equalities; i.e. we know the selection mechanism. Specifically, we set:

\begin{align}\label{eqn:selectionknown}
    \frac{\mathbb{E}[Y^{m_1m_2} \mid a, x, M_1 \ne m_1, M_2 \ne m_2]}{\mathbb{E}[Y^{m_1m_2} \mid a, x, m_1, m_2]} = (1 - \tau^\star(x)) \text{ when } x < 1 \\
    \nonumber\frac{\mathbb{E}[1 - Y^{m_1m_2} \mid a, x, M_1 \ne m_1, M_2 \ne m_2]}{\mathbb{E}[1- Y^{m_1m_2} \mid a, x, m_1, m_2]} = (1 - \tau^\star(x)) \text{ when } x \ge 1
\end{align}
Via \eqref{eqn:iie}, we can then plug \eqref{eqn:selectionknown} into (i) and then solve for $\mu_a(m_1, m_2, X)$ to arrive at \eqref{eqn:observedfun}. We then generate the functions: 

\begin{align*}
&\mathbb{E}_l[Y^{am_1m_2} \mid a, x, M_1 \ne m_1, M_2 \ne m_2] = \mu^\star_l(x, m_1, m_2) \\
&\mathbb{E}_u[Y^{am_1m_2} \mid a, x, M_1 \ne m_1, M_2 \ne m_2] = \mu^\star_u(x, m_1, m_2) 
\end{align*}
where for all $(x, m_1, m_2)$:

\begin{align*}
    [1 - \mu^\star_u(x, m_1, m_2; \sigma)] &= [1 - \mu_a(x, m_1, m_2)](1 - \tau^\star(x; \sigma)) \\
    \mu^\star_l(x, m_1, m_2; \sigma) &= \mu_a(x, m_1, m_2)(1 - \tau^\star(x; \sigma)) 
\end{align*}
and set 

\begin{align*}
    &\mathbb{E}[Y^{am_1m_2} \mid a, M_1 \ne m_1, M_2 \ne m_2, x] \\
    &= \mu^\star(x, m_1, m_2; \sigma) \\
    &= \mu_l^\star(x, m_1, m_2; \sigma)\mathds{1}(x < 1) + \mu_u^\star(x, m_1, m_2; \sigma)\mathds{1}(x \ge 1)
\end{align*}
for each value of $(x, m_1, m_2)$.

Figure~\ref{fig:sim-selectionbias} illustrates the true counterfactual outcome model $\mathbb{E}[Y^{am_1m_2} \mid a, X]$ in blue and the resulting biased target of inference $\mu_a(m_1, m_2, X)$ in red for each value of $(m_1, m_2)$ (setting $\sigma = 10/3$).

\begin{figure}[H]
\begin{center}
    \caption{Selection bias: observed data model versus targeted counterfactual outcome model}\label{fig:sim-selectionbias}
    \includegraphics[scale=0.45]{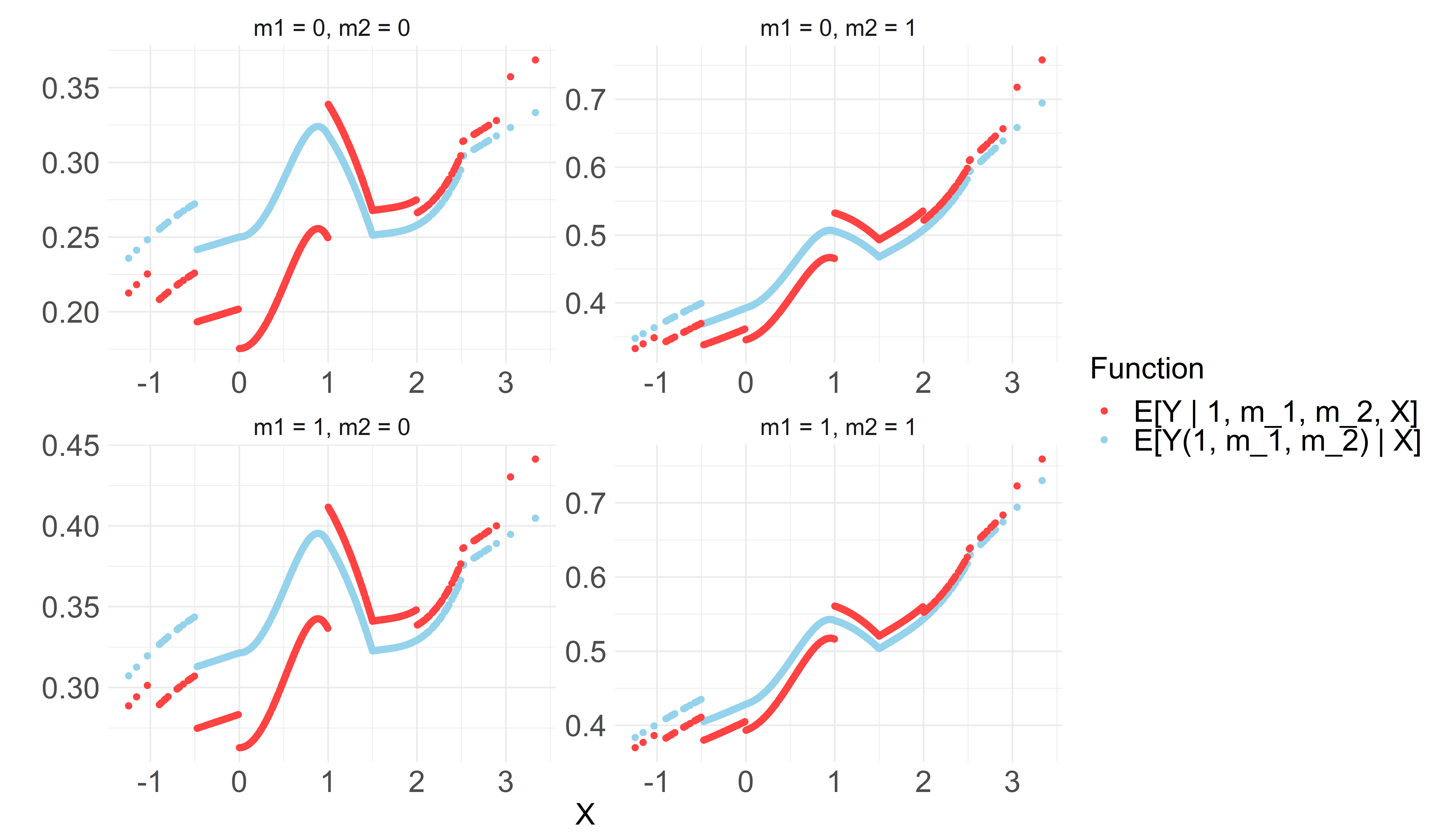}
\end{center}
\end{figure}

Figure~\ref{fig:sim-selectionbias2} instead illustrates $\mu_a(m_1, m_2, X)$ in red against the unobserved function $\mathbb{E}[Y^{m_1m_2} \mid A = a, M_1 \ne m_1, M_2 \ne m_2, X]$.

\begin{figure}[H]
\begin{center}
    \caption{Selection bias conditional on observed covariates}\label{fig:sim-selectionbias2}
    \includegraphics[scale=0.45]{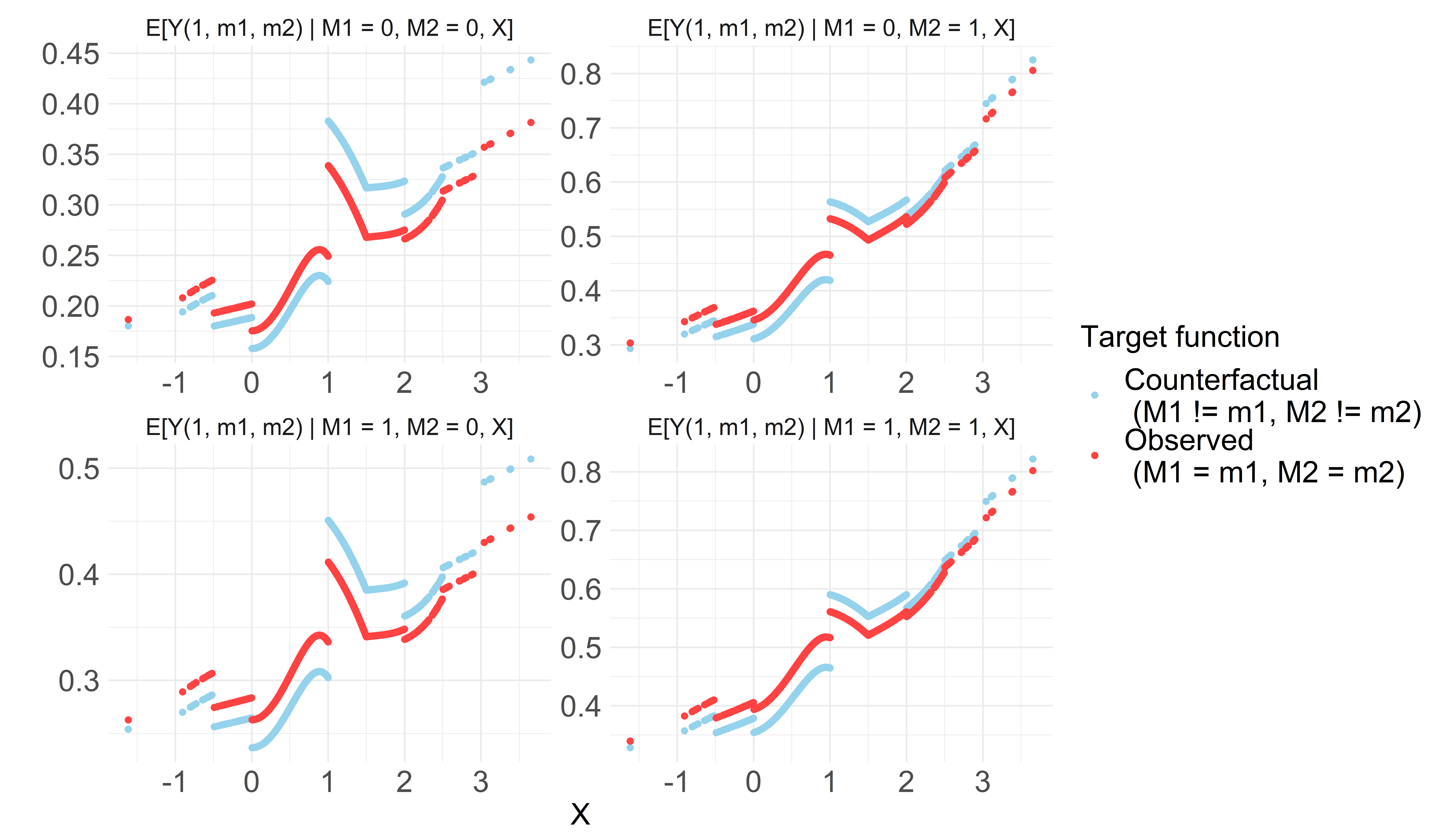}
\end{center}
\end{figure}

\newpage

\newpage

\section{Other application results}\label{app:application}

\subsection{Conditional effect estimates}

Figure~\ref{fig:application1} displays the results from the DR-Learner and Projection estimators applied to the average effect and all elements of its decomposition across the entire domain of the percent of Biden's vote share ($V$). We see that the total effect estimates are larger in Trump counties relative to Biden counties, and the direct effect estimates get closer to zero in Biden counties. The projection estimator suggests that the effects via social isolation are larger in Biden counties relative to Trump counties, though the DR-Learner suggests the presence of non-linearities that might not be correctly captured by this projection. The effects via worries about health are close to zero throughout the domain of $V$. Overall the largest heterogeneity appears to be with respect to the total effect and the direct effect, where both the total effect and proportion mediated via the direct effect is larger in absolute magnitude among Trump counties relative to Biden counties.

\begin{figure}[H]
\begin{center}
    \caption{Application results}\label{fig:application1}
    \includegraphics[scale=0.45]{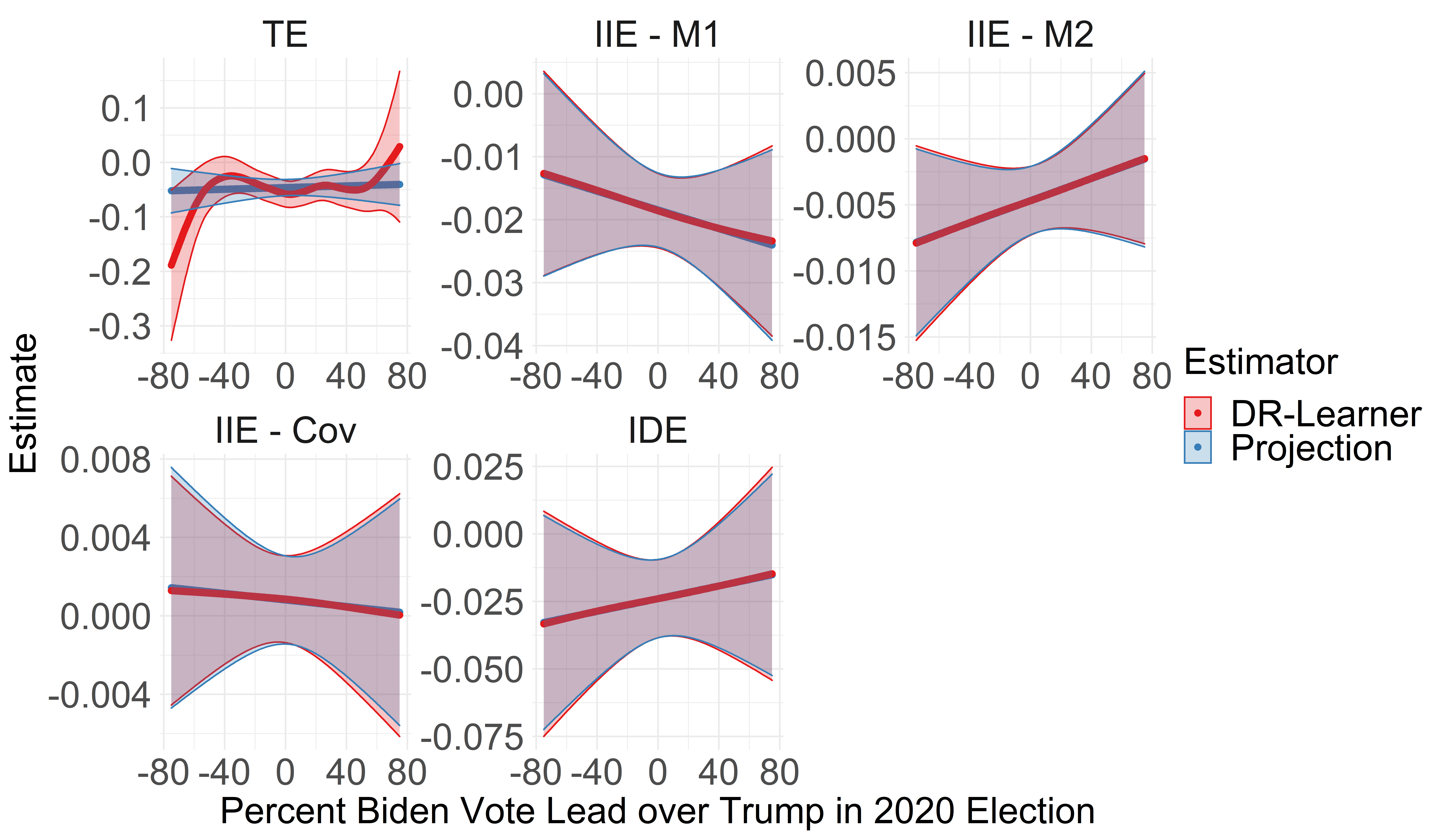}
\end{center}
\end{figure}

\subsection{Bounds on decomposition of total effect}

Figure~\ref{fig:application-bound} displays the results of the bounds on the average effects. The left-hand panel displays the results where we assume that total interventional indirect effect is biased upwards and the right-hand panel displays the results where we assume that this same effect is biased downwards. These results convey an interesting paradox: while the $\tau$ that can explain away $\psi_{M_1}$ is quite low, at the same time we would have to believe in the existence of a covariant effect that is of comparable magnitude to the average effect. Since we tend to think that the covariant effects are likely small in most applications, this implication makes such a $\tau$ seem unlikely. On the other hand, we also see that it is quite hard to entirely explain away the total indirect effect, and such a $\tau$ would also imply the existence of a positive covariant effect that is many times greater than the total effect estimate. 

\begin{figure}[H]
\begin{center}
    \caption{Application results}\label{fig:application-bound}
    \includegraphics[scale=0.45]{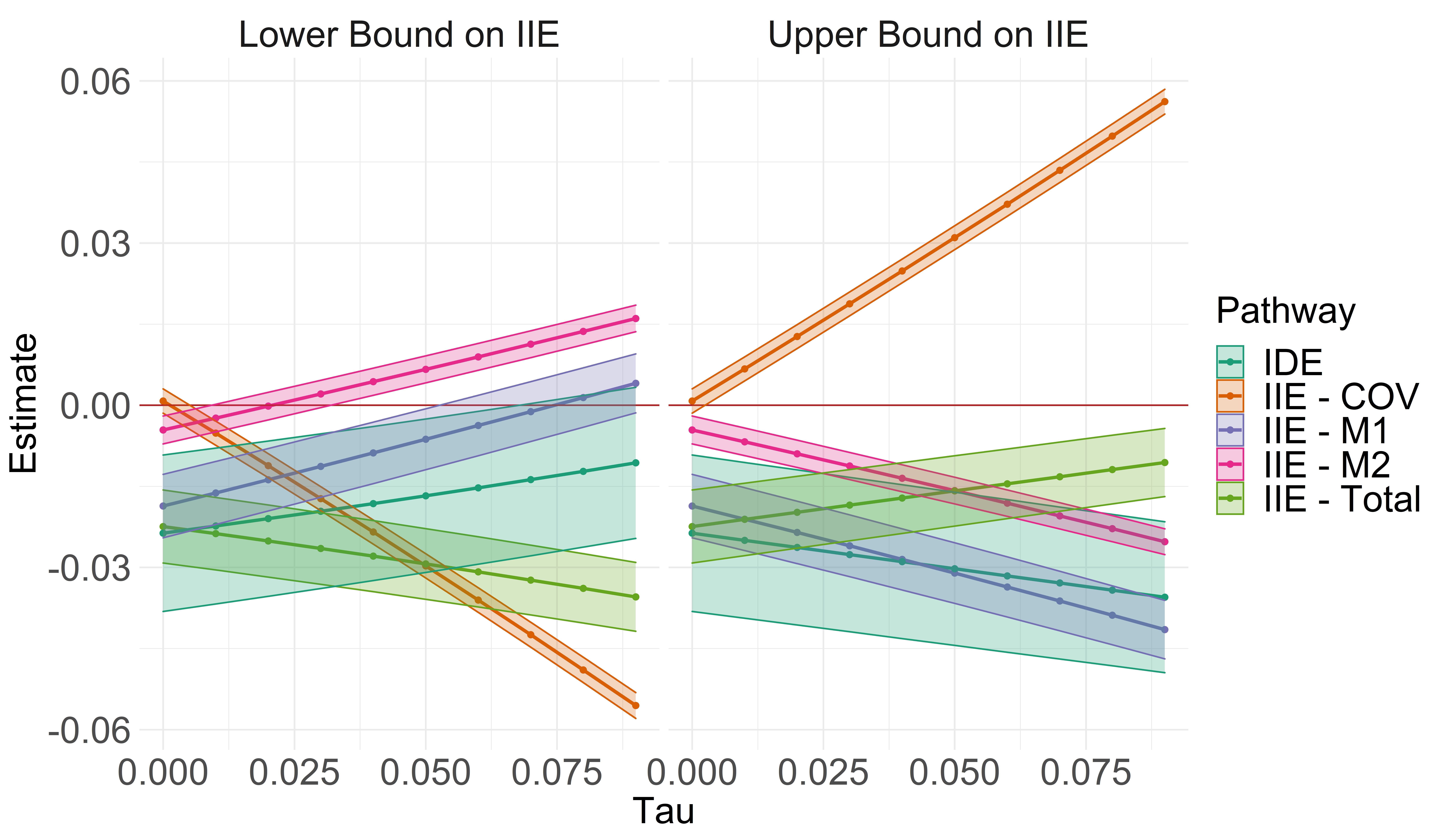}
\end{center}
\end{figure}

In short, these results suggest that if we assume away Y-A or M-A confounding, a very most amount of Y-M confounding would explain away our effect estimates. On the other hand, the implications of such a $\tau$ lead to improbable combinations of effect sizes. In summary, we interpret these results as suggesting that it is probably further worth examining how Y-M confounding might affect our results in the presence of Y-A or M-A confounding. However, such a sensitivity analysis is beyond the scope of this paper.

The results are quite similar for the bounds on the conditional effects, whether using the projection estimator or the DR-Learner, though the uncertainty estimates are quite wide. These results are available on request.

\newpage

\section{Bound extensions}\label{app:bounds}

We present results for the bounds on all of the parameters. We begin by defining the bounds that we consider. Letting $\psi_{IIE} = \psi_{M_1} + \psi_{M_2} + \psi_{Cov}$, we can decompose the total effect $\psi$:

\begin{align}
    \label{eqn:decombound1}\psi &= \psi_{IIE, ub} + \psi_{IDE, lb} \\
    \label{eqn:decombound2}&= \psi_{Cov, ub} + \psi_{M_1, lb} + \psi_{M_2, lb} + \psi_{IDE, lb}
\end{align}
Equation~\ref{eqn:decombound1} follows because since $\psi$ is identified in the data; even though $\psi_{IIE}$ and $\psi_{IDE}$ are not, since $\psi$ must equal the sum of the indirect effect plus the direct effect, we know that:

\begin{align}
    \label{eqn:decompboundident1a}\psi &= \psi_{IIE, lb} + \psi_{IDE, ub} \\
    \label{eqn:decompboundident1b}&= \psi_{IIE, ub} + \psi_{IDE, lb}
\end{align}
Furthermore because $\psi_{Cov} = \psi_{IIE} - \psi_{M_1} - \psi_{M_2}$, we obtain that:

\begin{align}
    \label{eqn:decompboundident1}\psi_{Cov, ub} &= \psi_{IIE, ub} - (\psi_{M_2, lb} + \psi_{M_1, lb}) \\
    \label{eqn:decompboundident2}\psi_{Cov, lb} &= \psi_{IIE, lb} - (\psi_{M_2, ub} + \psi_{M_1, ub}) \\ 
    \implies \label{eqn:decompboundident1}\psi_{IIE, ub} &= \psi_{Cov, ub} + \psi_{M_1, lb} + \psi_{M_2, lb} \\
    \implies \label{eqn:decompboundident2}\psi_{IIE, lb} &= \psi_{Cov, lb} + \psi_{M_1, ub} + \psi_{M_2, ub} 
\end{align}
We obtain the decomposition in (\ref{eqn:decombound3})-(\ref{eqn:decombound4}) similarly. 

\begin{align}
    \label{eqn:decombound3}\psi &= \psi_{IIE, lb} + \psi_{IDE, ub} \\
    \label{eqn:decombound4}&= \psi_{Cov, lb} + \psi_{M_1, ub} + \psi_{M_2, ub} + \psi_{IDE, ub}
\end{align}

This decomposition assumes that the relevant contrasts decompose the total effect $\psi$. However, this may not be the relevant comparison depending on the application: letting $G_m(a \mid x) \sim p(M_1, M_2 \mid a, X)$  $\mathbb{E}[Y(a,G_m(a \mid x))] \ne \mathbb{E}[Y(a)]$ in the presence of post-treatment confounding. If we instead wish to decompose the ``overall effect'' $\mathbb{E}[Y(1,G_m(1 \mid x))] - \mathbb{E}[Y(0, G_m(0 \mid x))]$, then our point estimates of the average effect are not unbiased for this quantity. While it would be straightforward to extend our sensitivity analysis for this setting, equalities (\ref{eqn:decombound1})-(\ref{eqn:decombound4}) would no longer hold. Since each parameter can be written as the difference in two components, valid upper bounds for each parameter could be derived via subtracting the lower bound from the second parameter from the upper bound of the first, and vice versa for the lower bounds.

\subsection{Expression of bounds}

This section provides bounds on $\psi_{M_2}$, $\psi_{IIE}$, $\psi_{Cov}$, and $\psi_{IDE}$ using the decompositions provided in \eqref{eqn:decombound2} and \eqref{eqn:decombound4}. First, let

\begin{align*}
    &\bar{\psi}_{M_2} = \bar{\psi}_{M_2, a} - \bar{\psi}_{M_2, a'} = \mathbb{E}\left\{\sum_{m_1, m_2}\mu_a(m_1, m_2, X)[p(m_2 \mid a, X) - p(m_2 \mid a', X)]p(m_1 \mid a, X)\right\} \\
    &\bar{\psi}_{IIE, a'} = \mathbb{E}\left\{\sum_{m_1, m_2}\mu_a(X, m_1, m_2)p(m_1, m_2 \mid a', X)\right\} \\
    &\psi_a = \mathbb{E}\left\{\mathbb{E}[Y \mid a, X]\right\} 
\end{align*}
We can then obtain that:

\begin{align*}
    \psi_{M_2, ub} &= \sum_x[\bar{\psi}_{M_2}(x) + \bar{\psi}_{M_2, a}(x)f_u(\tau)c_u - \bar{\psi}_{M_2, a'}(x)f_l(\tau)c_l + t_uf_u(\tau) - t_lf_l(\tau) \\
    &- f_u(\tau)\sum_{m_1,m_2}[c_u\mu_a(m_1,m_2,x) + t_u]p(m_1, m_2 \mid a, x)p(m_1 \mid a, x)p(m_2 \mid a, x) \\
    &+ f_l(\tau)\sum_{m_1,m_2}[c_l\mu_a(m_1,m_2,x) + t_l]p(m_1, m_2 \mid a, x)p(m_1 \mid a, x)p(m_2 \mid a', x)]p(x) \\
    \psi_{M_2, lb} &= \sum_x[\bar{\psi}_{M_2}(x) + \bar{\psi}_{M_2, a}(x)f_l(\tau)c_l - \bar{\psi}_{M_2, a'}(x)f_u(\tau)c_u + t_lf_l(\tau) - t_uf_u(\tau) \\
    &- f_l(\tau)\sum_{m_1,m_2}[c_l\mu_a(m_1,m_2,x) + t_l]p(m_1, m_2 \mid a, x)p(m_1 \mid a, x)p(m_2 \mid a, x) \\
    &+ f_u(\tau)\sum_{m_1,m_2}[c_u\mu_a(m_1,m_2,x) + t_u]p(m_1, m_2 \mid a, x)p(m_1 \mid a, x)p(m_2 \mid a', x)]p(x) \\
    \psi_{IIE, ub} &= \psi^a - f_u(\tau)c_u\bar{\psi}_{IIE, a'} - t_uf_u(\tau) \\
    &+ f_u(\tau)\sum_{x, m_1, m_2}[c_u\mu_a(m_1, m_2, x) + t_u]p(m_1, m_2 \mid a, x)p(m_1, m_2 \mid a', x)p(x)  \\
    \psi_{IIE, lb} &= \psi^a - f_l(\tau)c_l\bar{\psi}_{IIE, a'} - t_lf_l(\tau) \\
    &+ f_l(\tau)\sum_{x, m_1, m_2}[c_l\mu_a(m_1, m_2, x) + t_l]p(m_1, m_2 \mid a, x)p(m_1, m_2 \mid a', x)p(x)  
\end{align*}

Given these bounds and the bounds on $\psi_{M_1}$, the remaining bounds on $\psi_{IDE}$ and $\psi_{Cov}$ follow directly via equations (\ref{eqn:decompboundident1a})-(\ref{eqn:decompboundident2}).

\subsection{Influence function of bounds}

We provide expression for the (uncentered) influence functions of the $\Gamma$ terms above noting that the remaining influence functions have been previously defined and that $\text{IF}(\Gamma_a - \Gamma_b) = \text{IF}(\Gamma_a) - \text{IF}(\Gamma_b)$. We omit the proofs for brevity but note that the derivations are analogous to those provided in Section~\ref{app:proofs}.

First, we define the terms:

\begin{align*}
    \Gamma_{1a, M_2} &= \sum_{x,m_1,m_2}\mu_a(m_1,m_2,x)p(m_1,m_2\mid a,x)p(m_2\mid a,x)p(m_1 \mid a, x)p(x) \\
    \Gamma_{1a', M_2} &= \sum_{x,m_1,m_2}\mu_a(m_1,m_2,x)p(m_1,m_2\mid a,x)p(m_2\mid a',x)p(m_1 \mid a, x)p(x) \\
    \Gamma_{2a, M_2} &= \sum_{x,m_1,m_2}p(m_1,m_2\mid a,x)p(m_2\mid a,x)p(m_1 \mid a, x)p(x) \\
    \Gamma_{2a', M_2} &= \sum_{x,m_1,m_2}p(m_1,m_2\mid a,x)p(m_2\mid a',x)p(m_1 \mid a, x)p(x) \\
    \Gamma_{1, IIE} &= \sum_{x,m_1,m_2}\mu_a(m_1,m_2,x)p(m_1,m_2\mid a,x)p(m_1,m_2\mid a',x)p(x) \\
    \Gamma_{2, IIE} &= \sum_{x,m_1,m_2}p(m_1,m_2\mid a,x)p(m_1,m_2\mid a',x)p(x)
\end{align*}
Note that all of the bounds above can be expressed as a linear combination of these terms, $\bar{\psi}_{M_2, a}$, $\bar{\psi}_{M_2, a'}$, $\bar{\psi}_{IIE, a'}$, with various terms scaled by $f_l(\tau)$ and $f_u(\tau)$ and the constants $(c_l, c_u, t_l, t_u)$. We therefore only require expressions for the influence functions for these terms above.

\begin{align*}
    &\text{IF}[\Gamma_{1a, M2}]  = \frac{\mathds{1}(A = a)}{\pi_a(X)}\left\{Yp(M_1 \mid a, X) p(M_2  \mid a, X)p(M_1, M_2 \mid a, X) - \zeta_{1a, M_2}(X) \right\} \\
    &+\frac{\mathds{1}(A = a)}{\pi_a(X)}\left\{\sum_{m_2}\mu_a(M_1, m_2, X)p(m_2 \mid a, x)p(M_1, M_2 \mid a, X) - \zeta_{1a, M_2}(X)\right\} \\
    &+\frac{\mathds{1}(A = a)}{\pi_a(X)}\left\{\sum_{m_1}\mu_a(m_1, M_2, X)p(m_1 \mid a, x)p(m_1, M_2 \mid a, X) - \zeta_{1a, M_2}(X)\right\} \\ 
    &+ \zeta_{1a, M_2}(X) \\
    &\text{IF}[\Gamma_{1a', M2}]  = \frac{\mathds{1}(A = a)}{\pi_a(X)}\left\{Yp(M_1 \mid a, X) p(M_2  \mid a', X)p(M_1, M_2 \mid a, X) - \zeta_{1a', M_2}(X) \right\} \\
    &+\frac{\mathds{1}(A = a)}{\pi_a(X)}\left\{\sum_{m_2}\mu_a(M_1, m_2, X)p(m_2 \mid a', x)p(M_1, M_2 \mid a, X) - \zeta_{1a', M_2}(X)\right\} \\
    &+ \frac{\mathds{1}(A = a')}{\pi_{a'}(X)}\left\{\sum_{m_1}\mu_a(m_1, M_2, X)p(m_1 \mid a, x)p(m_1, M_2 \mid a, X) - \zeta_{1a',M_2}(X)\right\} \\
    &+ \zeta_{1a', M_2}(X) \\
    &\text{IF}[\Gamma_{2a, M2}]  = \frac{\mathds{1}(A = a)}{\pi_a(X)}\left\{p(M_1 \mid a, X)p(M_2 \mid a, X) - \zeta_{2a, M_2}(X)\right\} \\
    &+ \frac{\mathds{1}(A = a)}{\pi_a(X)}\left\{\sum_{m_1}p(m_1, M_2 \mid a, X)p(m_1 \mid a, X) - \zeta_{2a, M_2}(X)\right\} \\ 
    &+ \frac{\mathds{1}(A = a)}{\pi_a(X)}\left\{\sum_{m_2}[p(M_1, m_2 \mid a, X)p(m_2 \mid a, X) - \zeta_{2a, M_2}(x)\right\} + \zeta_{2a, M_2}(X) \\
    &\text{IF}[\Gamma_{2a', M2}]  = \frac{\mathds{1}(A = a)}{\pi_a(X)}\left\{p(M_1 \mid a, X)p(M_2 \mid a', X) - \zeta_{2a', M_2}(X)\right\} \\
    &+ \frac{\mathds{1}(A = a')}{\pi_{a'}(X)}\left\{\sum_{m_1}[p(m_1, M_2 \mid a, X)p(m_1 \mid a, X) - \zeta_{2a', M_2}(X)\right\} \\
    &+ \frac{\mathds{1}(A = a)}{\pi_a(X)}\left\{\sum_{m_2}p(M_1, m_2 \mid a, X)p(m_2 \mid a', X) - \zeta_{2a', M_2}(X)\right\} + \zeta_{2a', M_2}(X) \\
    &\text{IF}[\Gamma_{1, IIE}] = \frac{\mathds{1}(A = a)}{\pi_a(X)}\left\{Yp(M_1, M_2 \mid a', X) - \zeta_{1, IIE}(X)\right\} \\
    &+ \frac{\mathds{1}(A = a')}{\pi_{a'}(X)}\left\{\mu_a{M_1, M_2, X}p(M_1, M_2 \mid a, X) - \zeta_{1, IIE}(X)\right\} \\
    &+ \zeta_{1, IIE}(X) \\
    &\text{IF}[\Gamma_{2, IIE}] = \frac{\mathds{1}(A = a)}{\pi_a(X)}\left\{p(M_1, M_2 \mid a', X) - \zeta_{2, IIE}(X)\right\} \\
    &+\frac{\mathds{1}(A = a')}{\pi_{a'}(X)}\left\{p(M_1, M_2 \mid a, X) - \zeta_{3, IIE}(X)\right\} + \zeta_{2, IIE}(X)
\end{align*}
where

\begin{align*}
    &\zeta_{1a, M_2}(X) = \sum_{m_1, m_2}\mu_a(m_1, m_2, X)p(m_1,m_2 \mid a, X)p(m_1 \mid a, X)p(m_2 \mid a, X) \\
    &\zeta_{1a', M_2}(X) = \sum_{m_1, m_2}\mu_a(m_1, m_2, X)p(m_1,m_2 \mid a, X)p(m_1 \mid a, X)p(m_2 \mid a', X) \\
    &\zeta_{2a, M_2}(X) = \sum_{m_1, m_2}p(m_1 \mid a, X)p(m_2 \mid a, X)p(m_1, m_2 \mid a, X) \\
    &\zeta_{2a', M_2}(X) = \sum_{m_1, m_2}p(m_1 \mid a, X)p(m_2 \mid a', X)p(m_1, m_2 \mid a, X) \\
    &\zeta_{1, IIE}(X) = \sum_{m_1, m_2}\mu_a(m_1, m_2, X)p(m_1, m_2 \mid a, X)p(m_1, m_2 \mid a', X) \\
    &\zeta_{2, IIE}(X) = \sum_{m_1, m_2}p(m_1, m_2 \mid a, X)p(m_1, m_2 \mid a', X)
\end{align*}

\newpage

\section{Second-order term derivations}\label{app:soterms}

This section contains the algebra that shows the results in Appendix~\ref{app:proofs} for the second-order errors for the influence-function based estimators of the average effects and bounds on the average effects.

\subsection{Estimating $\psi_{M_1}$}

We show that $P[\hat{\varphi}(Z) - \varphi(Z)]$ can be decomposed as a product of the nuisance estimation. To simplify the derivation we only show this for the term:

\begin{align*}
    \psi_{M_1, a} = \mathbb{E}\left[\sum_{m_1, m_2}\mu_a(m_1, m_2, X)p(m_1 \mid a, X)p(m_2 \mid a', X)\right]
\end{align*}
noting that the derivation for the term $\psi_{M_1, a'}$ is virtually identical. $\psi_{M_1, a}$ has the (uncentered) influence curve $\varphi_a(Z; \eta)$:

\begin{align}\label{eqn:eifm1a}
    \varphi_a(Z; \eta) &= \frac{\mathds{1}(A = a)}{\pi_a(X)}\frac{p(M_1 \mid a, X)p(M_2 \mid a', X)}{p(M_1, M_2, \mid a, X)}(Y - \mu_a(M_1, M_2, X)) \\
    \nonumber &+ \frac{\mathds{1}(A = a)}{\pi_a(X)}\{\mu_{a, M_2^'}(M_1, X) - \mu_{a, M_1\times M_2^'}(X)\} \\ 
    \nonumber &+ \frac{\mathds{1}(A = a')}{\pi_{a'}(X)}\left(\mu_{a, M_1}(M_2, X) - \mu_{a, M_1\times M_2'}(X)\right) \\
    \nonumber &+ \mu_{a, M_1\times M_2^'}(X)
\end{align}

\begin{align}
    \nonumber&P[\varphi_a(Z; \hat{\eta}) - \varphi_{a}(Z; \eta)] \\
    \label{eqn:d1}&= \mathbb{E}\left[\frac{\pi_a(X)}{\hat{\pi}_a(X)}\sum_{m_1, m_2}\frac{\{\mu_a(m_1, m_2, X) - \hat{\mu}_a(m_1, m_2, X)\}p(m_1, m_2 \mid a, X)\hat{p}(m_1 \mid a, X)\hat{p}(m_2 \mid a, X)}{\hat{p}(m_1, m_2 \mid a, X)}\right] \\
    \label{eqn:d2}&+ \mathbb{E}\left[\frac{\pi_a(X)}{\hat{\pi}_a(X)}\left(\sum_{m_1, m_2}\hat{\mu}_a(m_2, m_1, X)\hat{p}(m_2 \mid a', X)(p(m_1 \mid a, X) - \hat{p}(m_1 \mid a, X))\right)\right] \\
    \label{eqn:d3}&+ \mathbb{E}\left[\frac{\pi_{a'}(X)}{\hat{\pi}_{a'}(X)}\left(\sum_{m_1, m_2}\hat{\mu}_a(m_1, m_2, X)\hat{p}(m_1 \mid a, x)(p(m_2 \mid a', x) - \hat{p}(m_2 \mid a', X)) \right)\right] \\
    \label{eqn:d4}&+ \mathbb{E}\left[\left(\underbrace{\hat{\mu}_{a, M_1\times M_2^'}(X)}_{(I)} - \mu_{a, M_1\times M_2^'}(X)\right)\right] 
\end{align}
where equations (\ref{eqn:d1})-(\ref{eqn:d4}) follow via iterating expectations, and where we leave the conditioning on the training data $D_1^n$ implicit. To ease notation, let $p_1 = p(m_1 \mid a, x)$, $p_2' = p(m_2 \mid a', x)$, $\mu_a = \mu_a(m_1, m_2, x)$, and $\pi_a = \pi_a(X)$. We also remove outer expectation for clarity, and note that the final expression takes the expectation of these remainder terms. First, consider \eqref{eqn:d1}:

\begin{align}
    \label{eqn:d1a}(\ref{eqn:d1}) &= \underbrace{\frac{\pi_a}{\hat{\pi}_a}\sum_{m_1, m_2}\frac{(\mu_a - \hat{\mu}_a)(p_{12} - \hat{p}_{12})\hat{p}_1\hat{p}_2'}{\hat{p}_{12}}}_{SO_1} + \frac{\pi_a}{\hat{\pi}_a}\sum_{m_1, m_2}(\mu_a - \hat{\mu}_a)\hat{p}_1\hat{p}_2' \\
    \label{eqn:d1b}&= SO_1 + \underbrace{\left(\frac{\pi_a - \hat{\pi}_a}{\hat{\pi}_a}\right)\sum_{m_1, m_2}\hat{p}_1\hat{p}_2'(\mu_a - \hat{\mu}_a)}_{SO_2} + \sum_{m_1, m_2}(\mu_a - \underbrace{\hat{\mu}_a)\hat{p}_1\hat{p}_2'}_{-(I)}
\end{align}

Equation (\ref{eqn:d1b}) expresses \eqref{eqn:d1} in terms of two second-order terms $SO_1$ and $SO_2$ and a remainder term. Moving forward we abbreviate all second-order terms as SO. We then consider the remaining terms in (\ref{eqn:d1})-(\ref{eqn:d4}), noting that term -(I) in \eqref{eqn:d1b} will cancel out with the corresponding term (I) in \eqref{eqn:d4}, to obtain:

\begin{align}
    \label{eqn:d5}(\ref{eqn:d1})-(\ref{eqn:d4}) &= SO + \underbrace{\left(\frac{\pi_a - \hat{\pi}_a}{\hat{\pi}_a}\right)\sum_{m_1, m_2}\hat{\mu}_a\hat{p}_2'(p_1 - \hat{p}_1)}_{SO_3} + \underbrace{\left(\frac{\pi_{a'} - \hat{\pi}_{a'}}{\hat{\pi}_{a'}}\right)\sum_{m_1, m_2}\hat{\mu}_a\hat{p}_1(p_2' - \hat{p}_2')}_{SO_4} \\
    \label{eqn:d6}&+ \sum_{m_1, m_2}\mu_a\hat{p}_1\hat{p}_2' - \sum_{m_1, m_2}\mu_ap_1p_2' + \sum_{m_1, m_2}\hat{\mu}_a\hat{p}_2'(p_1 - \hat{p}_1) + \sum_{m_1, m_2}\hat{\mu}_a\hat{p}_1(p_2' - \hat{p}_2') 
\end{align}
where \eqref{eqn:d6} includes the remainder from before \eqref{eqn:d1b}, the $-\psi_{M_1, a}(x)$ term from \eqref{eqn:d1b}, and additional remainder terms from \eqref{eqn:d5} obtained via adding and subtracting the same quantity to obtain second-order terms $SO_3$ and $SO_4$. We conclude by decomposing \eqref{eqn:d6}, the last part of the expression that has not yet been shown to be second-order:

\begin{align*}
    (\ref{eqn:d6}) &= \sum_{m_1, m_2}\mu_a(\hat{p}_1\hat{p}_2' - p_1p_2') - \underbrace{\sum_{m_1, m_2}(\mu_a - \hat{\mu}_a)\hat{p}_2'(p_1 - \hat{p}_1)}_{SO_5} + \sum_{m_1, m_2} \mu_a\hat{p}_2'(p_1 - \hat{p}_1) \\
    &+ \sum_{m_1, m_2} \hat{\mu}_a\hat{p}_1(p_2' - \hat{p}_2') \\
    &= SO - \sum_{m_1, m_2} \mu_a p_1 (p_2'  - \hat{p}_2') + \sum_{m_1, m_2} \hat{\mu}_a\hat{p}_1(p_2' - \hat{p}_2') \\
    &= SO - \sum_{m_1, m_2} \mu_a p_1 (p_2'  - \hat{p}_2') - \underbrace{\sum_{m_1, m_2} \hat{\mu}_a(p_1 - \hat{p}_1)(p_2' - \hat{p}_2')}_{SO_6} + \sum_{m_1, m_2} \hat{\mu}_a p_1(p_2' - \hat{p}_2') \\
    &= SO - \underbrace{\sum_{m_1, m_2} (\mu_a - \hat{\mu}_a)p_1(p_2' - \hat{p}_2')}_{SO_7}
\end{align*}

Thus we have shown that $P[\hat{\psi}_{M_1, a} - \psi_{M_1, a}]$ can be expressed by the sums of terms (i)-(vii), which are second-order in the nuisance estimation. The derivation for $P[\hat{\psi}_{M_1, a'} - \psi_{M_1, a'}]$ follows nearly identical steps due to the symmetry of the problem, yielding the result.

\subsection{Estimating the bounds}

We first consider $P[\hat{\phi}_{2,a} - \phi_{2,a}]$:

\begin{align}
    \label{eqn:b1} &= P\left[\frac{\mathds{1}(A = a)}{\hat{\pi}_a(X)}\hat{p}(M_1 \mid a, X)\hat{p}(M_2 \mid a', X) - \zeta_2(X)] \right. \\ 
    \nonumber&+\left. \frac{\mathds{1}(A = a)}{\hat{\pi}_a(X)}[\sum_{m_2}\hat{p}(M_1, m_2 \mid a, X)\hat{p}(m_2 \mid a', X) - \hat{\zeta}_2(X)] \right.\\
    \nonumber&+\left. \frac{\mathds{1}(A = a')}{\hat{\pi}_{a'}(X)}\sum_{m_1}[\hat{p}(m_1, M_2 \mid a, X)\hat{p}(m_1 \mid a, X) - \hat{\zeta}_2(X)] + \hat{\zeta}_2(X) - \zeta_2(X) \right] \\
    \label{eqn:b2} &= \mathbb{E}\left[\frac{\pi_a(X)}{\hat{\pi}_a(X)}\sum_{m_1, m_2}\hat{p}(m_1 \mid a, X)\hat{p}(m_2 \mid a', X)[p(m_1, m_2 \mid a, X) - \hat{p}(m_1, m_2 \mid a, X)] \right. \\ 
    \nonumber&+\left. \frac{\pi_a(X)}{\hat{\pi}_a(X)}[\sum_{m_2}\hat{p}(M_1, m_2 \mid a, X)\hat{p}(m_2 \mid a', X)[p(m_1 \mid a, X) - \hat{p}(m_1 \mid a, X)] \right.\\
    \nonumber&+\left. \frac{\pi_{a'}(X)}{\hat{\pi}_{a'}(X)}\sum_{m_1, m_2}(\hat{p}(m_1, M_2 \mid a, X)\hat{p}(m_1 \mid a, X)[p(m_2 \mid a', X) - \hat{p}(m_2 \mid a', X)]) + \hat{\zeta}_2(X) - \zeta_2(X) \right]
\end{align}
where \eqref{eqn:b2} follows via iterating expectations (and implicitly conditioning on $D_1^n$). For clarity, we again use the abbreviations noted above, as well as $p_{12} = p(m_1, m_2 \mid a, X)$ and $p_1' = p(m_1 \mid a', X)$, and again remove the outer expectation.
    
First, note that:

\begin{align}
     \label{eqn:b3a}\ref{eqn:b2} &=\underbrace{\left(\frac{\pi_a - \hat{\pi}_a}{\hat{\pi}_a}\right)\sum_{m_1, m_2} [p_{12} - \hat{p}_{12}]\hat{p}_1\hat{p}_2'}_{SO_1} +
    \underbrace{\left(\frac{\pi_a - \hat{\pi}_a}{\hat{\pi}_a}\right) \sum_{m_1, m_2}[p_1 - \hat{p}_1]\hat{p}_{12}\hat{p}_2'}_{SO_2} \\
    \label{eqn:b3b}&+\underbrace{\left(\frac{\pi_{a'} - \hat{\pi}_{a'}}{\hat{\pi}_{a'}}\right)\sum_{m_1, m_2}[p_2' - \hat{p}_2]\hat{p}_{12}\hat{p}_1}_{SO_3}  \\
    \label{eqn:b3c}&+ \underbrace{\sum_{m_1, m_2} [p_{12} - \hat{p}_{12}]\hat{p}_1\hat{p}_2' + \sum_{m_1, m_2}[p_1 - \hat{p}_1]\hat{p}_{12}\hat{p}_2'+ \sum_{m_1, m_2}[p_2' - \hat{p}_2']\hat{p}_{12}\hat{p}_1 + \hat{p}_1\hat{p}_2'\hat{p}_{12} - p_1p_2'p_{12}}_{R_1} 
\end{align}
    
Terms $SO_1$ through $SO_3$ are second-order. We consider $R_1$ and show that these term are second-order, proving the result.

\begin{align*}
    R_1 &= -\underbrace{\sum_{m_1, m_2} [p_{12} - \hat{p}_{12}][p_1p_2' - \hat{p}_1\hat{p}_2']}_{SO_4} - \underbrace{\sum_{m_1, m_2}\hat{p}_{12}[p_1 - \hat{p}_1][p_2' - \hat{p}_2']}_{SO_5} \\
    &+ \sum_{m_1, m_2} p_1p_2[p_{12} - \hat{p}_{12}] + \sum_{m_1, m_2}p_2'\hat{p}_{12}[p_1 - \hat{p}_1] \\
    &+ \sum_{m_1, m_2}[p_2' - \hat{p}_2']\hat{p}_{12}\hat{p}_1 + \sum_{m_1,m_2}[\hat{p}_1\hat{p}_2'\hat{p}_{12} - p_1p_2'p_{12}] \\
    &= SO_4 + SO_5
\end{align*}

Finally, notice that

\begin{align*}
    SO_4 &= -\sum_{m_1, m_2} [p_{12} - \hat{p}_{12}][p_1p_2' + \hat{p}_1p_2' - \hat{p}_1p_2' - \hat{p}_1\hat{p}_2'] \\
    &= -\sum_{m_1,m_2}[p_{12} - \hat{p}_{12}]p_2'[\hat{p}_1 - p_1] - \sum_{m_1,m_2}[p_{12} - \hat{p}_{12}]\hat{p}_1[\hat{p}_2' - p_2']
\end{align*}

This yields the final result. The proof for $P[\hat{\phi}_{2,a'} - \phi_{2,a'}]$ follows virtually identical steps so we omit it for brevity. 
 
We conclude by showing the result for $P[\hat{\phi}_{1, a}(Z) - \phi_{1, a}(Z)]$, noting that the derivation is analogous for $P[\hat{\phi}_{1,a'}(Z) - \phi_{1,a'}(Z)]$, noting that the derivation is analogous for the other terms. By the law of iterated expectations, we see that:

\begin{align*}
    &P[\hat{\phi}_{1, a}(Z) - \phi_{1, a}(Z)] \\
    &=\mathbb{E}[\underbrace{\frac{\pi_a}{\hat{\pi}_a}\sum_{m_1, m_2}(\mu_a - \hat{\mu}_a)p_{12}\hat{p}_1\hat{p}_2}_{(i)} + \underbrace{\frac{\pi_a}{\hat{\pi}_a}\sum_{m_1, m_2}\hat{\mu}_a(p_{12} - \hat{p}_{12})\hat{p}_1\hat{p}_2}_{(ii)} \\
    &+ \underbrace{\frac{\pi_a}{\hat{\pi}_a}\sum_{m_1, m_2}\hat{\mu}_{a}\hat{p}_{12}(p_1 - \hat{p}_1)\hat{p}_2' + \frac{\pi_{a'}}{\hat{\pi}_{a'}}\sum_{m_1, m_2}\hat{\mu}_a\hat{p}_{12}\hat{p}_1(p_2' - \hat{p}_2')}_{(iii)} \\
    &+ \underbrace{\left(\sum_{m_1, m_2}\hat{\mu}_a\hat{p}_{12}\hat{p}_1\hat{p}_2' - \mu_ap_{12}p_1p_2'\right)}_{(iv)_a + (iv)_b}]
\end{align*}

For simplicity we again leave the expectation as implied and show that the remaining expression is second-order. We can add and subtract the same terms to obtain the following:

\begin{align*}
    (i) + (iv)_{a} &= \underbrace{\frac{\pi_a}{\hat{\pi}_a}\sum_{m_1, m_2}(\mu_a - \hat{\mu}_a)(p_{12} - \hat{p}_{12})\hat{p}_1\hat{p}_2'}_{SO_1} + \frac{\pi_a}{\hat{\pi}_a}\sum_{m_1, m_2}(\mu_a - \hat{\mu}_a)\hat{p}_1\hat{p}_2'\hat{p}_{12} + \sum_{m_1, m_2}\hat{\mu}_a\hat{p}_1\hat{p}_2'\hat{p}_{12} \\
    &= SO + \underbrace{\left(\frac{\pi_a - \hat{\pi}_a}{\hat{\pi}_a}\right)\sum_{m_1, m_2}(\mu_a - \hat{\mu}_a)\hat{p}_1\hat{p}_2'\hat{p}_{12}}_{SO_2} + \sum_{m_1, m_2}[(\mu_a - \hat{\mu}_a)\hat{p}_1\hat{p}_2'\hat{p}_{12} + \hat{\mu}_a\hat{p}_1\hat{p}_2'\hat{p}_{12}] \\
    &= SO + \underbrace{\sum_{m_1, m_2}\mu_a\hat{p}_1\hat{p}_2'\hat{p}_{12}}_{(R_1)} \\
    (ii) &= \underbrace{-\frac{\pi_a}{\hat{\pi}_a}\sum_{m_1, m_2} \hat{\mu}_a(p_{12} - \hat{p}_{12})(p_1p_2' - \hat{p}_1\hat{p}_2')}_{SO_3} + \frac{\pi_a}{\hat{\pi}_a}\sum_{m_1,m_2}\hat{\mu}_a(p_{12}-\hat{p}_{12})\hat{p}_1\hat{p}_2' \\
    &= SO + \underbrace{\left(\frac{\pi_a - \hat{\pi}_a}{\hat{\pi}_a}\right)\sum_{m_1,m_2}\hat{\mu}_a(p_{12}-\hat{p}_{12})\hat{p}_1\hat{p}_2'}_{SO_4} + \sum_{m_1,m_2}\hat{\mu}_a(p_{12}-\hat{p}_{12})\hat{p}_1\hat{p}_2' \\
    &= SO + \underbrace{\sum_{m_1,m_2}\hat{\mu}_a(p_{12}-\hat{p}_{12})\hat{p}_1\hat{p}_2'}_{(R_2)} \\
    (iii) + (iv)_b &= \underbrace{\left(\frac{\pi_a - \hat{\pi}_a}{\hat{\pi}_a}\right)\left(\sum_{m_1, m_2}\hat{\mu}_{a}\hat{p}_{12}(p_1 - \hat{p}_1)\hat{p}_2'\right)}_{SO_5} + \underbrace{\left(\frac{\pi_{a'} - \hat{\pi}_{a'}}{\hat{\pi}_{a'}}\right)\sum_{m_1, m_2}\hat{\mu}_a\hat{p}_{12}\hat{p}_1(p_2' - \hat{p}_2')}_{SO_6} \\
    &+ \sum_{m_1,m_2}\hat{\mu}_{a}\hat{p}_{12}(p_1 - \hat{p}_1)\hat{p}_2' + \sum_{m_1, m_2}\hat{\mu}_a\hat{p}_{12}\hat{p}_1(p_2' - \hat{p}_2') - \sum_{m_1,m_2}\mu_ap_{12}p_1p_2'\\
    &= SO + \underbrace{\sum_{m_1,m_2}\hat{\mu}_{a}\hat{p}_{12}(p_1 - \hat{p}_1)\hat{p}_2' + \sum_{m_1, m_2}\hat{\mu}_a\hat{p}_{12}\hat{p}_1(p_2' - \hat{p}_2') - \sum_{m_1,m_2}\mu_ap_{12}p_1p_2'}_{(R_3)}
\end{align*}

It remains to show that the sum of ($R_1$)-($R_3$) are second-order:

\begin{align*}
    R_1 + R_2 + R_3 &= \underbrace{\sum_{m_1, m_2}\hat{\mu}_a(p_{12} - \hat{p}_{12})\hat{p}_1\hat{p}_2'  + \sum_{m_1,m_2}\hat{\mu}_{a}\hat{p}_{12}(p_1 - \hat{p}_1)\hat{p}_2'}_{(a)} + \underbrace{\sum_{m_1,m_2}\hat{\mu}_{a}\hat{p}_{12}(p_2' - \hat{p}_2')\hat{p}_1}_{(b)} \\
    &\underbrace{- \sum_{m_1,m_2}\mu_ap_{12}p_1p_2'+\sum_{m_1,m_2}\mu_a\hat{p}_{12}\hat{p}_1\hat{p}_2'}_{(c)} \\
\end{align*}

Consider $(a)$:

\begin{align*}
    (a) &=  \underbrace{-\sum_{m_1,m_2}\hat{\mu}_a(p_{12} - \hat{p}_{12})(p_1 - \hat{p}_1)\hat{p}_2'}_{SO_7} + \sum_{m_1, m_2}\hat{\mu}_ap_1\hat{p}_2'(p_{12} - \hat{p}_{12}) + \sum_{m_1,m_2}\hat{\mu}_{a}\hat{p}_{12}(p_1 - \hat{p}_1)\hat{p}_2' \\
    &= SO + \underbrace{\sum_{m_1,m_2}\hat{\mu}_a\hat{p}_2'(p_1p_{12} - \hat{p}_1\hat{p}_{12})}_{R_{1a}} \\
    &= SO + R_{1a}
\end{align*}

Next consider

\begin{align*}
    R_{1a} + (b) &= \sum_{m_1,m_2}\hat{\mu}_a\hat{p}_2'(p_1p_{12} - \hat{p}_1\hat{p}_{12}) + \sum_{m_1,m_2}\hat{\mu}_{a}\hat{p}_{12}(p_2' - \hat{p}_2')\hat{p}_1 \\
    &= \underbrace{- \sum_{m_1,m_2}\hat{\mu}_a(p_2' - \hat{p}_2')(p_1p_{12} - \hat{p}_1\hat{p}_{12})}_{SO_8} + \sum_{m_1,m_2}\hat{\mu}_ap_1p_2'p_{12} - \sum_{m_1,m_2}\hat{\mu}_a\hat{p}_1p_2'\hat{p}_{12} \\
    &+ \sum_{m_1,m_2}\hat{\mu}_{a}\hat{p}_1p_2'\hat{p}_{12} - \sum_{m_1,m_2}\hat{\mu}_{a}\hat{p}_{12}\hat{p}_2'\hat{p}_1 \\
    &= SO + \underbrace{\sum_{m_1,m_2}\hat{\mu}_a(p_1p_2'p_{12} - \hat{p}_1\hat{p}_2'\hat{p}_{12})}_{R_{2a}}
\end{align*}

Adding the final remaining terms gives us:

\begin{align*}
    R_{2a} + (c) &= \underbrace{-\sum_{m_1, m_2}(\mu - \hat{\mu})(p_1p_2'p_{12}-\hat{p}_1\hat{p}_2'\hat{p}_{12})}_{SO_9}
\end{align*}

Finally, notice that terms that are second-order in some function times the product of densities (e.g. $S0_8$) are also second-order in the sum of the function times each density individually. For example, consider the second-order term:

\begin{align*}
    &\sum_{m_1,m_2}\hat{\mu}_a(p_2' - \hat{p}_2')(p_1p_{12} - \hat{p}_1\hat{p}_{12}) \\
    &\sum_{m_1,m_2}\hat{\mu}_a(p_2' - \hat{p}_2')(p_1[p_{12} - \hat{p}_{12}] + \hat{p}_{12}[p_1 - \hat{p}_1]) \\
\end{align*}
Similarly, we obtain that

\begin{align*}
    SO_9 &= -\sum_{m_1,m_2}(\mu_a - \hat{\mu}_a)[p_2p_{12}(p_1 - \hat{p}_1) + \hat{p}_1p_{12}(p_2' - \hat{p}_2') + \hat{p}_1\hat{p}_2'(p_{12} - \hat{p}_{12})]
\end{align*}

Collecting all the SO terms and taking expectations gives the result.

\end{document}